\documentclass[12 pt,english]{article}
\usepackage[T1]{fontenc}
\usepackage[utf8]{inputenc}
\usepackage{geometry}
\usepackage[x11names]{xcolor}
\usepackage{babel}
\usepackage{verbatim}
\usepackage{prettyref}
\usepackage{float}
\usepackage{amsthm}
\usepackage{amsmath}
\usepackage{amsfonts}
\usepackage{amssymb}
\usepackage{bbm}
\usepackage{graphicx}
\usepackage{setspace}
\usepackage{esint}
\usepackage{xspace}
\usepackage{bm}

\usepackage{natbib}
\usepackage{multibib}
\newcites{supp}{References}

\usepackage{enumerate}
\usepackage{subcaption}
\usepackage[toc,page]{appendix}

\usepackage{pdflscape}
\usepackage{afterpage}
\usepackage{mathtools}

\DeclareFontFamily{U}{mathc}{}
\DeclareFontShape{U}{mathc}{m}{it}%
{<->s*[1.03] mathc10}{}

\DeclareMathAlphabet{\mathscr}{U}{mathc}{m}{it}

\usepackage[flushleft]{threeparttable}
\usepackage{booktabs,dcolumn,listings}
\usepackage{multirow}

\newcolumntype{d}[1]{D{.}{.}{#1}}

\usepackage[unicode=true,
 bookmarks=true,bookmarksnumbered=false,bookmarksopen=false,
 breaklinks=true,pdfborder={0 0 0},backref=false,colorlinks=true]
 {hyperref}

\newif\ifOutVer\OutVertrue
\OutVerfalse
\newif\ifFullVer\FullVerfalse
\newif\ifIncludeEmpirical\IncludeEmpiricalfalse
\newif\ifIncludeAppendixProof\IncludeAppendixProoffalse

\ifOutVer
\FullVertrue
\IncludeEmpiricalfalse
\newcommand*\mTODO[1]{}

\newcommand{\mQQC}[1]{}
\else
\newcommand*\mTODO[1]{{\scriptsize\color{blue}{\fbox{\textbf{TODO: {#1}}}}}}

\newcommand*{\mQQC}[1]{{\fbox{\scriptsize\color{magenta}\textbf{#1}}}}
\fi

\ifFullVer
\usepackage{tikz}
\usepackage{pgfplots}
\pgfplotsset{compat=newest}
\pgfplotsset{plot coordinates/math parser=false}
\usetikzlibrary{decorations.markings}
\newlength\figureheight
\newlength\figurewidth
\usepackage{chngcntr}
\usepackage{apptools}
\AtAppendix{\counterwithin{Lem}{section}}

\usepackage[figure]{hypcap}
\fi 

\makeatletter

\newtheorem{Lem}{Lemma}
\newtheorem{AppLem}{Lemma}
\numberwithin{AppLem}{section}
\newtheorem{AppThe}{Theorem}
\numberwithin{AppThe}{section}
\newtheorem{The}{Theorem}
\newtheorem{Cor}{Corollary}
\newtheorem{AppCor}{Corollary}
\numberwithin{AppCor}{section}
\theoremstyle{definition}

\newtheorem{Ass}{Assumption}
\newtheorem{Rem}{Remark}
\numberwithin{Rem}{section}
\newtheorem{Ex}{Example}
\newtheorem*{Ex*}{Example}

\numberwithin{equation}{section}
  \theoremstyle{plain}
  \newtheorem*{prop*}{\protect\propositionname}

\hypersetup{
    linkcolor=red,          
    citecolor=blue,        
    filecolor=magenta,      
    urlcolor=blue           
}

\newcommand{\pr}{\mathbb{P}}
\newcommand{\ind}{\mathbbm{1}}
\newcommand{\e}{\mathbb{E}}

\newcommand{\prob}{\overset{p}{\rightarrow}}

\newcommand{\sumin}{\sum_{i=1}^n}

\newcommand{\lsup}{\limsup_{n \rightarrow \infty}}

\renewcommand{\qed}{\hfill Q.E.D.}

\DeclareMathOperator*{\argmin}{\mathrm{arg}\!\min\limits}

\newcommand*\norm[1]{\left\Vert #1\right\Vert}
\newcommand{\abs}[1]{\left\vert #1\right\vert}
\newcommand{\pro}[1]{\pr\left(#1\right)}
\newcommand{\ex}[1]{\e\left[#1\right]}

\newcommand{\supp}[1]{\text{supp}\left(#1\right)}

\newcommand{\dxi}{\frac{1}{2 d_\xi}}
\newcommand{\hatN}{\hat{\mathcal N}}
\newcommand{\hatM}{\hat{\mathcal M}}

\newcommand{\msd}{\mathscr{d}}
\newcommand{\msh}{\mathscr{h}}
\newcommand{\hatC}{\hat{\mathcal C}}

\newcommand*{\QEDA}{\hfill\ensuremath{\blacksquare}}

\usepackage[capitalise]{cleveref}

\emergencystretch=\maxdimen
\makeatother

\providecommand{\propositionname}{Proposition}

\providecommand{\propositionname}{Proposition}

\title{{Identification and Estimation of Network Models\\ with Nonparametric Unobserved Heterogeneity}}
\author{Andrei \textsc{Zeleneev}\thanks{\textsf{a.zeleneev@ucl.ac.uk.} University College London, Department of Economics, and CeMMAP. I am grateful to my advisors Bo Honor{\'e}, Ulrich M{\"u}ller and especially Kirill Evdokimov for their continuous guidance and support. I also thank Michal Koles{\'a}r, Mikkel Plagborg-Møller, Christopher Sims, Mark Watson, and the participants of numerous seminars and conferences for valuable comments and suggestions.}}
\date{\today \\ First version: November, 2019\thanks{The orignal JMP version is available \href{https://www.homepages.ucl.ac.uk/~uctpzel/azeleneev_jmp.pdf}{here}.}} 

\begin{document}

\maketitle


\begin{abstract}
Homophily based on observables is widespread in networks. Therefore, homophily based on unobservables (fixed effects) is also likely to be an important determinant of the interaction outcomes. Failing to properly account for latent homophily (and other complex forms of unobserved heterogeneity) can result in inconsistent estimators and misleading policy implications. To address this concern, we consider a network model with nonparametric unobserved heterogeneity, leaving the role of the fixed effects unspecified. We argue that the interaction outcomes can be used to identify agents with the same values of the fixed effects. The variation in the observed characteristics of such agents allows us to identify the effects of the covariates, while controlling for the fixed effects. Building on these ideas, we construct several estimators of the parameters of interest and characterize their large sample properties. Numerical experiments illustrate the usefulness of the suggested approaches and support the asymptotic theory.

\bigskip

\noindent \textbf{Keywords:} network data, homophily, fixed effects

\end{abstract}
\newpage

\section{Introduction}

Unobserved heterogeneity is pervasive in economics. The importance of accounting for unobserved heterogeneity is well recognized in microeconometrics, in general, as well as in the network context, in particular. For example, since \citet*{Abowd1999}, estimating a linear regression with additive (two-way) fixed effects has become a standard approach to analyzing interaction data. Originally employed to account for workers and firms fixed effects in the wage regression context, this technique has become a standard tool to control for two-sided unobserved heterogeneity and decompose it into agent specific effects.\footnote{For example, recent applications to employer-employee matched data feature \citet*{Card2013,Helpman2017,Song2019} among others. The numerous applications of this approach also include the analysis of students-teachers \citep*{Hanushek2003,Rivkin2005,Rothstein2010}, patients-hospitals \citep*{Finkelstein2016}, firms-banks \citep*{Amiti2018}, and residents-counties matched data \citep*{chetty2018impacts}.} Since the seminal work of \citet*{Anderson2003}, the importance of controlling for exporters and importers fixed effects has also been well acknowledged in the context of the international trade network, including nonlinear settings of \citet*{SantosSilva2006} and \citet*{Helpman2008}. \citet*{Graham2017} stresses the importance of accounting for agents' degree heterogeneity (captured by the additive fixed effects) in network formation models.

While the additive fixed effects framework is commonly employed to control for unobservables in networks, it is not flexible enough to capture more complicated forms of unobserved heterogeneity, which are likely to appear in many settings. This concern can be vividly illustrated in the context of estimation of homophily effects, one of the main focuses of the empirical network analysis.\footnote{The term homophily typically refers to the tendency of individuals to assortatively match based on their characteristics. For example, individuals tend to form social connections based on gender, race, age, education level, and other socioeconomic characteristics. Similarly, countries that share a border, have the same legal system, language or currency, are more likely to have higher trade volumes.} Since homophily (assortative matching) based on observables is widespread in networks (e.g., \citealp*{McPherson2001}), homophily based on unobservables (fixed effects) is also likely to be an important determinant of the interaction outcomes. Since observed and unobserved characteristics (i.e., covariates and fixed effects) are typically correlated, the presence of latent homophily significantly complicates identification of the homophily effects associated with the observables (e.g., \citealp*{Shalizi2011}). Failing to properly account for homophily based on unobservables (and other complex forms of unobserved heterogeneity, in general) is likely to result in inconsistent estimators and misleading policy implications.

To address the concern discussed above, we consider a dyadic network model with a flexible (nonparametric) form of unobserved heterogeneity, where the outcome of the interaction between agents $i$ and $j$ is given by
\begin{align}
  \label{eq: intro model}
  Y_{ij} = F(W_{ij}' \beta_0 + g(\xi_i,\xi_j)) + \varepsilon_{ij}.
\end{align}
Here, $W_{ij}$ is a $p \times 1$ vector of pair-specific observed covariates, $\beta_0 \in \mathbb R^p$ is the parameter of interest, $\xi_i$ and $\xi_j$ are unobserved fixed effects, and $\varepsilon_{ij}$ is an idiosyncratic error. The fixed effects are allowed to interact via the coupling function $g(\cdot,\cdot)$, which is treated as $\emph{unknown}$. Importantly, we do not require $g(\cdot,\cdot)$ to have any particular structure and do not specify the dimension of $\xi$. Finally, $F(\cdot)$ is a known (up to location and scale normalizations) invertible link function. The presence of $F(\cdot)$ ensures that \eqref{eq: intro model} is flexible enough to cover a broad range of the previously studied dyadic network models with unobserved heterogeneity, including important nonlinear specifications such as network formation or Poisson regression models.\footnote{For example, with $F(\cdot)$ equal to the logistic CDF and $g(\xi_i,\xi_j) = \xi_i + \xi_j$, \eqref{eq: intro model} corresponds to the network formation model of \citet*{Graham2017}.}

Being agnostic about the dimensions of the fixed effects and the nature of their interactions, \eqref{eq: intro model} allows for a wide range of forms of unobserved heterogeneity, including homophily based on unobservables.

\begin{Ex}[Nonparametric homophily based on unobservables]
  \label{ex: intro}
	Let $\xi = (\alpha, \nu)' \in \mathbb R^2$ and
	\begin{align*}
		g(\xi_i,\xi_j) = \alpha_i + \alpha_j - \psi(\nu_i,\nu_j),
	\end{align*}
	where $\psi(\cdot,\cdot)$ is some function satisfying $\psi(\nu_i,\nu_j) = 0$ whenever $\nu_i = \nu_j$ and increasing in $\abs{\nu_i - \nu_j}$ (e.g., $\psi(\nu_i,\nu_j) = c \abs{\nu_i - \nu_j}^{\zeta}$ for some $c>0$ and $\zeta \geqslant 1$). Here $\alpha$ represents the standard additive fixed effect, and $\psi(\cdot,\cdot)$ captures latent homophily based on $\nu$: agents with similar values of $\nu$ tend to interact with higher intensity compared to agents distant in terms of $\nu$. Again, since the dimension of $\xi$ is not specified, \eqref{eq: intro model} can also incorporate homophily based on several unobserved characteristics (multivariate $\nu$) in a similar manner. \QEDA
\end{Ex}

We study identification and estimation of \eqref{eq: intro model} under the assumption that we observe a single network of a growing size.\footnote{The large single network asymptotics is standard for the literature focusing on identification and estimation of network models with unobserved heterogeneity. See, for example, \citet*{Graham2017,Dzemski2018,Candelaria2016,Jochmans2018,Toth2017,gao2020nonparametric,gao2023logical}.} First, we focus on a simpler version of \eqref{eq: intro model}
\begin{align}
	\label{eq: intro linear}
	Y_{ij} = W_{ij}' \beta_0 + g(\xi_i,\xi_j) + \varepsilon_{ij}.
\end{align}
We argue that the outcomes of the interactions can be used to identify agents with the same values of the unobserved fixed effects. Specifically, we introduce a certain pseudo-distance~$d_{ij}$ measuring similarity between agents $i$ and $j$ in terms of their latent characteristics. We argue that (i) $d_{ij} = 0$ if and only if $\xi_i = \xi_j$, and (ii) $d_{ij}$ is identified and can be estimated from the data. Consequently, agents with the same values of $\xi$ can be identified based on the pseudo-distance $d_{ij}$. Then, the variation in the observed characteristics of such agents allows us to identify the parameter of interest $\beta_0$ while controlling for the impact of the fixed effects. Importantly, this result is not driven by the linearity or particular functional form of~\eqref{eq: intro linear}: we show that a similar identification argument also applies when $W_{ij}'\beta_0$ is replaced by an unknown (nonparametric) function of observables.

Demonstrating that the introduced pseudo-distance $d_{ij}$ is identified plays a central role in the argument outlined above. First, we show that, if the idiosyncratic errors $\varepsilon_{ij}$ are homoskedastic, $d_{ij}$ can be readily estimated from the data using simple pairwise-difference regressions. Second, we extend this argument to models with general heteroskedasticity by leveraging and advancing recent developments in the matrix estimation/completion literature. Specifically, we demonstrate that the error free outcomes $Y_{ij}^* \coloneqq W_{ij}' \beta_0 + g(\xi_i,\xi_j)$ are identified and can be uniformly (across all pairs of agents) consistently estimated. This is a powerful result allowing us to treat $Y_{ij}^*$ as effectively observed and thus greatly simplifying the analysis. In particular, working with $Y_{ij}^*$ instead of $Y_{ij}$ effectively reduces~\eqref{eq: intro linear} to a model without the error term $\varepsilon_{ij}$, which can be interpreted as an extreme form of homoskedasticity. This, in turn, allows us to establish identification of $\beta_0$ by applying the same argument as in the homoskedastic model.

Building on these ideas, we construct an estimator of $\beta_0$ and characterize its rate of convergence. Following the identification argument, we first estimate the pseudo-distances $\hat d_{ij}$, which are used to find agents with similar latent characteristics. Then, we estimate $\beta_0$ by combining pairwise-difference regressions of $Y_{ik} - Y_{jk}$ on $W_{ik} - W_{jk}$ for all pairs of agents $i$ and $j$ sufficiently similar in terms of $\hat d_{ij}$. Consistency of this estimator crucially relies on the matched agents being different in terms of their observables, resulting in sufficient residual variation in $W_{ik} - W_{jk}$. We characterize the asymptotic behavior of the estimator's bias due to imperfect matching, and provide conditions under which it is consistent. 

In the general heteroskedastic case, construction of $\hat d_{ij}$ also involves preliminary estimation of the error free outcomes $Y_{ij}^*$. To provide estimators of $Y_{ij}^*$ and $d_{ij}$ valid under heteroskedastic errors, we build on and extend the approach of \citet*{Zhang2017} originally employed in the context of nonparametric graphon estimation. Moreover, to formally establish identification of the error free outcomes, we also propose a modified version of \citet{Zhang2017}'s estimator and demonstrate its consistency in the \emph{max} (matrix) norm, meaning that, in a large network, our estimator recovers $Y_{ij}^*$'s for all pairs of agents with a high precision at once. To the best of our knowledge, this result is also new to the statistics literature on graphon estimation, which has previously focused on constructing estimators $\hat Y_{ij}^*$ and deriving their rates of convergence in terms of the mean square error/Frobenius norm \citep{Chatterjee2015,Gao2015,Klopp2017,Zhang2017,li2019nearest}.

Finally, we want to stress that identification of the error free outcomes is a powerful result, which applies to a general class of dyadic network models beyond \eqref{eq: intro model} and can be used as a foundation for establishing new identification results. To the best of our knowledge, this result has not been previously recognized and leveraged in the econometrics literature. In particular, building on it, we demonstrate how the proposed identification and estimation strategies can be naturally extended to cover model \eqref{eq: intro model}, as well as its nonparametric version. We also argue that the pair-specific fixed effects $g_{ij} = g(\xi_i,\xi_j)$ are identified for all pairs of agents $i$ and $j$ and can be (uniformly) consistently estimated. Identification of $g_{ij}$ is an important result in itself since in many applications the fixed effects are the central objects of interest. Moreover, this result is also of special significance when $F$ is nonlinear because identification of $g_{ij}$ allows us to identify important policy relevant quantities such as pair-specific and average partial effects.

\bigskip

\noindent This paper contributes to the literature on econometrics of networks and, more generally, two-way models. The distinctive feature of our model is allowing for flexible nonparametric unobserved heterogeneity: the fixed effects can interact via the unknown coupling function $g(\cdot,\cdot)$. Importantly, we do not require $g(\cdot,\cdot)$ to have any particular structure (other than satisfying a weak smoothness requirement) and do not specify the dimensionality of the fixed effects. This is in contrast to most of the existing approaches, which either explicitly specify the form of $g(\cdot,\cdot)$ or impose additional restrictive assumptions on its shape and smoothness.

Among explicitly specified forms of $g(\cdot,\cdot)$, the additive fixed effects structure ${g(\xi_i,\xi_j) = \xi_i + \xi_j}$ is by far the most popular way of incorporating unobserved heterogeneity in dyadic network models, e.g., see \citet{Graham2017}, \citet{Charbonneau2017}, \citet{Jochmans2018}, \citet{Dzemski2018}, \citet*{Yan2019}, \citet{Candelaria2016}, \citet{gao2020nonparametric} and \citet{Toth2017} among others. While this specification provides a practical way of controlling for degree heterogeneity, it does not account for more complicated forms of unobserved heterogeneity including latent homophily. Recent semiparametric extensions of \citet{Graham2017} and related frameworks feature \citet{gao2023logical} relaxing separability between $W_{ij}' \beta_0$, $\xi_i$, and $\xi_j$. However, \citet{gao2023logical} still require $\xi$ to be scalar and assume that the linking probability is increasing in $\xi_i$ and $\xi_j$, thus ruling out homophily based on unobservables.

The linear factor specification $g(\xi_i,\xi_j) = \xi_i' \xi_j$ is widely used in both network and panel models as a generalization of the additive fixed effect framework.\footnote{Recent studies considering network models with unobserved effects having a linear factor structure include, among others, \citet{chen2021nonlinear,ma2022detecting,zeleneev2025tractable}.} While, as argued in \citet*{chen2021nonlinear}, this specification allows for certain forms of latent homophily, approximating a general function $g(\cdot, \cdot)$ by the linear factor model requires a growing number of factors (e.g., \citealp{fernandez2021low}). Such low-rank approximations, for example, are utilized by \citet{freeman2023linear} and \citet{beyhum2024inference} who consider a panel variation of \eqref{eq: intro linear}. To control the accuracy of the proposed low-rank approximations, both of these papers require $g(\cdot,\cdot)$ to have sufficiently many continuous derivatives and, more importantly, to ensure consistency of their estimator of $\beta_0$, they require regressors $W$ to be ``high-rank''. Both of these requirements are restrictive in the network setting;\footnote{Non-differentiable functions $g(\cdot,\cdot)$, such as the one provided in Example~\ref{ex: intro} for $\zeta = 1$, are a common feature of popular latent homophily models (e.g., \citealp{Hoff2002,handcock2007model}).} see Section~\ref{sssec: ID comparison} for a detailed comparison of the frameworks and approaches to identification. Allowing for both low-rank regressors $W$ and possibly non-differentiable $g(\cdot,\cdot)$ is an important and unique combination of features of our setting differentiating this paper from the rest of the literature.\footnote{Low-rank regressors are typically ruled out to ensure identification of $\beta_0$ even in linear factor models; see, for example, \citet{bai2009panel,moon2015linear,armstrong2022robust}.}

Employing clustering methods to partition units into groups with similar unobserved characteristics is another method commonly employed to control for interactive (time-varying) unobserved heterogeneity. This approach was proposed and applied to likelihood panel models in the seminal work by \citet{bonhomme2022discretizing} and then extended to semiparametric panel regressions by \citet{beyhum2024inference}. While the specific clustering methods (involving the observed covariates as well) employed in these papers proved to be instrumental in panel settings, they would not allow one to identify and estimate $\beta_0$ in network models when the regressors take the typical form of $W_{ij} = w(X_i,X_j)$, where $X_i$ and $X_j$ are observed characteristics of agents $i$ and $j$. In this case, the previously employed methods would cluster units into groups with similar values of both $\xi$ and $X$ preventing one from disentangling the effects of observables and unobservables and thus from identifying $\beta_0$ as well.\footnote{Moreover, the clustering methods employed in \citet{bonhomme2022discretizing} and \citet{beyhum2024inference} are based on individual-specific moments. Such clustering methods might fail to meaningfully group units in network settings, in which informative individual-specific moments might not exist at all; see Section~\ref{sssec: ID comparison} and specifically Footnote~\ref{fn: clustering failure} for details.} Importantly, unlike the previously employed clustering methods, our approach can be used to find units with similar values of $\xi$ and yet different values of $X$ allowing us to identify $\beta_0$. Thus, we also complement the methodology of \citet{bonhomme2022discretizing} by providing a new grouping approach which can be instrumental in both network and panel models with low-rank regressors.

When $\xi$ is discrete, the considered model \eqref{eq: intro model} belongs to the class of stochastic block models (SBM) with covariates (e.g., \citealp{mele2023spectral,ma2022detecting}). While our framework and estimation approach are general enough to cover this special setup, in this paper, we focus on the case when unobserved heterogeneity is continuous.\footnote{In fact, when $\xi$ is discrete, agents can be correctly classified into groups with the same values of $\xi$ based on the already introduced pseudo-distance $\hat d_{ij}$, which greatly simplifies the asymptotic analysis; we will discuss this in more detail shortly after we introduce Assumption~\ref{ass: xi dist} in Section~\ref{sec: LST}.} In a recent study, \citet{kitamura2024estimating} study a nonseparable nonparametric variation of SBM with covariates. Since their approach crucially relies on the discreteness of $\xi$ whereas ours exploits separability between the observables and unobservables as in \eqref{eq: intro model}, the frameworks considered in this paper and by \citet{kitamura2024estimating} are non-nested and complementary.

Another strand of the literature emphasizes the importance of using network data to control for endogeneity in peer effects and other related models (e.g., \citealp*{goldsmith2013social,johnsson2021estimation,auerbach2022identification,starck2025improving}). In these papers, the agents' latent characteristics $\xi$ affect both the individual outcomes of interest as well as the network formation process. To tackle this problem, \citet{auerbach2022identification} and \citet{johnsson2021estimation} propose using certain network statistics to identify agents with the same values of $\xi$ allowing the authors to control for the unobservables in the other (cross-sectional) regression of interest (e.g., in a peer effects model). The important difference between their and our work is that we use the \emph{same} network data both to control for unobserved heterogeneity and to identify the parameters of interest, allowing us to disentangle the effects of observables and unobservables within a single network model.\footnote{For example, similarly to the clustering methods of \citet{bonhomme2022discretizing} and \citet{beyhum2024inference}, a direct application of \citet{auerbach2022identification}'s approach to the observed network only allows one to identify agents with the same values of both $\xi$ and $X$ once again precluding identification of $\beta_0$.}

Finally, we emphasize that the considered model \eqref{eq: intro model} does not incorporate interaction externalities. Specifically, we assume that conditional on the agents' observed and unobserved characteristics, the interaction outcomes are independent. This assumption is plausible when the interactions are primarily bilateral. For excellent recent reviews of econometrics of networks, with and without strategic interactions, we refer the reader to \citet{graham2020econometric,GRAHAM2020111,de2020econometric}.

\bigskip

{The rest of the paper is organized as follows. In Section~\ref{sec: ID}, we formally introduce the framework and provide (heuristic) identification arguments for the semiparametric regression model~\eqref{eq: intro linear}. Section \ref{sec: estimation} turns these ideas into estimators of the parameters of interest. In Section \ref{sec: LST}, we establish consistency of the proposed estimators and derive their rates of convergence. In Section \ref{sec: ext}, we generalize the proposed identification argument to cover more general settings including the nonlinear model \eqref{eq: intro model} as well as its nonparametric analogue. Section \ref{sec: numerical} provides numerical and empirical illustrations.  A supplementary appendix contains all proofs, and additional discussions and illustrations. 

\section{Identification of the Semiparametric Model}
\label{sec: ID}
In this section, we introduce the framework and provide a conceptual discussion of identification of the semiparametric regression model. This discussion is supposed to illustrate the anatomy of the model and to highlight the main insights of our identification strategy, and is deliberately not formalized here. In Section~\ref{sec: estimation}, we will turn these ideas into practical estimators. Identification then is formally demonstrated in Section~\ref{sec: LST}, where establish consistency of our estimators and provide their rates of convergence.

\subsection{The model}
We consider a network consisting of $n$ agents. Each agent $i$ is endowed with characteristics $Z_i = (X_i,\xi_i)$, where $X_i \in \mathcal X$ is observed by the econometrician, while $\xi_i \in \mathcal E$ is not. We consider the following semiparametric regression model, where the (scalar) outcome of the interaction between agents $i$ and $j$ is given by
\begin{align}
  \label{eq: Y_ij def}
  Y_{ij} = w(X_i,X_j)'\beta_0 + g(\xi_i,\xi_j) + \varepsilon_{ij}, \quad i \neq j.
\end{align}
Here, $w: \mathcal X \times \mathcal X \rightarrow \mathbb R^p$ is a known function, which transforms the observed characteristics of agents $i$ and $j$ into a pair-specific vector of covariates $W_{ij} \coloneqq w(X_i,X_j)$, $\beta_{0} \in \mathbb R^p$ is the parameter of interest, and $\varepsilon_{ij}$ is an unobserved idiosyncratic error. Note that unlike $w(\cdot,\cdot)$, the coupling function $g: \mathcal E \times \mathcal E \rightarrow \mathbb R$ is \emph{unknown}, and the dimension of the fixed effect $\xi_i \in \mathcal E$ is \emph{not specified}. For simplicity of exposition, we will suppose that $\xi_i \in \mathbb R^{d_\xi}$ even though, in principle, the same insights apply when $\mathcal E$ is a general metric space.

For concreteness, we will also focus on an undirected model with $Y_{i j} = Y_{j i}$, so $w(\cdot,\cdot)$ and $g(\cdot,\cdot)$ are symmetric functions, and $\varepsilon_{ij} = \varepsilon_{ji}$. The methodology presented in this paper straightforwardly extends to directed networks and general two-way settings including panel models; see Section \ref{ssec: two-way} for a more detailed discussion.

The following assumption formalizes the sampling process.

\begin{Ass}
  \label{ass: DGP}
  \noindent
  \begin{enumerate}[(i)]
    \item \label{item: iid} $\{Z_i\}_{i=1}^n$ are i.i.d.;
    \item \label{item: errors} conditional on $\{Z_i\}_{i=1}^n$, the idiosyncratic errors $\{\varepsilon_{ij}\}_{i<j}$ are independent draws from $P_{\varepsilon_{ij}|Z_i,Z_j}$ with $\ex{\varepsilon_{ij}|Z_i,Z_j} = 0$, and $\varepsilon_{ij} = \varepsilon_{ji}$;
    \item \label{item: observables} the econometrician observes $\{X_i\}_{i=1}^n$ and $\{Y_{ij}\}_{i \neq j}$ determined by \eqref{eq: Y_ij def}.
  \end{enumerate}  
\end{Ass}

Assumption \ref{ass: DGP} is standard for the networks literature.\footnote{See, e.g., \citet{Graham2017,gao2020nonparametric,gao2023logical}.} The sampling process could be thought of as follows. First, the characteristics of agents $\{Z_i\}_{i=1}^n$ are independently drawn from some population distribution. Then, conditional on the drawn characteristics, the idiosyncratic errors $\{\varepsilon_{ij}\}_{i < j}$ are independently drawn from the conditional distributions, which potentially depend on the characteristics of the corresponding agents $Z_i$ and $Z_j$.

\begin{Rem}
  For simplicity of exposition, we assume that we observe $Y_{ij}$ for all pairs of agents $i$ and $j$. In Section \ref{ssec: missing}, we will discuss how to incorporate missing outcomes and sparse networks into the considered framework.
\end{Rem}

\subsection{Identification of $\beta_0$: main insights}
\label{ssec: beta ID}

We study identification and estimation of $\beta_0$ under the large network asymptotics, which takes $n \rightarrow \infty$. The identification argument is based on the following observation. Suppose that we can identify two agents $i$ and $j$ with the same unobserved characteristics, i.e., with $\xi_i = \xi_j$. Then, for any third agent $k$, the difference between $Y_{ik}$ and $Y_{jk}$ is given by
\begin{align}
  \label{eq: delta Y delta W reg def}
  Y_{ik} - Y_{jk} = \underbrace{(w(X_i,X_k) - w(X_j,X_k))'}_{(W_{ik} - W_{jk})'} \beta_0 + \varepsilon_{ik} - \varepsilon_{jk}.
\end{align}
The conditional mean independence of the regression errors now guarantees that $\beta_0$ can be identified from the regression of $Y_{ik} - Y_{jk}$ on $W_{ik} - W_{jk}$, provided that we have ``enough'' variation in $W_{ik} - W_{jk}$. Formally, we have
\begin{align}
  \label{eq: beta_0 ID}
  \beta_0 = \ex{(W_{ik} - W_{jk})(W_{ik} - W_{jk})'|Z_i,Z_j}^{-1} \ex{(W_{ik} - W_{jk})(Y_{ik} - Y_{jk})|Z_i,Z_j},
\end{align}
provided that $\ex{(W_{ik} - W_{jk})(W_{ik} - W_{jk})'|Z_i,Z_j}$ is invertible. Since agents $i$ and $j$ are treated as fixed, the expectations are conditional on their characteristics $Z_i$ and $Z_j$. At the same time, $Z_k$, the characteristics of agent $k$, and the idiosyncratic errors $\varepsilon_{ik}$ and $\varepsilon_{jk}$ are treated as random and integrated over. Note that the invertibility requirement insists on $X_i$ and $X_j$, the observed characteristics of agents $i$ and $j$, to be ``sufficiently different''. Indeed, if not only $\xi_i = \xi_j$ but also $X_i = X_j$, this condition is clearly violated since $W_{ik} - W_{jk} = 0$ for any agent $k$: in this case, $\beta_{0}$ cannot be identified from the regression \eqref{eq: delta Y delta W reg def}.

Hence, the problem of identification of $\beta_0$ can be reduced to the problem of identification of agents $i$ and $j$ with the same values of the unobserved fixed effects ($\xi_i = \xi_j$) but with ``sufficiently different'' values of $X_i$ and $X_j$.

Let $Y_{ij}^*$ be the error free part of $Y_{ij}$, i.e.,
\begin{align*}
  Y_{ij}^* \coloneqq \ex{Y_{ij}|Z_i,Z_j} = w(X_i,X_j)' \beta_0 + g(\xi_i,\xi_j).
\end{align*}
Consider the following (squared) pseudo-distance between agents $i$ and $j$
\begin{align}
  \label{eq: d_ij star def}
  d_{ij}^2 &\coloneqq \min_{\beta \in \mathcal B} \ex{(Y_{ik}^* - Y_{jk}^* - (W_{ik} - W_{jk})' \beta)^2|Z_i,Z_j} \\
  &=\min_{\beta \in \mathcal B} \e \bigg[ (\underbrace{g(\xi_i,\xi_k) - g(\xi_j,\xi_k)}_{=0,\text{ when } \xi_i = \xi_j} - (W_{ik} - W_{jk})'(\beta - \beta_0))^2|Z_i,Z_j \bigg], \nonumber
\end{align}
where $\mathcal B \ni \beta_0$ is some parameter space. Here, the expectation is conditional on the characteristics of agents $i$ and $j$ and is taken over $Z_k$. Clearly, $d_{ij}^2 = 0$ when $\xi_i = \xi_j$: in this case, the minimum is achieved at $\beta = \beta_0$. Moreover, under a suitable (rank) condition (which we will formally discuss in Section \ref{ssec: beta rate}), $d_{ij}^2 = 0$ also necessarily implies that $\xi_i = \xi_j$. Consequently, if $d_{ij}^2$ were available, agents with the same values of $\xi$ could be identified based on this pseudo-distance.

However, the expectation \eqref{eq: d_ij star def} cannot be directly identified, since the error free outcomes $Y_{ij}^*$ are not observed. In the following Sections~\ref{ssec: homo id} and \ref{ssec: hetero ID}, we will argue that the pseudo-distances $d_{ij}^2$ (or their close analogues) are identified for all pairs of agents $i$ and $j$ and, hence, can be used to identify agents with the same values of $\xi$ (and different values of $X$).

\begin{Rem}
  Notice that none of the arguments provided in this section changes if $\xi_i = \xi_j$ is understood as equivalence of the associated functions $g(\xi_i, \cdot)$ and $g(\xi_j, \cdot)$ (in terms of the $L^2$ distance associated with the distribution $P_\xi$). Thus, for simplicity of exposition, we will implicitly assume that different values of $\xi$ are associated with different functions $g(\xi,\cdot)$. This assumption is not restrictive since we do not normalize the distribution of $\xi$, e.g., we do not impose $\xi \sim U [0,1]$. In particular, the absence of unobserved heterogeneity is allowed: in this case, we have $\xi = \xi_0$ for all agents for some fixed $\xi_0$.
\end{Rem}

\subsubsection{Comparison with the existing approaches to identification}
\label{sssec: ID comparison}

Before we proceed with identification of $d_{ij}^2$, we would like to compare our identification argument with the existing results for (interactive fixed effects) panel and network models and highlight some fundamental differences.

One of the most common ways to ensure identification of $\beta_0$ in panel regression models with interactive fixed effects is to assume that covariates $W_{it}$ have sufficient independent variation over $i$ and $t$. This assumption is also commonly referred to as the ``high-rank'' regressors condition. It plays a crucial role in establishing consistency and deriving asymptotic properties of various estimators of $\beta_0$ in models with unobserved heterogeneity following a linear factor form (e.g., \citealp{bai2009panel,moon2015linear,armstrong2022robust}) and in settings with more general nonparametric structures of unobserved heterogeneity similar to the one studied in this work \citep{bonhomme2022discretizing,freeman2023linear,beyhum2024inference}.

In the context of this paper, the ``high-rank'' regressors condition translates into assuming that $W_{ij} = w(X_i,X_j) + \eta_{ij}$, where $\eta_{ij}$ exhibits sufficient independent variation across dyads, e.g., $\eta_{ij}$'s are (conditionally) independent. Identification of $\beta_0$ then can be ensured using the ``exogenous'' (``high-rank'') variation in $\eta_{ij}$ orthogonal to the unobserved (low-dimensional or even low-rank) component $g(\xi_i,\xi_j)$. While this strategy is commonly employed in panel models with high-rank regressors, it cannot be applied in network models with $W_{ij}$ typically taking the (low-dimensional) form of $W_{ij} = w(X_i,X_j)$ \emph{without} the additional exogenous ``high-rank'' component $\eta_{ij}$.

The absence of the $\eta_{ij}$ component in the model substantially complicates both identification of $\beta_0$ as well as the asymptotic analysis of the subsequently constructed analogue estimators. To see this, notice that if $W_{ij} = w(X_i,X_j) + \eta_{ij}$, it would suffice to match agents with the same values of \emph{both} $X$ and $\xi$, i.e., with $X_i = X_j$ and $\xi_i = \xi_j$, because in this case $\beta_0$ can still be identified as in \eqref{eq: beta_0 ID} thanks to the remaining variation in $W_{ik} - W_{jk} = \eta_{ik} - \eta_{jk}$. This matching strategy, for example, is employed by \citet{beyhum2024inference} who generalized the clustering approach of \citet{bonhomme2022discretizing} to semiparametric panel regressions. However, when $W_{ij} = w(X_i,X_j)$ and the $\eta_{ij}$-component is absent, such matching approaches fail to identify $\beta_0$ due to the lack of the remaining variation in $W_{ik} - W_{jk} = 0$.

In this paper, we overcome this difficulty and establish identification of $\beta_{0}$ by demonstrating how to find agents with the same values of $\xi$ but \emph{different} values of $X$ allowing us to use the variation in $W_{ik} - W_{ij}$ \emph{without} relying on the ``high-rank'' exogenous variation in $\eta_{ij}$ like the panel literature does. The latter is the most important and fundamental difference between our work and the recent panel papers by \citet{freeman2023linear} and \citet{beyhum2024inference}. Likewise, our identification approach is new and structurally different from the ones previously employed in the literature, and while we focus on network models in this paper, the proposed methodology can also be instrumental in establishing new identification results in panels with low-rank regressors.

Finally, we want to stress that the proposed approach to finding units with the same values of $\xi$ and different values of $X$ is new and structurally different from the previously employed approaches used to match agents with similar latent characteristics in panel and network models. In particular, our framework features two types of individual specific-characteristics, observed $X_i$ and unobserved $\xi_i$, which are needed to be treated differently: we want to identify the ``causal'' effect of $X_i$ while keeping $\xi_i$ fixed. At the same time, the existing matching methods, including k-means clustering employed in \citet{bonhomme2022discretizing} and \citet{beyhum2024inference}, and the similarity based approach pioneered by \citet{Zhang2017} and \citet{auerbach2022identification}, either drop $X_i$ altogether or effectively match on $Z_i = (X_i,\xi_i)$, which, as explained above, does not allow one to identify $\beta_0$.\footnote{Also, notice that, even if one abstracts from having or explicitly controlling for observed $X_i$ in the studied network model, clustering based on individual-specific moments $h_i$ originally proposed by \citet{bonhomme2022discretizing} and subsequently employed in \citet{beyhum2024inference} might fail to match agents with similar values of $\xi$. Specifically, their approach requires availability of (consistently estimable) $h_i = \varphi(\xi_i)$ informative about $\xi_i$ but such individual-specific moments might not exist at all in the network setting. To see this, consider a simplified version of the studied model $Y_{ij} = g(\xi_i,\xi_j) + \varepsilon_{ij}$ without covariates, where (i) $\xi_i$ is uniformly distributed on some sphere in $\mathbb R^{d_\xi}$, (ii) $g(\xi_i,\xi_j) = \mathscr g (\norm{\xi_i - \xi_j})$ is effectively determined by $\norm{\xi_i - \xi_j}$ exclusively and captures homophily based on $\xi$, (iii) $\varepsilon_{ij}$'s are iid and independent from all the $\xi$'s. It is clear that, given the spherical symmetry of this model, any individual-specific moments $h_i$ are completely uninformative about $\xi_i$, rendering the described clustering method inappropriate.\label{fn: clustering failure}}

\subsection{Identification under conditional homoskedasticity}
\label{ssec: homo id}
In this section, we consider the case when the regression errors are homoskedastic, i.e., when 
\begin{align}
  \label{eq: homosked eps}
  \ex{\varepsilon_{ij}^2|Z_i,Z_j} = \sigma^2 \quad {\text{a.s.}}
\end{align}

For a pair of agents $i$ and $j$, consider the following conditional expectation
\begin{align}
  \label{eq: q_ij def}
  q_{ij}^2 \coloneqq \min_{\beta \in \mathcal B} \ex{(Y_{ik} - Y_{jk} - (W_{ik} - W_{jk})' \beta)^2|Z_i,Z_j}.
\end{align}
Essentially, $q_{ij}^2$ is a feasible analogue of $d_{ij}^2$ with $Y_{ik}$ and $Y_{jk}$ replacing $Y_{ik}^*$ and $Y_{jk}^*$. Importantly, unlike $d_{ij}^2$, $q_{ij}^2$ is immediately identified and can be estimated by
\begin{align}
  \label{eq: hat q_ij def}
  \hat q_{ij}^2 \coloneqq \min_{\beta \in \mathcal B} \frac{1}{n-2} \sum_{k \neq i,j} (Y_{ik} - Y_{jk} - (W_{ik} - W_{jk})' \beta)^2.
\end{align}
Notice that since $Y_{ik} = Y_{ik}^* + \varepsilon_{ik}$ and $Y_{jk} = Y_{jk}^* + \varepsilon_{jk}$,
\begin{align}  
  q_{ij}^2 &= \min_{\beta \in \mathcal B} \ex{(Y_{ik}^* - Y_{jk}^* - (W_{ik} - W_{jk})' \beta + \varepsilon_{ik} - \varepsilon_{jk})^2|Z_i,Z_j} \nonumber \\
  &= \min_{\beta \in \mathcal B} \ex{(Y_{ik}^* - Y_{jk}^* - (W_{ik} - W_{jk})' \beta)^2|Z_i,Z_j} + \ex{\varepsilon_{ik}^2 + \varepsilon_{jk}^2|Z_i,Z_j} \nonumber \\
  &= d_{ij}^2 + \ex{\varepsilon_{ik}^2|Z_i} + \ex{\varepsilon_{jk}^2|Z_j}, \label{eq: d decomposition}
\end{align}
where the second and the third equalities follow from Assumption \ref{ass: DGP}\eqref{item: errors}. Hence, when the errors are homoskedastic and \eqref{eq: homosked eps} holds, we have
\begin{align}
  \label{eq: homo d}
  q_{ij}^2 = d_{ij}^2 + 2 \sigma^2.
\end{align}
Thus, for every pair of agents $i$ and $j$, $q_{ij}^2$ differs from $d_{ij}^2$ by a constant term $2 \sigma^2$.

Imagine that for a fixed agent $i$, we are looking for a match $j$ with the same value of $\xi$. As discussed in Section \ref{ssec: beta ID}, such an agent can be identified by minimizing $d_{ij}^2$ over potential matches. Then, \eqref{eq: homo d} ensures that, in the homoskedastic setting, such an agent can also be identified by minimizing $q_{ij}^2$. Hence, agents with the same values of $\xi$ (and different values of $X$) can be identified based on $q_{ij}^2$, which can be directly estimated.

\begin{Rem}
The identification argument provided for the homoskedastic model can be naturally extended to allow for $\e[\varepsilon_{ij}^2|Z_i,Z_j] = \e[\varepsilon_{ij}^2|X_i,X_j]$. Indeed, if the skedastic function does not depend on the unobserved characteristics, conditioning on some fixed value $X_j = x$ makes the last term $\e[\varepsilon_{jk}^2|X_j=x,\xi_j] = \e[\varepsilon_{jk}^2|X_j = x]$ in \eqref{eq: d decomposition} constant again. In this case, like in the homoskedastic model, $q_{ij}^2$ is minimized whenever $d_{ij}^2$ is, which allows us to identify agents with the same values of $\xi$. 
\end{Rem}

\subsection{Identification under general heteroskedasticity}
\label{ssec: hetero ID}

Under general heteroskedasticity of the errors, the identification strategy based on $q_{ij}^2$ no longer guarantees finding agents with the same values of $\xi$. Consider the same process of finding an appropriate match $j$ for a fixed agent $i$. As shown in \eqref{eq: d decomposition}, $q_{ij}^2$ can be represented as a sum of three components. The first term $d_{ij}^2$, which we will call the signal, identifies agents with the same values of $\xi$. The second term $\ex{\varepsilon_{ik}^2|X_i,\xi_i}$ does not depend on $j$. However, under general heteroskedasticity, the third term $\e[\varepsilon_{jk}^2|X_j,\xi_j]$ depends on $\xi_j$ and distorts the signal. Hence, the identification argument provided in Section \ref{ssec: homo id} is no longer valid in this case.

In this section, we will address this issue and extend the arguments of Sections \ref{ssec: beta ID} and \ref{ssec: homo id} to a model with general heteroskedasticity. Specifically, we will (heuristically) argue that the error free outcomes $Y_{ij}^*$ are identified for all pairs of agents $i$ and $j$. As a result, the pseudo-distance $d_{ij}^2$ introduced in \eqref{eq: d_ij star def} is also identified and can be directly employed to find agents with the same values of $\xi$ (and different values of $X$).

\subsubsection{Identification of $Y_{ij}^*$}
\label{ssec: Y star ID}
With $Y_{ij}^*$ and $Y_{ij} = Y_{ij}^* + \varepsilon_{ij}$ collected as entries of $n \times n$ matrices $Y^*$ and $Y$ (with diagonal elements of $Y$ missing), the problem of identification and estimation of $Y^*$ based on its noisy proxy $Y$ can be interpreted as a particular variation of the classic matrix estimation/completion problem. Specifically, it turns out that the considered network model \eqref{eq: Y_ij def} is an example of the latent space model (see, for example, \citet*{Chatterjee2015} and the references therein). In the standard formulation of the latent space model, the entries of (symmetric) matrix $Y$ have the form of
\begin{align}
  \label{eq: LVM}
  Y_{ij} = f(Z_i,Z_j) + \varepsilon_{ij},
\end{align}
where $f$ is some (unknown) symmetric function, $Z_1, \ldots, Z_n$ are some latent variables associated with the corresponding rows and columns of $Y$, and the errors $\{\varepsilon_{ij}\}_{i < j}$ are assumed to be (conditionally) independent.\footnote{As noted, for example, in \citet*{bickel2009nonparametric} and \citet*{bickel2011method}, the latent space model is natural in exchangeable settings due to the Aldous-Hoover theorem \citep*{aldous1981representations,hoover1979relations}. For a detailed discussion of this result and other representation theorems for exchangeable random arrays, see, for example, \citet*{kallenberg2005probabilistic} and \citet*{orbanz2015bayesian}.}\footnote{The problem of estimation of $Y_{ij}^* = f (Z_i,Z_j)$ from $Y$ is also known as nonparametric regression without knowing the design \citep*{Gao2015} or blind regression \citep*{li2019nearest}. If $Y_{ij}$ is binary, $Y$ can be interpreted as the adjacency matrix of a random graph. In this case, the function $f(\cdot,\cdot)$ is called a graphon, and this problem is commonly referred to as graphon estimation (see, for example, \citealp*{Gao2015,Klopp2017,Zhang2017}).} Notice, that the studied model fits this general formulation, even though, in our setting, we observe $X_i$ (a subvector of $Z_i$).

It turns out that the particular structure of the latent space model \eqref{eq: LVM} allows one to construct a consistent estimator of $Y^*$ based on a single measurement $Y$. For example, \citet*{Chatterjee2015,Gao2015,Klopp2017,Zhang2017} construct such estimators and establish their consistency in terms of the mean square error (MSE).

In particular, we build on the estimation strategy of \citet*{Zhang2017} to argue that the error free outcomes $Y_{ij}^*$ are identified for all pairs of agents $i$ and $j$. The proposed identification strategy consists of two main steps. First, we argue that we can identify agents with the same values of (both) $X$ and $\xi$. Then, building on this result, we demonstrate how $Y_{ij}^*$ can be constructively identified.

\bigskip

\noindent
\textbf{Step 1: Identification of agents with the same values of $X$ and $\xi$}\\
Consider a subpopulation of agents with a fixed value of $X = x$ exclusively. Let $g_x (\xi_i,\xi_j) \coloneqq w(x,x)' \beta_0 + g(\xi_i,\xi_j)$ and $P_{\xi|X}(\xi|x)$ denote the conditional distribution of $\xi$ given $X = x$. In this subpopulation, consider the following (squared) pseudo-distance between agents $i$ and $j$
\begin{align*}
  d_{\infty}^2 (i,j;x) & \coloneqq \sup_{\xi_k \in \supp{\xi|X=x}} \abs{\ex{(Y_{i \ell} - Y_{j \ell}) Y_{k \ell} | \xi_i, \xi_{j}, \xi_k, X = x}} \\
          &= \sup_{\xi_k \in \supp{\xi|X=x}} \abs{\int (g_x(\xi_i,\xi_\ell) - g_x(\xi_{j},\xi_\ell)) g_x(\xi_k, \xi_\ell) dP_{\xi|X}(\xi_\ell;x)},
\end{align*}
where the second equality uses Assumption \ref{ass: DGP}\eqref{item: errors}.

The finite sample analogue of $d_\infty^2$ was originally introduced in \citet*{Zhang2017} in the context of nonparametric graphon estimation. It is also closely related to the so-called similarity distance inducing a weak topology on graphons (e.g., see, \citet{lovasz2012large} and the references therein).

First, notice that (under weak smoothness conditions) $d_{\infty}^2(i,j;x)$ is directly identified and, if a sample of $n_x$ agents with $X = x$ is available, it can be estimated by
\begin{align*}
  \hat d_{\infty}^2(i,j;x) \coloneqq \max_{k \neq i,j} \Big\vert{(n_x - 3) \sum_{\ell \neq i,j,k} (Y_{i \ell} - Y_{j \ell}) Y_{k \ell}}\Big\vert.
\end{align*}
Second, note that $d_\infty^2 (i,j;x) = 0$ implies that
\begin{align}
  \int \left(g_x(\xi_i,\xi_\ell) - g_x(\xi_j, \xi_\ell) \right) g_x(\xi_k,\xi_\ell) dP_{\xi|X}(\xi_\ell;x) = 0 \label{eq: d_int to l2 arguement}
\end{align}
for almost all $\xi_k$. Evaluating \eqref{eq: d_int to l2 arguement} at $\xi_k = \xi_i$ and $\xi_k = \xi_j$ and subtracting the latter from the former, we conclude
\begin{align*}
   \int \left(g_x(\xi_i,\xi_\ell) - g_x(\xi_{j}, \xi_\ell) \right)^{2} dP_{\xi|X}(\xi_\ell;x) = \int \left(g(\xi_i,\xi_\ell) - g(\xi_{j}, \xi_\ell) \right)^{2} dP_{\xi|X}(\xi_\ell;x) = 0.
\end{align*}
Thus, $d_{\infty}^2(i,j;x) = 0$ implies that $g(\xi_i,\cdot)$ and $g(\xi_j,\cdot)$ are the same (in terms of the $L^2$ distance associated with the conditional distribution of $\xi|X=x$).\footnote{Similar arguments are also provided in \citet*{lovasz2012large} and \citet*{auerbach2022identification}.} Finally, since the equivalence of $g(\xi_i, \cdot)$ and $g(\xi_j, \cdot)$ is understood as $\xi_i = \xi_j$, we can use $d_{\infty}^2(i,j;x)$ to identify agents with the same values of both $X$ and $\xi$.

\begin{Rem}
The idea of using different versions of $d^2_{\infty}$ for finding agents with similar characteristics is not new. For example, \citet{Zhang2017} originally employed it for nonparametric graphon estimation, and \citet{auerbach2022identification} used it to control for unobservables in a partially linear model using network data.\footnote{Other recent applications of this methodology also include construction of network resampling methods \citep{nowakowicz2024nonparametric}, estimation of grouped fixed effects models \citep{mugnier2025simple}, and treatment effect estimation in panels \citep{athey2025identificationaveragetreatmenteffects,deaner2025inferring}.} Once again, we want to stress that this strategy alone only allows one to identify agents with the same values of $\emph{both}$ $X$ and $\xi$. Such matches, however, cannot be used to disentangle the effects of observables and unobservables and to identify $\beta_0$, which is the primary focus of this paper.
\end{Rem}

\noindent
\textbf{Step 2: Identification of $Y_{ij}^*$}\\
Now, being able to identify agents with the same values of $X$ and $\xi$, we can also identify the error free outcome $Y_{ij}^* = w(X_i,X_j)' \beta_0 + g(\xi_i,\xi_j)$ for any pair of agents $i$ and $j$. Specifically, for a fixed agent $i$, we can construct a collection of agents with $X = X_i$ and $\xi = \xi_i$, i.e., $\mathcal N_i \coloneqq \{i': X_{i'} = X_i, \xi_{i'} = \xi_i \}$. Similarly, we construct $\mathcal N_j \coloneqq \{j': X_{j'} = X_j, \xi_{j'} = \xi_j\}$. Then,
\begin{align}
  \frac{1}{n_i n_j} \sum_{i' \in \mathcal N_i} \sum_{j' \in \mathcal N_j} Y_{i' j'} 
  &= \frac{1}{n_i n_j} \sum_{i' \in \mathcal N_i} \sum_{j' \in \mathcal N_j} (w(X_{i},X_{j}) + g(\xi_{i},\xi_{j}) + \varepsilon_{i' j'} ) \nonumber \\
  &= Y_{ij}^* + \frac{1}{n_i n_j} \sum_{i' \in \mathcal N_i} \sum_{j' \in \mathcal N_j} \varepsilon_{i' j'} \quad \prob \quad Y_{ij}^*, \quad \quad n_i, n_j \rightarrow \infty, \label{eq: Y star ID}
\end{align}
where $n_i$ and $n_j$ denote the number of elements in $\mathcal N_i$ and $\mathcal N_j$, respectively. Since in the population we can construct arbitrarily large $\mathcal N_i$ and $\mathcal N_j$, \eqref{eq: Y star ID} implies that $Y_{ij}^*$ is identified.

\begin{Rem}
  Although the identification argument provided above is heuristic, it captures the main insights and will be formalized later. Specifically, in Section \ref{ssec: uniform estimation of Y}, we will construct a particular estimator $\tilde Y_{ij}^*$ and establish its uniform consistency, i.e., we will demonstrate that $\max_{i, j} \vert{\tilde Y_{ij}^* - Y_{ij}^*}\vert = o_p(1)$. This formally proves that $Y_{ij}^*$ is identified for all $i$ and $j$.
\end{Rem}

Identifiability of $Y_{ij}^*$ is a strong result, which, to the best of our knowledge, is new to the econometrics literature on identification of network and, more generally, two-way models. Importantly, it is not due to the specific parametric form or additive separability (in $X$ and $\xi$) of the model \eqref{eq: Y_ij def}. In fact, by essentially the same argument, the error free outcomes $Y_{ij}^* = f(X_i,\xi_i,X_j,\xi_j)$ are also identified in a fully non-separable nonparametric model
\begin{align*}
  Y_{ij} = f(X_i,\xi_i,X_j,\xi_j) + \varepsilon_{ij}, \quad  \ex{\varepsilon_{ij}|X_i,\xi_i,X_j,\xi_j} = 0.
\end{align*}

The established result implies that for studying identification of such models, the error free outcome $Y_{ij}^*$ can be treated as directly observed. Removing the error part greatly simplifies the analysis and provides a powerful foundation for establishing further identification results in more complicated settings; see Section~\ref{sec: ext} for details and examples.

For example, in the particular context of the model \eqref{eq: Y_ij def}, identifiability of $Y_{ij}^*$ implies that the pseudo-distances $d_{ij}^2$ are also identified for all pairs of agents $i$ and $j$. Hence, as discussed in Section \ref{ssec: beta ID}, agents with the same values of $\xi$ (and different values of $X$) and, subsequently, $\beta_0$ can be identified based on $d_{ij}^2$.

\section{Estimation of the Semiparametric Model}
\label{sec: estimation}
In this section, we turn the ideas of Section \ref{sec: ID} into an estimation procedure. First, we construct an estimator of $\beta_0$ assuming that some estimator of the pseudo-distances $\hat d_{ij}^2$ is already available for the researcher. Then, we discuss how to construct $\hat d_{ij}^2$ in the homoskedastic and general heteroskedastic settings.

\subsection{Estimation of $\beta_0$}
\label{ssec: est of beta}
Suppose that we start with some (uniformly consistent) estimator of the pseudo-distances $d_{ij}^2$ denoted by $\hat d_{ij}^2$. Using $\hat d_{ij}^2$, we construct the following kernel based estimator of $\beta_0$
\begin{align}
  \label{eq: hat beta def}
  \hat \beta \coloneqq \left(\sum_{i < j} K\left(\frac{\hat d_{ij}^2}{h_n^2}\right) \sum_{k \neq i,j} \Delta W_{ijk} \Delta W_{ijk}' \right)^{-1} \left(\sum_{i < j} K\left(\frac{\hat d_{ij}^2}{h_n^2}\right) \sum_{k \neq i,j} \Delta W_{ijk} \Delta Y_{ijk}\right),
\end{align}
where $\Delta W_{ijk} \coloneqq W_{ik} - W_{jk}$ and $\Delta Y_{ijk} \coloneqq Y_{ik} - Y_{jk}$, $K: \mathbb R_+ \rightarrow \mathbb R$ is some kernel supported on $[0,1]$, and $h_n$ is a bandwidth, which needs to satisfy $h_n \rightarrow 0$ (and some additional requirements) as $n \rightarrow \infty$. Hereafter, we will also use the notations $\sum_{i < j} \coloneqq \sum_{i,j \in [n], i < j}$ and $\sum_{k \neq i,j} \coloneqq \sum_{k \in [n], k \neq i,j}$, where $[n] = \{1, \dots, n\}$.

As discussed previously, $\beta_0$ can be estimated by the regression of $\Delta Y_{ijk}$ on $\Delta W_{ijk}$ with fixed agents $i$ and $j$ satisfying $\xi_i = \xi_j$, and, consequently, $d_{ij}^2 = 0$; see \eqref{eq: delta Y delta W reg def} and \eqref{eq: beta_0 ID}. However, in a finite sample, we are never guaranteed to find a pair of agents with exactly the same values of unobserved characteristics. The proposed estimator $\hat \beta$ addresses this issue: it combines all of the pairwise-difference regressions weighted by $K({\hat d_{ij}^2}/{h_n^2})$. Typically, the smaller $\hat d_{ij}^2$ is, the closer agents $i$ and $j$ appear to be in terms of $\xi_i$ and $\xi_j$, and the higher weight is given to the corresponding pairwise-difference regression. Specifically, we show that, with probability approaching one, only the pairs that satisfy $\norm{\xi_i - \xi_j} \leqslant \alpha h_n$ are given positive weights, where $\alpha$ is some positive constant. Since $h_n \rightarrow 0$, the quality of these matches increases and the bias introduced by the imperfect matching vanishes as the sample size grows. In Section \ref{ssec: beta rate}, we formalize this discussion by providing the necessary regularity conditions and establishing the rate of convergence for $\hat \beta$.

\subsection{Estimation of $d_{ij}^2$}
\label{ssec: est of d}

The kernel based estimator \eqref{eq: hat beta def} builds on the estimated pseudo-distances $\{\hat d_{ij}^2\}_{i \neq j}$. In this section, we construct particular estimators of $d_{ij}^2$ for both homoskedastic and general heteroskedastic settings. Their asymptotic properties will be established in Section~\ref{ssec: d star rates}.

\subsubsection{Estimation of $d_{ij}^2$ under conditional homoskedasticity of $\varepsilon_{ij}$}

We start with considering the homoskedastic setting. Recall that in this case, the pseudo-distance of interest $d_{ij}^2$ is closely related to another quantity $q_{ij}^2$ defined in \eqref{eq: q_ij def}. Specifically, according to \eqref{eq: homo d}, $q_{ij}^2 = d_{ij}^2 + 2 \sigma^2$, where $\sigma^2$ stands for the conditional variance of $\varepsilon_{ij}$. Moreover, unlike $d_{ij}^2$, $q_{ij}^2$ can be directly estimated from the raw data as in \eqref{eq: hat q_ij def}.

Then, a natural way to estimate $d_{ij}^2$ is to subtract $2 \hat \sigma^2$ from $\hat q_{ij}^2$, where $\hat \sigma^2$ is an estimator $\sigma^2$. One candidate estimator of $\sigma^2$ is given by
\begin{align}
  \label{eq: 2 sigma hat def}
  2 \hat \sigma^2 = \min_{i, j \neq j} \hat q_{ij}^2.
\end{align}
Indeed, in large samples, we expect $\min_{i, j \neq i} d_{ij}^2$ to be small since we are likely to find a pair of agents similar in terms of $\xi$. Hence, in large samples, $\min_{i, j \neq i} q_{ij}^2 = \min_{i, j \neq i} d_{ij}^2 + 2 \sigma^2$ is expected to be close to $2 \sigma^2$. Then, $d_{ij}^2$ can be estimated by
\begin{align}
  \label{eq: hat d def}
  \hat d_{ij}^2 = \hat q_{ij}^2 - 2 \hat \sigma^2 = \hat q_{ij}^2 - \min_{i, j \neq i} \hat q_{ij}^2.
\end{align}

\subsubsection{Estimation of $d_{ij}^2$ under general heteroskedasticity of $\varepsilon_{ij}$}
As suggested by Section \ref{ssec: hetero ID}, under general heteroskedasticity of the errors, the first step of estimation of $d_{ij}^2$ is to construct an estimator of $Y_{ij}^*$. Once $\hat Y_{ij}^*$ are constructed for all pairs of agents, $d_{ij}^2$ can be estimated by
\begin{align}
  \label{eq: hat d star def}
  \hat d_{ij}^2 \coloneqq \min_{\beta \in \mathcal B} \frac{1}{n} \sum_{k=1}^n (\hat Y_{ik}^* - \hat Y_{jk}^* - (W_{ik} - W_{jk})' \beta )^2.
\end{align}

As pointed out in Section~\ref{ssec: Y star ID}, the problem of estimation of $Y_{ij}^*$ is closely related to the classic matrix (graphon) estimation/completion problem, with a number of candidate estimators available in the literature (see, e.g., \citealp{Chatterjee2015}). The researcher can pick an appropriate estimator of $\hat Y_{ij}^*$ depending on the context.

In this paper, we propose estimating $Y_{ij}^*$ extending the approach of \citet{Zhang2017}. For simplicity, first, we consider the case when $X$ is discrete and takes finitely many values. Formal statistical guarantees for the proposed estimator and a discussion of the general case are provided in Section \ref{sssec: general hetero rates}.

First, for all pairs of agents $i$ and $j$, we estimate
\begin{align}
  \label{eq: hat d inf def}
  \hat d_{\infty}^2(i,j) \coloneqq \max_{k \neq i,j} \Big\vert{(n-3)^{-1} \sum_{\ell \neq i,j,k} (Y_{i \ell} - Y_{j \ell}) Y_{k \ell}}\Big\vert.
\end{align}
Then, for any agent $i$, we define its neighborhood $\hatN_i (n_i)$ as a collection of $n_i$ agents closest to agent $i$ in terms of $\hat d_{\infty}^2$ among all agents with $X = X_i$, i.e.
\begin{align}
  \label{eq: N_i def}
  \hatN_i (n_i) \coloneqq \{i': X_{i'} = X_i, \text{Rank}(\hat d_{\infty}^2(i,i')|X = X_i) \leqslant n_i \}.
\end{align}
Also notice that by construction, $i \in \hatN_i (n_i)$, so agent $i$ is always included in its neighborhood. Essentially, for any agent $i$, its neighborhood $\hatN_i (n_i)$ is a collection of agents with the same observed and similar unobserved characteristics. Note that since $X$ is discrete and takes finitely many values, we require $X_{i'} = X_i$. Also, note that the number of agents included in the neighborhoods should grow (at a certain rate) as the sample size increases.

Once the neighborhoods are constructed, we estimate $Y_{ij}^*$ by
\begin{align}
  \label{eq: hat Y def}
  \hat Y_{ij}^* = \frac{\sum_{i' \in \hatN_i(n_i)} Y_{i' j}}{n_i},
\end{align}
where, for the ease of notation, we let $Y_{i' j} = 0$ whenever $i' = j$. Note that $\hat Y_{ij}^*$ is also defined for $i = j$: despite $Y_{ii}$ is not observed, we still can estimate the associated error free outcome $Y_{ii}^* \coloneqq w(X_i,X_i) + g(\xi_i,\xi_i)$.

\begin{Rem}
  Notice that the proposed estimator \eqref{eq: hat Y def} differs from the one discussed in Section \ref{ssec: Y star ID}. Specifically, \eqref{eq: Y star ID} suggests using
  \begin{align}
    \label{eq: tilde Y def}
    \tilde Y_{ij}^* = \frac{1}{n_i n_j} \sum_{i' \in \hatN_i(n_i)} \sum_{j' \in \hatN_j(n_j)} Y_{i' j'}.
  \end{align}
  Next, we will demonstrate that the rate of convergence for $\hat \beta$ depends on the asymptotic properties of the first step estimator $\hat d_{ij}^2$. While $\tilde Y_{ij}^*$ is a natural and (uniformly) consistent estimator of $Y_{ij}^*$, i.e., we will later establish $\max_{i,j} \vert\tilde Y_{ij}^* - Y_{ij}^*\vert = o_p(1)$, it turns out that using $\hat Y_{ij}^*$ as in \eqref{eq: hat Y def} delivers better rates of converges for $\hat d_{ij}^2$ and, consequently, for $\hat \beta$ too.
\end{Rem}

\section{Large Sample Theory}
\label{sec: LST}
In this section, we formally study the asymptotic properties of the estimators we provided in Section \ref{sec: estimation}. The following set of basic regularity conditions will be used throughout the rest of the paper.

\begin{Ass}
  \label{ass: basic}
  \leavevmode
  \begin{enumerate}[(i)]
    \item \label{item: w} $w: \mathcal X \times \mathcal X \rightarrow \mathbb R^{p}$ is a symmetric bounded function, where $\supp{X} \subseteq \mathcal X$;

    \item \label{item: xi space} $\supp{\xi} \subseteq \mathcal E$, where $\mathcal E$ is a compact subset of $\mathbb R^{d_\xi}$;

    \item \label{item: lipschitz g} $g: \mathcal E \times \mathcal E \rightarrow \mathbb R$ is a symmetric bounded function; moreover, for some $\overline G > 0$, we have
    \begin{align*}
      \abs{g(\xi_1, \xi) - g(\xi_2,\xi)} \leqslant \overline G \norm{\xi_1 - \xi_2} \quad \text{for all $\xi_1,\xi_2,\xi \in \mathcal E$;}
    \end{align*}

    \item \label{item: eps moments} for some $c > 0$, $\ex{e^{\lambda \varepsilon_{ij}}|X_i,\xi_i,X_j,\xi_j} \leqslant e^{c \lambda^2}$ for all $\lambda \in \mathbb R$ a.s.
  \end{enumerate}
\end{Ass}

Conditions \eqref{item: w} and \eqref{item: xi space} are standard. Condition \eqref{item: lipschitz g} requires $g(\cdot,\cdot)$ to be (uniformly) Lipschitz continuous. Condition \eqref{item: eps moments} requires the conditional distribution of the error $\varepsilon_{ij}|X_i,\xi_i,X_j,\xi_j$ to be sub-Gaussian, uniformly over $(X_i,\xi_i,X_j,\xi_j)$. It allows us to invoke certain concentration inequalities and derive rates of uniform convergence.

\subsection{Rate of convergence for $\hat \beta$}
\label{ssec: beta rate}
In this section, we provide necessary regularity conditions and establish the rate of convergence for the kernel based estimator $\hat \beta$ introduced in \eqref{eq: hat beta def}. We will do that assuming that our candidate estimator $\hat d_{ij}^2$ converges to $d_{ij}^2$ (uniformly across all pairs) at a certain rate $R_n \rightarrow \infty$, i.e.,
\begin{align}
  \label{eq: R_n def}
  \min_{i,j\neq i} \vert\hat d_{ij}^2 - d_{ij}^2\vert = O_p(R_n^{-1}).
\end{align}
We will first derive the result treating \eqref{eq: R_n def} as a high-level assumption, and then verify that it holds and characterize $R_n$ for the candidate estimators \eqref{eq: hat d def} and \eqref{eq: hat d star def}.

We establish consistency and derive a guaranteed rate of convergence of $\hat \beta$ under the following regularity conditions. For simplicity of exposition, we state these conditions for the case when $\xi$ is scalar but the main result of this section still holds for multivariate $\xi$ under minimal appropriate modifications of the assumptions below.

\begin{Ass}
  \label{ass: xi dist}
  \leavevmode
  \begin{enumerate}[(i)]
    \item \label{item: xi pdf}  $\xi \in \mathbb R$, and $\xi|X = x$ is continuously distributed for all $x \in \supp{X}$; its conditional density $f_{\xi|X}$ (with respect to the Lebesgue measure) satisfies $\sup_{x \in \supp{X}} \sup_{\xi \in \mathcal E} f_{\xi|X}(\xi|x) \leqslant \overline f_{\xi|X}$ for some constant $\overline f_{\xi|X} > 0$;
    
    \item \label{item: xi cont} for all $x \in \supp{X}$, $f_{\xi|X}(\xi|x)$ is continuous at almost all $\xi$ (with respect to the conditional distribution of $\xi|X=x$); moreover, there exist positive constants $\overline \delta$ and $\gamma$ such that for all $\delta \in (0, \overline \delta)$ and for all $x \in \supp{X}$,
    \begin{align}
      \label{eq: xi cont bound}
      \pr\left(\xi_i \in \{\xi: f_{\xi|X}(\xi|x) \text{ is continuous on } B_{\delta}(\xi) \} |X_i = x\right) \geqslant 1 - \gamma \delta;
    \end{align}

    \item \label{item: xi lipschitz} there exists $C_\xi > 0$ such that for all $x \in \supp{X}$ and for any convex set $\mathcal D \in \mathcal E$ such that $f_{\xi|X}(\cdot;x)$ is continuous on $\mathcal D$, we have $\abs{f_{\xi|X}(\xi_1|x) - f_{\xi|X}(\xi_2|x)} \leqslant C_\xi \abs{\xi_1 - \xi_2}$.    
  \end{enumerate}
\end{Ass}


Assumption \ref{ass: xi dist} describes the properties of the conditional distribution of $\xi|X$. Note that we focus on the case when $\xi|X=x$ is continuously distributed for all $x \in \supp{X}$ even though our framework straightforwardly allows for the (conditional) distribution of~$\xi$ to have point masses or to be discrete. In fact, the asymptotic analysis is substantially simpler in the latter case. Specifically, if $\xi$ is discrete (and takes finitely many values), the agents can be consistently clustered into groups with the same values of $\xi$ based on the same pseudo-distance $\hat d_{ij}^2$. In this case, $\hat \beta$ is asymptotically equivalent to the oracle pairwise-difference estimator using the knowledge of the true cluster membership, and, as a result, it is asymptotically normal and unbiased. Moreover, in this case, $\beta_0$ can also be estimated by the pooled linear regression, which includes additional interactions of the dummy variables for the estimated cluster membership.\footnote{Similar ideas are also explored in \citet{mugnier2025simple} in the context of grouped panel models.}

Conditions \eqref{item: xi cont} and \eqref{item: xi lipschitz} are weak smoothness requirements. The second part of Condition~\eqref{item: xi cont} bounds the probability mass of $\xi|X=x$, for which $f_{\xi|X}(\xi|X)$ is not potentially continuous on a ball $B_\delta(\xi)$. It allows us to control the probability mass of $\xi$ close to the boundary of its support, where $f_{\xi|X}(\xi|X)$ is allowed to be discontinuous.

\begin{Ex*}[Illustration of Assumption \ref{ass: xi dist}\eqref{item: xi cont}]
  Suppose $\xi|X=x$ is supported and continuously distributed on $[0,1]$ for all $x \in \supp{X}$. Then $f_{\xi|X}(\xi|x)$ is continuous on $B_\delta (\xi)$ for all $\xi \in [\delta, 1 - \delta]$. Then \eqref{eq: xi cont bound} is satisfied with $\gamma = 2 \overline f_{\xi|X}$, where $\overline f_{\xi|X}$ is as in Assumption \ref{ass: xi dist}\eqref{item: xi pdf}. \QEDA
\end{Ex*}

\begin{Ass}
  \label{ass: beta ID}
  \leavevmode
  \begin{enumerate}[(i)]
    \item \label{item: xi full rank} there exist $\underline \lambda > 0$ and $\underline \delta > 0$ such that
    \begin{align*}
      \pr \left( (X_i,X_j) \in \left\{ (x_1,x_2): \lambda_{min}(\mathcal C(x_1,x_2)) > \underline \lambda, \int f_{\xi|X}(\xi|x_1) f_{\xi|X}(\xi|x_2) d \xi > \underline \delta \right\} \right) > 0,
    \end{align*}
    where
    \begin{align}
      \label{eq: C x1 x2 def}
      \mathcal C (x_1,x_2) \coloneqq \ex{(w(x_1,X) - w(x_2,X)) (w(x_1,X) - w(x_2,X))'};
    \end{align}

    \item \label{item: d* xi ID} for each $\delta > 0$, there exists $C_\delta > 0$ such that
    \begin{align*}
      \inf_{\beta} \e\left[\left(g(\xi_i,\xi_k) - g(\xi_j,\xi_k) - \left(w(X_i,X_k) - w(X_j,X_k)\right)'\beta\right)^2 |X_i,\xi_i,X_j,\xi_j\right] > C_\delta
    \end{align*}
    a.s. for $(X_i,\xi_i)$ and $(X_j,\xi_j)$ satisfying $\abs{\xi_i - \xi_j} \geqslant \delta$;
  
    \item $d_{ij}^2 \equiv d^2(X_i,\xi_i,X_i,\xi_j) = c(X_i,X_j,\xi_i) (\xi_j - \xi_i)^2 + r(X_i,\xi_i,X_j,\xi_j)$, where $\abs{r(X_i,\xi_i,X_j,\xi_j)} \leqslant C \abs{\xi_j - \xi_i}^3$ a.s. for some $C > 0$, and $0 < \underline c < c(X_i,X_j,\xi_i) < \overline c$ a.s. \label{item: d* approximation}
  \end{enumerate}  
\end{Ass}

Assumption \ref{ass: beta ID} is a collection of identification conditions. Specifically, Condition \eqref{item: xi full rank} is the identification condition for $\beta_0$. It ensures that in a growing sample, it is possible to find a pair of agents $i$ and $j$ such that (i) $X_i$ and $X_j$ are ``sufficiently different'', so the minimal eigenvalue $\lambda_{min}(\mathcal C(X_i,X_j)) > \underline \lambda > 0$, (ii) and yet $\xi_i$ and $\xi_j$ are increasingly similar. The latter is guaranteed by $\int f_{\xi|X}(\xi|X_i) f_{\xi|X}(\xi|X_j) d \xi > \underline \delta$, which implies that the conditional supports of $\xi_i|X_i$ and $\xi_j|X_j$ have a non-trivial overlap. Condition \eqref{item: xi full rank} is crucial for establishing consistency of $\hat \beta$. 

Condition \eqref{item: d* xi ID} ensures that $d_{ij}^2$ is bounded away from zero whenever $\abs{\xi_i - \xi_j}$ is. Notice that it also guarantees that agents that are close in terms of $d_{ij}^2$, must also be similar in terms of $\xi$. Hence, Condition \eqref{item: d* xi ID} justifies using the pseudo-distance $d_{ij}^2$ for finding agents with similar values of $\xi$ in finite samples. It also can be interpreted as a rank type condition: for fixed agents $i$ and $j$ with $\xi_i \neq \xi_j$, $g(\xi_i,\xi_k) - g(\xi_j,\xi_k)$ cannot be expressed as a linear combination of the components of $W_{ik} - W_{jk}$.

Condition \eqref{item: d* approximation} is a local counterpart of Condition \eqref{item: d* xi ID}. It says that, as a function of $\xi_j$, $d^2(X_i,X_j,\xi_i,\xi_j)$ can be locally quadratically approximated around $\xi_j = \xi_i$, and the approximation remainder can be uniformly bounded as $O(\abs{\xi_j - \xi_i}^3)$. Condition \eqref{item: d* approximation} also explains why we divide $\hat d_{ij}^2$ by $h_n^2$ for computing the kernel weights. Indeed, locally $d_{ij}^2 \propto (\xi_j - \xi_i)^2$, so the bandwidth $h_n$ effectively controls how large $\abs{\xi_j - \xi_i}$ can be for the pair of agents $i$ and $j$ to get a positive weight $K(\frac{\hat d_{ij}^2}{h_n^2})$.

\begin{Ass}
  \label{ass: K and h}
  \leavevmode
  \begin{enumerate}[(i)]
    \item \label{item: kernel} $K: \mathbb R_+ \rightarrow \mathbb R$ is supported on $[0,1]$ and bounded by $\overline K < \infty$. $K$ satisfies $\mu_K \coloneqq \int K(u^2) du > 0$ and $\abs{K(z) - K(z')} \leqslant \overline K' \abs{z - z'}$ for all $z,z' \in \mathbb R_+$ for some $\overline K' > 0$;


    \item \label{item: h_n} $h_n \rightarrow 0$, $n h_n / \ln n \rightarrow \infty$ and $R_n h_n^2 \rightarrow \infty$ for some $R_n \rightarrow \infty$ satisfying \eqref{eq: R_n def}.
  \end{enumerate}
\end{Ass}

Assumption \ref{ass: K and h} specifies the properties of the kernel $K$ and the bandwidth $h_n$. Condition \eqref{item: kernel} imposes a number of fairly standard restrictions on $K$ including Lipschitz continuity. Condition \eqref{item: h_n} restricts the rates at which the bandwidth is allowed to shrink towards zero. The requirement $n h_n / \ln n \rightarrow \infty$ ensures that we have a growing number of potential matches as the sample size increases. Additionally, to get the desired results we need $R_n h_n^2 \rightarrow \infty$: the bandwidth cannot go to zero faster than $R_n^{-1/2}$. This requirement allows us to bound the effect of the sampling variability coming from the first step (estimation of $\{d_{ij}^2\}_{i \neq j}$) on the second step (estimation of $\beta_0$).

\begin{Ass}
  \label{ass: local smooth g} There exists a bounded function $G: \mathcal E \times \mathcal E \rightarrow \mathbb R$ such that for all $\xi_1,\xi_2,\xi \in \mathcal E$
    \begin{align*}
      g(\xi_1,\xi) - g(\xi_2,\xi) = G(\xi_1,\xi) (\xi_1 - \xi_2) + r_g (\xi_1,\xi_2,\xi);
    \end{align*}
    and there exists $C > 0$ such that for all $\delta_n \downarrow 0$
    \begin{align*}
      \lsup \frac{\sup_{\xi} \sup_{\xi_1: \abs{\xi_1 - \xi} > \delta_n } \sup_{\xi_2: \abs{\xi_2 - \xi_1} \leqslant \delta_n} \abs{r_g (\xi_1, \xi_2, \xi)} }{\delta_n^2}  < C.
    \end{align*}
\end{Ass}

Assumption \ref{ass: local smooth g} is a weak smoothness requirement. It guarantees that as a function of $\xi_2$, the difference $g(\xi_1,\xi) - g(\xi_2,\xi)$ can be (locally) linearized around $\xi_2 = \xi_1$ provided that $\xi_2$ is close to $\xi_1$ relative to the distance between $\xi_1$ and $\xi$ (guaranteed by the restrictions $\abs{\xi_1 - \xi} > \delta_n$ and $\abs{\xi_2 - \xi_1} \leqslant \delta_n$). The goal of introducing these restrictions is to allow for a \emph{possibly non-differentiable} $g$, e.g., $g(\xi_i,\xi_j) = \kappa \abs{\xi_i - \xi_j}$, since such models of latent homophily are common in the literature (e.g., \citealp{Hoff2002,handcock2007model}). We provide an illustration of Assumption~\ref{ass: local smooth g} in Appendix~\ref{sec: non-diff illustration}.

\begin{The}
  \label{the: kernel rate}
  Suppose that~\eqref{eq: R_n def} holds for some $R_n \rightarrow \infty$. Then, under Assumptions \ref{ass: DGP}-\ref{ass: local smooth g},
  \begin{align}
    \label{eq: the rate}
    \hat \beta - \beta_0 = O_p \left(h_n^2 + \frac{R_n^{-1}}{h_n} + \frac{R_n^{-1}}{h_n^2} \left(\frac{\ln n}{n}\right)^{1/2} + n^{-1} \right).
  \end{align}
\end{The}

Theorem~\ref{the: kernel rate} provides a guaranteed rate of convergence of $\hat \beta$ and, thus, formally establishes identification of $\beta_0$. The derived rate consists of four terms. The $h_n^2$ term accounts for the bias due to the imperfect $\xi$-matching. The terms involving $R_n$ are due to the fact that $\{d_{ij}^2\}_{i \neq j}$ are unknown and need to be estimated. Finally, $n^{-1}$ is the standard sampling variability term in a regression with $O(n^2)$ observations.

In Section \ref{ssec: d star rates}, we will derive $R_n$ for the considered estimators of $\{d_{ij}^2\}_{i \neq j}$ and complete the characterization of the rate of convergence for $\hat \beta$ provided in \eqref{eq: the rate}. In particular, we will demonstrate that, under certain conditions, the candidate estimators \eqref{eq: hat d def} and \eqref{eq: hat d star def} satisfy \eqref{eq: R_n def} with $R_n = (\frac{n}{\ln n})^{1/2}$, or with an even slower-growing $R_n$ depending on $d_\xi$. For such values of $R_n$, \eqref{eq: the rate} effectively simplifies as
\begin{align}
  \label{eq: beta rate}
  \hat \beta - \beta_0 = O_p \left(h_n^2 + \frac{R_n^{-1}}{h_n} \right),
\end{align}
where we also used $\frac{R_n^{-1}}{h_n^2} = o(1)$ implied by Assumption \ref{ass: K and h}\eqref{item: h_n}. In particular, under $h_n \propto R_n^{-1/3}$, $\hat \beta$ is guaranteed to achieve the following rate of convergence
\begin{align}
  \label{eq: beta opt rate}
  \hat \beta - \beta_0 = O_p\left(R_n^{-2/3}\right).
\end{align}

The rate of convergence established by Theorem~\ref{the: kernel rate} is not necessarily optimal and potentially can be improved, especially if additional smoothness conditions are imposed. The main goal of Theorem~\ref{the: kernel rate} is to demonstrate consistency of $\hat \beta$ and thus complement Section~\ref{sec: ID} by formally establishing identification of $\beta_0$ under minimal primitive conditions.

Finally, we want to highlight again two important features of the setting differentiating it from the previously studied frameworks and complicating the asymptotic analysis. First, as discussed in more detail in Section~\ref{sssec: ID comparison}, we consider low-dimensional (or even low-rank) regressors $W_{ij} = w(X_i,X_j)$ typical for network models instead of relying on ``high-rank'' regressors like most of the related literature does. Second, we only impose minimal smoothness assumptions on $g(\cdot, \cdot)$ and do not require it to be differentiable whereas the alternative approaches building low-rank approximations of interactive unobserved heterogeneity rely on the existence of multiple continuous derivatives of $g(\cdot,\cdot)$ (e.g., \citealp{fernandez2021low,freeman2023linear,beyhum2024inference}); see also Remark~\ref{rem: smoothness} below for a clarification of the role of Assumption~\ref{ass: local smooth g}.

As discussed earlier, these features are highly representative of network models. Moreover, their unique combination makes the asymptotic analysis highly nonstandard, potentially invalidating the previously proposed methods and established statistical guarantees. For this reason, in this paper, we do not attempt to construct an asymptotically normal estimator of $\beta_0$ or provide an inference method at the cost of imposing additional assumptions restrictive in network models. Instead, we focus on establishing identification of $\beta_0$ by providing a new estimator and showing its consistency in the general setting of the paper.

\begin{Rem}
  \label{rem: smoothness}
  Assumption~\ref{ass: local smooth g} is not essential for consistency of $\hat \beta$. It can be dropped at the cost of increasing the magnitude of the bias of $\hat \beta$ from $h_n^2$ to $h_n$. Even when Assumption~\ref{ass: local smooth g} is imposed, bounding the bias of $\hat \beta$ is non-trivial. In particular, it requires careful accounting for the agents with $\xi$ ``close'' to the boundary of its support because kernel smoothing methods suffer from larger biases on the boundary even when the estimated function is smooth.
\end{Rem}

\begin{Rem}[Extension to $d_\xi > 1$]
  \label{rem: multivariate xi}
  Importantly, the result of Theorem~\ref{the: kernel rate} remains the same for $d_\xi > 1$ provided that the condition $n h_n/\ln n \rightarrow \infty$ in Assumption \ref{ass: K and h}\eqref{item: h_n} is generalized as $n h_n^{d_{\xi}}/\ln n \rightarrow \infty$, and the other conditions are analogously restated in terms of multivariate $\xi$, if needed. Also notice that, the rates provided in \eqref{eq: the rate} and \eqref{eq: beta rate} would still implicitly depend on $d_\xi$ through $R_n$ (and the restrictions imposed on $h_n$ by Assumption \ref{ass: K and h}\eqref{item: h_n}).
\end{Rem}

\subsection{Rates of uniform convergence for $\hat d_{ij}^2$}
\label{ssec: d star rates}
In this section, we complement the result of Theorem~\ref{the: kernel rate} by characterizing $R_n$, the rate of uniform convergence for $\hat d_{ij}^2$ defined in \eqref{eq: R_n def}, for the candidate estimators \eqref{eq: hat d def} and \eqref{eq: hat d star def} in the homoskedastic and then in the general heteroskedastic settings.

\subsubsection{Homoskedastic model}
\label{sssec: homoskedastic rate}

First, we consider the homoskedastic case, i.e., we assume that the idiosyncratic errors satisfy \eqref{eq: homosked eps}. As discussed in Section \ref{ssec: est of d}, in this case, the suggested estimator is given by $\hat d_{ij}^2 = \hat q_{ij}^2 - 2 \hat \sigma^2$, and $d_{ij}^2 = q_{ij}^2 - 2 \sigma^2$ (see \eqref{eq: hat d def} and \eqref{eq: homo d}, respectively). Thus,
\begin{align}
  \max_{i, j \neq i} \abs{\hat d_{ij}^2 - d_{ij}^2} = \max_{i, j \neq i}\abs{(\hat q_{ij}^2 - q_{ij}^2) - (2 \hat \sigma^2 - 2 \sigma^2)} \leqslant \max_{i, j \neq i} \abs{\hat q_{ij}^2 - q_{ij}^2} + \abs{2 \hat \sigma^2 - 2 \sigma^2}. \label{eq: homo d star bound}
\end{align}
This bound, together with the following lemma, allows us to characterize $R_n$.

\begin{Lem}
  \label{lem: homosked d UC}
  Suppose that \eqref{eq: homosked eps} holds and $\mathcal B$ is compact. Then, under Assumptions \ref{ass: DGP} and \ref{ass: basic},
  \begin{align}
    \max_{i, j \neq i} \abs{\hat q_{ij}^2 - q_{ij}^2} &= O_{p}\left((\ln n / n)^{1/2}\right),\label{eq: hat d ij UC}\\
    2 \hat \sigma^2 - 2 \sigma^2 &=  \overline G^2 \min_{i \neq j} \norm{\xi_i - \xi_j}^2 + O_{p}\left((\ln n / n)^{1/2}\right),\label{eq: 2 hat sigma rate}
  \end{align}
  where $\hat q_{ij}^2$, $q_{ij}^2$, $2 \hat \sigma^2$ are given by \eqref{eq: hat q_ij def}, \eqref{eq: q_ij def}, \eqref{eq: 2 sigma hat def}, and $\overline G$ is defined in Assumption \ref{ass: basic}\eqref{item: lipschitz g}.
\end{Lem}

The bounds~\eqref{eq: hat d ij UC} and \eqref{eq: 2 hat sigma rate} established by Lemma~\ref{lem: homosked d UC} together with \eqref{eq: homo d star bound} allow us to provide $R_n$ for the estimator \eqref{eq: hat d def} in the homoskedastic setting. To make this characterization complete, we also need to provide a bound on $\norm{\xi_i - \xi_j}^2$ in \eqref{eq: 2 hat sigma rate}.

In particular, since $\mathcal E$ is bounded (Assumption \ref{ass: basic}\eqref{item: xi space}), we can guarantee that
\begin{align}
  \label{eq: dist xi simple bound}
  \min_{i \neq j} \norm{\xi_i - \xi_j}^2 \leqslant C n^{-2/d_{\xi}}
\end{align}
for some $C > 0$. While this bound is conservative, it ensures that, for $d_\xi \leq 4$, the contribution of $\norm{\xi_i - \xi_j}^2$ is negligible, resulting in the following corollary of Lemma~\ref{lem: homosked d UC}.

\begin{Cor}
  \label{cor: homosked low d_xi}
  Suppose that the hypotheses of Lemma \ref{lem: homosked d UC} are satisfied. Then, for $d_\xi \leqslant 4$,
  \begin{align*}
    \max_{i, j \neq i} \abs{\hat d_{ij}^2 - d_{ij}^2} = O_{p}\left((\ln n / n)^{1/2}\right),
  \end{align*}
  where $\hat d_{ij}^2$ and $d_{ij}^2$ are given by \eqref{eq: hat d def} and \eqref{eq: d_ij star def}, respectively.
\end{Cor}

Corollary \ref{cor: homosked low d_xi} ensures that when the errors are homoskedastic and $d_\xi \leqslant 4$, $\hat d_{ij}^2$ given by \eqref{eq: hat d def} satisfies \eqref{eq: R_n def} with $R_n = \left(\frac{n}{\ln n}\right)^{1/2}$. Hence, when $h_n \propto R_n^{-1/3} = \left(\frac{\ln n}{n}\right)^{-1/6}$, \eqref{eq: beta opt rate} implies that
\begin{align*}
  \hat \beta - \beta_0 = O_{p}\left((\ln n / n)^{1/3}\right),
\end{align*}
so the guaranteed rate of convergence for $\hat \beta$ is $\left(\frac{n}{\ln n}\right)^{1/3}$.

\begin{Rem}
  For $d_\xi > 5$, it is still possible to provide a simple yet \emph{conservative} bound on $R_n$ by combining \eqref{eq: dist xi simple bound} with the result of Lemma~\ref{lem: homosked d UC}. In particular, in this case, we can guarantee that $2 \hat \sigma^2 - 2 \sigma^2 = O_p\left(n^{-2/d_\xi}\right)$, and, consequently, we also have
  \begin{align*}
    \max_{i, j \neq i} \abs{\hat d_{ij}^2 - d_{ij}^2} = O_p\left(n^{-2/d_\xi}\right).
  \end{align*}
  We stress that these bounds are loose. With a more detailed analysis of the asymptotic behavior of $\min_{i \neq j} \norm{\xi_i - \xi_j}^2$ (which is outside the scope of this paper), these results can be substantially refined.
\end{Rem}

\subsubsection{Model with general heteroskedasticity}
\label{sssec: general hetero rates}

In this section, we establish the rate of uniform convergence for $\hat d_{ij}^2$ under general heteroskedasticity of the errors. First, we suppose that $X$ is discrete and derive $R_n$ for the estimator given by \eqref{eq: hat d star def}-\eqref{eq: hat Y def}. Then, we discuss how this estimator can be modified to accommodate continuously distributed $X$. Finally, we also provide a general result establishing $R_n$ for $\hat d_{ij}^2$ based on a generic estimator of $Y_{ij}^*$, potentially other than \eqref{eq: hat Y def}.

\bigskip

\noindent
\textbf{Estimation of $d_{ij}^2$ when $X$ is discrete\\}
Now we formally derive the rate of uniform convergence for $\hat d_{ij}^2$ given by \eqref{eq: hat d star def}-\eqref{eq: hat Y def}. This rate crucially depends on the asymptotic properties of the denoising estimator $\hat Y_{ij}^*$.

To formally state the asymptotic properties of $\hat Y_{ij}^*$, we introduce the following assumption.

\begin{Ass}
  \label{ass: d* rate}
  \leavevmode
  \begin{enumerate}[(i)]
    \item \label{item: disc X} $X$ is discrete and takes finitely many values $\{x_1, \dots, x_R\}$;
    
    \item \label{item: ball mass} there exist positive constants $\kappa$ and $\overline \delta$ such that for all $x \in \supp{X}$, for all $\xi' \in \supp{\xi|X = x}$, $\pr (\xi \in B_{\delta} (\xi') | X = x ) \geqslant \kappa \delta^{d_\xi}$ for all positive $\delta \leqslant \overline \delta$.
  \end{enumerate}
\end{Ass}

As pointed out before, we suppose that $X$ is discrete and takes finitely many values. Condition \eqref{item: ball mass} is a weak condition imposed on the conditional distribution of $\xi|X$. If $\xi|X$ is continuously distributed, it is satisfied when the conditional density $f_{\xi|X}(\xi|x)$ is (uniformly) bounded away from zero and its support is not ``too irregular''.

For any matrix $A \in \mathbb R^{n \times n}$, let $\norm{A}_{2,\infty} \coloneqq \max_{i} \sqrt{\sum_{j=1}^n A_{ij}^2}$. Also let $\hat Y^*$ and $Y^*$ denote $n \times n$ matrices with entries given by $\hat Y_{ij}^*$ and $Y_{ij}^*$.

\begin{The}
  \label{the: Y star convergence}
  Suppose that for all $i$, $ \underline C (n \ln n)^{1/2} \leqslant n_i \leqslant \overline C (n \ln n)^{1/2}$ for some positive constants $\underline C$ and $\overline C$. Then, under Assumptions \ref{ass: DGP}, \ref{ass: basic}, \ref{ass: d* rate}, for $\hat Y_{ij}^*$ given by \eqref{eq: hat Y def} we have:
  \begin{enumerate}[(i)]
    \item \label{item: 2 inf rate} $n^{-1} \Vert{\hat Y^* - Y^*}\Vert_{2,\infty}^2 = O_p\left(\left({\ln n}/{n}\right)^{\frac{1}{2 d_\xi}}\right)$;

    \item \label{item: Y star covar} $\max_{k} \max_{i} \vert{n^{-1} \sum_{\ell} Y_{k \ell}^* (\hat Y_{i \ell}^* - Y_{i \ell}^*) }\vert = O_p\left(\left({\ln n}/{ n}\right)^{\frac{1}{2 d_\xi}}\right).$
  \end{enumerate} 
\end{The}

Theorem \ref{the: Y star convergence} establishes two important asymptotic properties of $\hat Y_{ij}^*$. In fact, both results play key roles in bounding $R_n$, the rate of uniform convergence for $\hat d_{ij}^2$.

Part \eqref{item: 2 inf rate} is analogous to the result of \citet{Zhang2017}. While \citet{Zhang2017} only consider binary outcomes and do not allow for observed covariates, we extend their result by allows for (i) $d_\xi > 1$, (ii) possibly non-binary outcomes and unbounded idiosyncratic errors, (iii) observed (discrete) covariates $X$. Part \eqref{item: Y star covar} is new. It allows us to substantially improve on $R_n$ compared to what Part \eqref{item: 2 inf rate} can guarantee individually.\footnote{See Lemma \ref{lem: d star rate} for the comparison of the rates.}

Note that Theorem \ref{the: Y star convergence} requires $n_i$, the number of agents included into $\hatN_i (n_i)$, to grow at $(n \ln n)^{1/2}$ rate. As shown in \citet{Zhang2017}, this rate is optimal.\footnote{The optimal choice of $n_i$ remains the same for $d_\xi > 1$.} In applications, the authors also recommend taking $n_i = C (n \ln n)^{1/2}$ with $C \simeq 1$ and document robustness of their numerical results to the choice of $C$.

Building on Theorem \ref{the: Y star convergence}, we now provide rate of uniform convergence for $\hat d_{ij}^2$.

\begin{The}
  \label{the: disc X d star rate}
  Suppose that the hypotheses of Theorem \ref{the: Y star convergence} are satisfied and $\mathcal B = \mathbb R^p$. Then,
  \begin{align*}
    \max_{i, j \neq i} \vert{\hat d_{ij}^2 - d_{ij}^2}\vert = O_p\left(\left({\ln n}/{n}\right)^{\frac{1}{2 d_\xi}}\right),
  \end{align*}
  where $\hat d_{ij}^2$ and $d_{ij}^2$ are given by \eqref{eq: hat d star def}-\eqref{eq: hat Y def} and \eqref{eq: d_ij star def}, respectively.
\end{The}

Theorem \ref{the: disc X d star rate} establishes the rate of uniform convergence for the proposed estimator of $d_{ij}^2$. Specifically, it ensures that, for the considered $\hat d_{ij}^2$, \eqref{eq: R_n def} holds with $R_n = \left(\frac{n}{\ln n}\right)^{\frac{1}{2 d_\xi}}$. Combined with \eqref{eq: beta opt rate}, this guarantees that under the appropriate choice of the bandwidth $h_n$, we have
\begin{align*}
  \hat \beta - \beta_0 = O_p\left(\left({\ln n}/{n}\right)^{\frac{1}{3 d_\xi}}\right).
\end{align*}
Thus, when $X$ is discrete, $\beta_0$ can be estimated (at least) at $\left(\frac{n}{\ln n}\right)^{\frac{1}{3 d_\xi}}$ rate. If $\xi$ is scalar, the guaranteed rate of convergence for $\hat \beta$ is the same as in the homoskedastic case.

\bigskip

\noindent
\textbf{Estimation of $d_{ij}^2$ when $X$ is continuously distributed\\}
The proposed estimator $\hat Y_{ij}^*$ given by \eqref{eq: N_i def}-\eqref{eq: hat Y def} can be straightforwardly modified to allow for continuously distributed (components of) $X$.  Since, in this case, the probability of finding two agents with exactly the same values of $X$ is zero, we have to modify the construction of $\hatN_i (n_i)$ previously provided in \eqref{eq: N_i def}. One natural possibility is to consider
\begin{align}
  \label{eq: hat N cont X def}
  \hatN_i (n_i;\delta_n) \coloneqq \{i': \norm{X_{i'} - X_i} \leqslant \delta_n, \text{Rank}(\hat d_{\infty}^2(i,i')|\norm{X - X_i} \leqslant \delta_n) \leqslant n_i \}.
\end{align}
The parameter $\delta_n$ controls the quality of matching based on $X$. To guarantee consistency of $\hat Y_{ij}^*$, we will need $\delta_n$ to converge to zero but slowly enough to ensure that we can still find a growing number of good matches increasingly similar in terms of both $X$ and $\xi$. Once the constructed neighborhoods are appropriately modified, $\hat Y_{ij}^*$ can still be computed as in \eqref{eq: hat Y def}, and, with some work, the result of Theorem~\ref{the: Y star convergence} can be generalized accordingly to allow for continuously distributed covariates.

\begin{Rem}
  \label{rem: cont X Y star consistency}
  Another possibility to allow for continuously distributed (components of) $X$ is to simply treat it as unobserved, similarly to $\xi$. In this case, $(X,\xi)$ becomes the effective latent variable and $\hatN_i (n_i)$ can be constructed without conditioning on $X = X_i$. 
  Then the result of Theorem \ref{the: Y star convergence} can be applied to the resulting estimator delivering analogous rates of convergence with $d_\xi + d_X$ taking of the place of $d_\xi$, where $d_X$ denotes the dimension of $X$.
  While this construction of $\hat Y_{ij}^*$ might not necessarily be optimal, it allows us to formally cover the case when $X$ is continuously distributed by providing analogous consistency results.
\end{Rem}

\bigskip

\noindent
\textbf{Estimation of $d_{ij}^2$ using general matrix denoising techniques\\}
Theorem \ref{the: disc X d star rate} establishes uniform consistency of $\hat d_{ij}^2$ leveraging the specific structure of $\hat Y_{ij}^*$ in \eqref{eq: hat Y def} and requires $X$ to be discrete. To complement Theorem \ref{the: disc X d star rate} and formally establish identification of $d_{ij}^2$ and $\beta_0$ in the general setting, we will now provide a generic consistency result, which does not require $\hat Y_{ij}^*$ to have any particular structure and holds regardless of whether $X$ is discrete or continuously distributed.

\begin{Lem}
  \label{lem: d star rate}
  Suppose that $\hat Y^*$ satisfies $n^{-1} \Vert{\hat Y^* - Y^*}\Vert_{2, \infty}^2 = O_p (\mathcal R_n^{-1})$ for some $\mathcal R_n \rightarrow \infty$. Also suppose that $\mathcal B$ is compact. Then, under Assumptions \ref{ass: DGP} and \ref{ass: basic}, we have
  \begin{align*}
    \max_{i, j \neq i} \vert{\hat d_{ij}^2 - d_{ij}^2}\vert = O_p\left(\left({\ln n}/{n}\right)^{1/2} + \mathcal R_n^{-1/2}\right),
  \end{align*}
  where $\hat d_{ij}^2$ and $d_{ij}^2$ are given by \eqref{eq: hat d star def} and \eqref{eq: d_ij star def}, respectively.
\end{Lem}

Lemma~\ref{lem: d star rate} guarantees that $\hat d_{ij}^2$ in \eqref{eq: hat d star def} is uniformly consistent for $d_{ij}^2$ provided that $\hat Y_{ij}^*$ is consistent for $Y_{ij}^*$ in the $(2,\infty)$ norm, and thus it justifies using alternative estimators $\hat Y_{ij}^*$ available in the literature.

Together with the result of Theorem~\ref{eq: the rate}, Lemma~\ref{lem: d star rate} delivers consistency of $\hat \beta$ and thus formally establishes identification of $\beta_0$ in the general setting. In particular, as argued in Remark \ref{rem: cont X Y star consistency}, if $X$ is continuously distributed, we can still construct $\hat Y^*$ satisfying the requirement of the lemma with $\mathcal R_n = \left(\frac{n}{\ln n}\right)^{\frac{1}{2(d_\xi + d_X)}}$. Thus, Lemma~\ref{lem: d star rate} guarantees that $\beta_0$ can be consistently estimated when $X$ is continuously distributed.

\subsection{Uniformly consistent estimation of $Y_{ij}^*$}
\label{ssec: uniform estimation of Y}
One of the contributions of this paper is establishing identification of the error free outcomes $Y_{ij}^*$'s. In Section \ref{ssec: Y star ID}, we heuristically argued that $Y_{ij}^*$ is identified. In this section, we construct a uniformly consistent estimator of $Y_{ij}^*$ and, hence, formally prove its identification.

As before, first, we suppose that $X$ is discrete and takes finitely many values. The estimator we propose is an analogue of $\hat Y_{ij}^*$ given by \eqref{eq: hat Y def}, which we used before to construct $\hat d_{ij}^2$. It utilizes exactly the same neighborhoods as in \eqref{eq: N_i def} but, unlike $\hat Y_{ij}^*$, averages over all unique outcomes $Y_{i' j'}$ with $i' \in \hatN_i(n_i)$ and $j' \in \hatN_j(n_j)$.\footnote{Recall that in the studied undirected model, $Y_{ij} = Y_{ji}$, and $Y_{ij}$ is not observed for $i = j$.} For example, if $\hatN_i(n_i)$ and $\hatN_j(n_j)$ have no elements in common, then the proposed estimator takes a simple form as in \eqref{eq: tilde Y def}. More generally, for any $i$ and $j$, let
\begin{align}
  \hatM_{ij} \coloneqq \{ (i',j'): i' < j', (i' \in \hatN_i (n_i), j' \in \hatN_j (n_j)) \thickspace \text{or} \thickspace  (i' \in \hatN_j(n_j), j' \in \hatN_i(n_i))\}. \label{eq: hat Mij def}
\end{align}
Essentially, $\hatM_{ij}$ is a collection of unique unordered pairs of indices from the Cartesian product of $\hatN_i (n_i)$ and $\hatN_j (n_j)$.
 Then, $Y_{ij}^*$ is estimated by 
\begin{align}
  \label{eq: proper tilde Y}
  \tilde Y_{ij}^* = \frac{1}{m_{ij}} \sum_{(i',j') \in \hatM_{ij}} Y_{i'j'},
\end{align}
where $m_{ij}$ denotes the number of elements in $\hatM_{ij}$.

\begin{The}
  \label{the: Y star UC} Suppose that the hypotheses of Theorem \ref{the: Y star convergence} hold.
  Suppose that for any $\delta > 0$, there exists $C_\delta > 0$ such that $\int (g(\xi_i, \xi) - g(\xi_j, \xi))^2 dP_\xi (\xi) > C_\delta$ a.s. for $\abs{\xi_i - \xi_j} \geqslant \delta$. Then,
  \begin{align*}
    \max_{i,j} \vert{\tilde Y_{ij}^* - Y_{ij}^*}\vert = o_p(1).
  \end{align*}
\end{The}

Theorem \ref{the: Y star UC} demonstrates that $\tilde Y_{ij}^*$ is uniformly consistent for $Y_{ij}^*$ and, consequently, it formally proves that $Y_{ij}^*$ is identified. We also stress that the previously employed estimator $\hat Y_{ij}^*$ is not necessarily uniformly consistent since it averages over only $n_i$ outcomes $Y_{i' j}$. At the same time, $\tilde Y_{ij}^*$ averages over $m_{ij} = O(n_i n_j)$ outcomes, which allows us to establish the desired result.

\begin{Rem}
  Recall that, as discussed in Section~\ref{ssec: Y star ID}, the similarity distance $\hat d_{\infty}^2 (i,j)$ used to construct $\hat Y_{ij}^*$ and $\tilde Y_{ij}^*$ allows us to find agents $i$ and $j$ similar in terms of the $L^2$ distance between functions $g(\xi_i,\cdot)$ and $g(\xi_j,\cdot)$. The additional requirement imposed in Theorem~\ref{the: Y star UC} ensures that this similarity also translates into similarity between $\xi_i$ and $\xi_j$, which helps us to establish uniform consistency of $\tilde Y_{ij}^*$.\footnote{This condition is a weaker version of Assumption~\ref{ass: beta ID}\eqref{item: d* xi ID}. Together with Assumption~\ref{ass: basic}\eqref{item: lipschitz g}, it allows us to guarantee that the matched agents $i$ and $j$ are also similar in terms of the $L^\infty$ distance $\norm{g(\xi_i, \cdot) - g(\xi_j, \cdot)}_\infty$. While it is also possible to characterize the rate of uniform convergence under additional conditions allowing one to translate the $L^2$ rate into the $L^\infty$ rate for $\norm{g(\xi_i, \cdot) - g(\xi_j, \cdot)}$ such as Assumption~\ref{ass: beta ID}\eqref{item: d* approximation}, we do not pursue this direction here because the primary goal of Theorem~\ref{the: Y star UC} is establishing identification of $Y_{ij}^*$.}
\end{Rem}

\begin{Rem}
  \label{lem: UC with continuous X}
  Analogous uniform consistency results can also be obtained when $X$ is continuously distributed if $\tilde Y_{ij}^*$ is properly adjusted. As previously discussed, the possible adjustments include (i) constructing neighborhoods as in \eqref{eq: hat N cont X def}, or (ii) treating $X$ as unobserved (see Remark \ref{rem: cont X Y star consistency}). The result of Theorem~\ref{the: Y star UC} can be directly applied to the latter estimator to formally establish identification of $Y_{ij}^*$ in this case.
\end{Rem}

\begin{Rem}
  To the best of our knowledge, identifiability of the error free outcomes $Y_{ij}^*$ is a new result to the econometrics literature on identification of network and, more generally, two-way models. Moreover, Theorem \ref{the: Y star UC} also contributes to the statistics literature on graphon and, more generally, the latent space model estimation. Specifically, most of the previous work focused on establishing consistency and deriving rates in terms of the mean squared error (MSE) for $\hat Y^*$ (e.g., \citealp{Chatterjee2015,Gao2015,Klopp2017,Zhang2017,li2019nearest}). Theorem \ref{the: Y star UC} contributes to this literature by establishing $\Vert{\tilde Y^* - Y^*}\Vert_{\max} = o_p(1)$, i.e., demonstrating consistency of $\tilde Y^*$ in the max norm.
\end{Rem}

Another implication of Theorem \ref{the: Y star UC} is that the pair-specific fixed effects $g(\xi_i,\xi_j)$ can also be consistently estimated and, hence, are identified for all pair of agents $i$ and $j$. Consider  
\begin{align}
  \label{eq: hat g def}
  \hat g_{ij} = \tilde Y_{ij}^* - W_{ij}' \hat \beta,
\end{align}
where $\tilde Y_{ij}^*$ is given by \eqref{eq: proper tilde Y}. Since we have already demonstrated consistency of $\hat \beta$ and uniform consistency of $\tilde Y_{ij}^*$, $\hat g_{ij}$ is also uniformly consistent for $g_{ij} \coloneqq g(\xi_i,\xi_j)$. 

\begin{Cor}
  \label{cor: ID of g}
  Suppose that the hypotheses of Theorem \ref{the: Y star UC} are satisfied. Also, suppose that $\hat \beta - \beta_0 = o_p(1)$. Then, $\max_{i,j} \abs{\hat g_{ij} - g_{ij}} = o_p(1)$, where $\hat g_{ij}$ is given by \eqref{eq: hat g def}.
\end{Cor}

Establishing nonparametric identification of the pair-specific fixed effects $g_{ij}$'s is another contribution of the paper. This result is also of high empirical importance since in certain applications, the fixed effects are the primary object of interest.

\section{Extensions}
\label{sec: ext}

\subsection{Identification of single index and nonparametric models}
\label{ssec: NID}
In this section, we extend the identification arguments of Section \ref{sec: ID} to cover a wide range of network models, both semiparametric and nonparametric, beyond the model \eqref{eq: Y_ij def}.

First, recall that, as discussed in Section~\ref{ssec: Y star ID}, the error free outcomes are still identified in the most general analogue of \eqref{eq: Y_ij def} given by
\begin{align}
  \label{eq: general Y def}
  Y_{ij} = f(X_i,\xi_i,X_j,\xi_j) + \varepsilon_{ij}, \quad  \ex{\varepsilon_{ij}|X_i,\xi_i,X_j,\xi_j} = 0.
\end{align}
For example, we can formally establish identification of $Y_{ij}^* = f(X_i,\xi_i,X_j,\xi_j)$ by constructing a version of $\tilde Y_{ij}^*$ introduced in \eqref{eq: proper tilde Y}, treating $X$ as unobserved, and and showing its uniform consistency using the result of Theorem \ref{the: Y star UC}.

However, identification of $Y_{ij}^*$, \emph{the value of} $f(X_i,\xi_i,X_j,\xi_j)$, for any pair of agents $i$ and $j$ is fundamentally different from identification of \emph{function} $f$. Importantly, $Y_{ij}^*$ is not a causal object and cannot be directly employed in counterfactual analysis. Moreover, since $\xi$ is not observed, function $f$ or any features of it relevant for counterfactual analysis cannot be identified unless some additional structure is imposed on $f$.

For example, in Sections~\ref{sec: ID}-\ref{sec: LST}, we demonstrated that $\beta_0$ is identified and can be consistently estimated when $f(X_i,\xi_i,X_j,\xi_j) = W_{ij}'\beta_0 + g(\xi_i,\xi_j)$, which allows us to recover the ceteris paribus effect of $W_{ij}$ (or $X_{ij}$) on $Y_{ij}$ in this model. Below, we extend these identification results to more general forms of $f$ covering nonlinear and nonparametric models.

\subsubsection{Identification of the semiparametric single index model}
\label{sssec: single index}
One empirically relevant generalization of \eqref{eq: Y_ij def} is allowing $f$ to have a single index structure
\begin{align}
  \label{eq: single index}
  f(X_i,\xi_i,X_j,\xi_j) = F(W_{ij}' \beta_0 + g(\xi_i,\xi_j)),
\end{align}
where $F(\cdot)$ is a known invertible link function. Notice that the presence of the link function $F(\cdot)$ ensures that \eqref{eq: single index} is flexible enough to cover a wide range of the previously studied nonlinear network models. For example, consider the following network formation model
\begin{align*}
  Y_{ij} = \ind\{W_{ij}'\beta_0 + g(\xi_i,\xi_j) - U_{ij} \geqslant 0\},
\end{align*}
where $Y_{ij}$ is a binary variable indicating whether agents $i$ and $j$ are connected or not, and $U_{ij}$'s are iid draws from some distribution, e.g., logistic or $N(0,1)$. This model is covered by \eqref{eq: single index} with $F(\cdot)$ standing for 
the CDF of $U_{ij}$. If $F(\cdot) = \exp(\cdot)$, then \eqref{eq: single index} generalizes the dyadic Poisson regression model commonly used to analyze trade networks.

The previously developed arguments can be immediately applied to establish identification of $\beta_0$ and to construct an analogue estimator. First, note that since $Y_{ij}^* = F(W_{ij}' \beta_0 + g(\xi_i,\xi_j))$ is identified and $F(\cdot)$ is invertible, we can also identify
\begin{align*}
  \mathcal Y_{ij}^* \coloneqq W_{ij}'\beta_0 + g(\xi_i,\xi_j) = F^{-1} (Y_{ij}^*). 
\end{align*}
Next, since $\mathcal Y_{ij}^*$ are effectively observed, we are back in the additively separable setting previously studied in Sections~\ref{sec: ID}-\ref{sec: LST} with $\mathcal Y_{ij}^*$ replacing $Y_{ij}^*$, and all the relevant features of the model such as $\beta_0$ and the pair-specific fixed effects $g_{ij} = g(\xi_i,\xi_j)$ are identified for all $i$ and~$j$. 

In particular, we can still estimate $\beta_0$ as in \eqref{eq: hat beta def} with $\hat {\mathcal Y}_{ij}^*$ replacing $Y_{ij}$, using $\hat d_{ij}^2$ given by \eqref{eq: hat d star def} with $\hat {\mathcal Y}_{ij}^*$ replacing $\hat Y_{ij}^*$. In this case, the preliminary step involves denoising the observed outcomes by constructing $\hat Y_{ij}^*$ as before, and then computing $\hat{\mathcal Y}_{ij}^* = F^{-1} (\hat Y_{ij}^*)$.

Finally, we want to highlight the importance of identification of the pair-specific fixed effects $g_{ij} = g(\xi_i,\xi_j)$ in the model \eqref{eq: single index}. Since $F(\cdot)$ is potentially nonlinear, the knowledge of $\beta_0$ alone is not sufficient for identifying some policy relevant quantities such as partial effects. However, since the fixed effects $g_{ij}$'s are identified for all pairs of agents, we can also identify both pair-specific and average partial effects as well as other policy relevant counterfactuals.

\subsubsection{Identification of the nonparametric model}
In the previously considered models \eqref{eq: Y_ij def} and \eqref{eq: single index}, the contribution of observables is parameterized by $\beta_0 \in \mathbb R^p$, which was the primary object of interest in our analysis. In this section, we consider a nonparametric version of \eqref{eq: Y_ij def} with $f$ given by
\begin{align}
  \label{eq: simple separable}
  f(X_i,\xi_i,X_j,\xi_j) = \msh (X_i, X_j) + g(\xi_i,\xi_j),
\end{align}
where both $\mathscr h: \mathcal X \times \mathcal X \rightarrow \mathbb R$ and $g: \mathcal E \times \mathcal E \rightarrow \mathbb R$ are unknown symmetric functions.

In Sections~\ref{sec: ID}-\ref{sec: LST}, we considered a special case of \eqref{eq: simple separable} with $\msh (X_i,X_j) = w(X_i,X_j)' \beta_0$ and established identification of $\beta_0$. In this section, we will show that, in the studied setting, function $\msh(\cdot, \cdot)$ is nonparametrically identified. Importantly, nonparametric identification of $\msh (\cdot,\cdot)$ implies that the previously obtained identification results were not driven by the parametric restrictions or linearity imposed on $\msh (\cdot, \cdot)$, justifying our focus on the semiparametric model \eqref{eq: Y_ij def} as a practical approximation of the general nonparametrically identified model \eqref{eq: simple separable}.

We establish identification of $\msh (\cdot, \cdot)$ and $g_{ij}$ in \eqref{eq: simple separable} under the following assumption.

\begin{Ass}
  \label{ass: nonparam w}
  Suppose that \eqref{eq: simple separable} holds and
  \begin{enumerate}[(i)]
    \item \label{item: w norm} $\mathscr h: \mathcal X \times \mathcal X \rightarrow \mathbb R$ is a symmetric measurable function, and $\mathscr h(x,x) = 0$ for all $x \in \mathcal X$; 
    \item \label{item: w xi supp} For any $x, \tilde x \in \supp{X}$, there exists ${\mathcal E}_{x,\tilde x} \subseteq \mathcal E$ such that $\pr(\xi \in {\mathcal E}_{x,\tilde x}|X = x) > 0$ and $\pr(\xi \in {\mathcal E}_{x,\tilde x}|X = \tilde x) > 0$.
  \end{enumerate} 
\end{Ass}

\noindent
\textit{Discussion of Assumption \ref{ass: nonparam w}\eqref{item: w norm}.} The requirement $\mathscr h(x,x) = 0$ is a normalization. Indeed, since $g(\cdot,\cdot)$ and the dimension of $\xi$ are not specified, it is without loss of generality to let ${g(\xi_i,\xi_j) = \alpha_i + \alpha_j + \psi(\theta_i,\theta_j)}$, where $\psi(\cdot,\cdot)$ is symmetric, and $\xi = (\alpha, \theta')'$. Consider
\begin{align}
  \label{eq: unnorm nonparam}
  f(X_i,\xi_i,X_j,\xi_j) = \mathscr h(X_i,X_j) + \alpha_i + \alpha_j + \psi(\theta_i,\theta_j),
\end{align}
where $\mathscr h(\cdot,\cdot)$ is symmetric. Then, we can construct an observationally equivalent model with
\begin{align}
  \label{eq: norm nonparam}
  \tilde f (X_i,\tilde \xi_i,X_j,\tilde \xi_j) &= \tilde{\mathscr h}(X_i,X_j) + \tilde \alpha_i + \tilde \alpha_j + \psi(\theta_i,\theta_j),
  \\
  \tilde{\mathscr h}(X_i,X_j) &= \msh(X_i,X_j) - (\msh(X_i,X_i) + \msh(X_j,X_j))/2,\notag
\end{align}
where $\tilde \xi_i = (\tilde \alpha_i, \theta_i')'$ with $\tilde \alpha_i = \alpha_i + \msh (X_i,X_i)/2$. Note that $\tilde{\mathscr h}(\cdot,\cdot)$ is symmetric, continuous, and satisfies $\tilde{\mathscr h}(x,x) = 0$ for all $x \in \mathcal X$. Since $\tilde f(X_i,\tilde \xi_i, X_j, \tilde \xi_j) = f(X_i,\xi_i,X_j,\xi_j)$ for any pair of agents $i$ and $j$, the normalized model \eqref{eq: norm nonparam} is equivalent to the original model \eqref{eq: unnorm nonparam}. \QEDA

\begin{Rem}
  While the normalization introduced in Assumption \ref{ass: nonparam w}\eqref{item: w norm} is not the only possible one, it is natural for network models, especially when $\mathscr h(X_i,X_j)$ captures homophily based on observables (e.g., similar normalizations are also imposed in \citealp{Toth2017} and \citealp{gao2020nonparametric}).
\end{Rem}

First, we argue that $\msh (x, \tilde x)$ is identified for any fixed $x, \tilde x \in \supp{X}$. Specifically, fix $X_i = x$ and $X_j = \tilde x$, and consider
\begin{align}
  \mathscr d_{ij}^2 (\mu; x,\tilde x) \coloneqq   &\ex{\left(Y_{ik}^* - Y_{jk}^* + \mu \right)^2|(X_i,X_j,X_k) = (x, \tilde x, x),\xi_i,\xi_j} \notag \\
  &+ \ex{\left(Y_{ik}^* - Y_{jk}^* - \mu \right)^2|(X_i,X_j,X_k) = (x, \tilde x, \tilde x),\xi_i,\xi_j}, \notag \\
   \mathscr d_{ij}^2 (x,\tilde x) \coloneqq &\min_{\mu \in \mathbb R} d_{ij}^2 (\mu; x,\tilde x),  \quad \quad \quad  \mu^*_{ij}(x, \tilde x) \coloneqq \argmin_{\mu \in \mathbb R} \mathscr d_{ij}^2 (\mu; x,\tilde x). \label{eq: nonparam d}
\end{align}
Note that since $Y_{ik}^*$ and $Y_{jk}^*$ are identified in the general model \eqref{eq: general Y def}, $\msd_{ij}^2 (x, \tilde x)$ and $\mu_{ij}^*(x, \tilde x)$ are also identified.

$\msd_{ij}^2(x, \tilde x)$ is a nonparametric analogue of the previously considered pseudo-distance $d_{ij}^2$. Since we are interested in identifying $\msh (x, \tilde x)$, we additionally condition on $(X_i,X_j) = (x, \tilde x)$ and only consider $X_k = x$ and $X_k = \tilde x$. Notice that when $\xi_i = \xi_j$, then we have
\begin{align*}
  \msd_{ij}^2 (x,\tilde x) = \min_{\mu \in \mathbb R} \left( (-\msh(x,\tilde x) + \mu)^2 + (\msh(x,\tilde x) - \mu)^2 \right) = 0,
\end{align*}
where the minimum is achieved at $\mu_{ij}^*(x,\tilde x) = \msh (x, \tilde x)$. The following lemma establishes that the converse is also true: if $\msd_{ij}(x, \tilde x)^2 = 0$, then we also necessarily have $\mu_{ij}^*(x, \tilde x) = \msh (x, \tilde x)$.

\begin{Lem}
  \label{lem: nonparam ID}
  Suppose that Assumption \ref{ass: nonparam w} holds, and that the expectations in \eqref{eq: nonparam d} exist for any agents $i$ and $j$, and for any $x, \tilde x \in \mathcal X$. Then, for any $x, \tilde x \in \supp{X}$, $x \neq \tilde x$, $\msd_{ij}^2 (x,\tilde x) = 0$ implies that $\mu_{ij}^*(x, \tilde x) = \msh (x, \tilde x)$.
\end{Lem}

Lemma \ref{lem: nonparam ID} establishes identification of $\msh (x, \tilde x)$ by showing that $\msh (x, \tilde x) = \mu^*_{ij}(x, \tilde x)$ for all agents $i$ and $j$, with $X_i = x$ and $X_j = \tilde x$, satisfying $\msd_{ij} (x, \tilde x) = 0$; notice that Assumption~\ref{ass: nonparam w}\eqref{item: w xi supp} guarantees that such agents exist. Moreover, since the value of $\mu_{ij}^*(x,\tilde x)$ should be the same for all such $i$ and $j$ (and this is true for any fixed $x, \tilde x \in \supp{X}$), the nonparametric model~\eqref{eq: simple separable} as well as the previously considered models including~\eqref{eq: Y_ij def} are overidentified, and, in principle, can be falsified.

Next, notice that since $\msh (X_i,X_j)$ is identified for any pair of agents $i$ and $j$, we can also identify the pair-specific fixed effect as $g_{ij} = Y_{ij}^* - \msh (X_i,X_j)$. As a result, we conclude that both $\msh$ and the pair-specific fixed effects are nonparametrically identified in model \eqref{eq: simple separable}.

Finally, consider a nonlinear single index version of \eqref{eq: simple separable}
\begin{align}
  \label{eq: general additive}
  f(X_i,\xi_i,X_j,\xi_j) = F(\mathscr h(X_i,X_j) + g(\xi_i,\xi_j)),
\end{align}
where $F(\cdot)$ is a known and invertible link function. This model is a nonparametric version of \eqref{eq: single index} considered in Section~\ref{sssec: single index}. By first constructing $\mathcal Y_{ij}^* = F^{-1}(Y_{ij}^*) = \msh (X_i, X_j) + g(\xi_i,\xi_j)$ and then applying the results of this section, we conclude that $\msh $, the fixed effects $g_{ij}$'s, as well as both pair-specific and average partial effects are also identified in \eqref{eq: general additive}.

\subsection{Incorporating missing outcomes}
\label{ssec: missing}
From the beginning of Section \ref{sec: ID}, to simplify the exposition and facilitate the formal analysis, we assumed that $\{Y_{ij}\}_{i \neq j}$ are observed for all pairs of agents $i$ and $j$. While this assumption is standard in the network formation context, where the absence of an interaction between agents $i$ and $j$ is still recorded as observing $Y_{ij} = 0$, in many other applications, interaction outcomes are available only for a limited number of pairs of agents (e.g., in the matched employer-employee setting). Hence, it is important to discuss (i) how to properly adjust the constructed estimators to account for missing outcomes, and (ii) under which conditions the proposed method remains valid in this case. For simplicity, we will stick with considering an undirected network as before. We will discuss directed networks and general two-way settings covering the important matched employer-employee example in Section~\ref{ssec: two-way} below.

Let $D_{ij}$ be a binary variable such that $D_{ij} = 1$ if $Y_{ij}$ is observed and $D_{ij} = 0$ otherwise, with $D$ denoting the resulting adjacency matrix (by construction, $D_{ii} = 0$). Also, let ${\mathcal O_{ij} \coloneqq \{k: D_{ik} = D_{jk} = 1 \}}$ denote a set of agents $k$ such that $Y_{ik}$ and $Y_{jk}$ are observed.

To fix ideas, consider the homoskedastic estimator first. We adjust \eqref{eq: hat q_ij def} and \eqref{eq: hat beta def} as
\begin{gather}
  \hat q_{ij}^2 = \min_{\beta \in \mathcal B} \frac{1}{\abs{\mathcal O_{ij}}} \sum_{k \in \mathcal O_{ij}} (Y_{ik} - Y_{jk} - (W_{ik} - W_{jk})' \beta)^2, \label{eq: general q} \\
  \hat \beta =\left(\sum_{i < j} K\left(\frac{\hat d_{ij}^2}{h_n^2}\right) \sum_{k \in \mathcal O_{ij}} \Delta W_{ijk} \Delta W_{ijk}' \right)^{-1} \left(\sum_{i < j} K\left(\frac{\hat d_{ij}^2}{h_n^2}\right) \sum_{k \in \mathcal O_{ij}} \Delta W_{ijk} \Delta Y_{ijk} \right), \label{eq: general beta}
\end{gather}
where $\hat d_{ij}^2$ is still computed as in \eqref{eq: hat d def}, and $\abs{\mathcal O_{ij}}$ denotes the cardinality of $\mathcal O_{ij}$. Notice that to consistently estimate $q_{ij}^2$, we need $\abs{\mathcal O_{ij}} \rightarrow \infty$. Hence, in practice one may want to limit their attention to pairs of agents for which $\abs{\mathcal O_{ij}}$ is sufficiently large for $\hat q_{ij}^2$ to be informative.

Next, we want to discuss under which conditions the proposed method, with appropriate modifications as described above, remains valid when some interaction outcomes are missing. First, the selection mechanism needs to be exogenous conditional on the observed and unobserved characteristics of agents $\{(X_i,\xi_i)\}_{i=1}^n$, meaning that the $D$ should be (conditionally) independent of the errors $\{\varepsilon_{ij}\}$.\footnote{This assumption is satisfied if agents form connections based on $\{(X_i,\xi_i)\}_{i=1}^n$ but not on the idiosyncratic errors. For example, in a structural model, it can be rationalized if $\varepsilon_{ij}$'s are drawn after the network is formed.} This assumption is standard in the network regression literature including the seminal AKM model of \citet{Abowd1999} and subsequent work, assuming exogeneity of the employer-employee network conditional on the firms' and workers' \emph{additive} fixed effect (and their observed characteristics). While this assumption is widely used in the AKM literature, it is also often criticized for severely restricting the network formation process. We want to stress that, since we allow for a much more general form of unobserved heterogeneity than the additive AKM model, the network exogeneity assumption is substantially less restrictive in our setting. Specifically, since we do not specify the dimensionality of the fixed effects and their role in the model (e.g., we allow for complementarity between the firm and worker fixed effects), our framework allows for a much broader class of selection mechanism compatible with the exogeneity assumption. Thus, rather than seeing network endogeneity as a potential threat to the validity of our approach, we consider it to be a new tool addressing this concern.

Second, for $\hat \beta$ in \eqref{eq: general beta} to be consistent, we need to have a growing number of pairs $i$ and $j$, for which $\abs{\mathcal{O}_{ij}} \rightarrow \infty$ as $n \rightarrow \infty$. Specifically, the requirement $\abs{\mathcal{O}_{ij}} \rightarrow \infty$ ensures that we can consistently estimate $\hat d_{ij}^2$ for a growing group of agents. At the same time, a growing pool of potential matches allows us to find increasingly similar agents controlling the bias of~$\hat \beta$ and ensuring its consistency. Notice that this requirement  still allows the network to be sparse, and that its adequacy can be evaluated in a given application. Moreover, in Section~\ref{ssec: two-way}, we will argue that this requirement becomes even less restrictive in general two-way settings and is plausible for many matched employer-employee data sets.

In the general heteroskedastic case, the first estimation step is to construct $\hat Y_{ij}^*$. In fact, recent developments in the matrix completion literature allow one to consistently estimate $Y^*$ (in terms of the MSE), even when the observed matrix $Y$ is sparse (e.g., \citealp*{Chatterjee2015,Klopp2017,li2019nearest}). Once $\hat Y_{ij}^*$ are constructed (for example, using one of the already developed matrix completion techniques), the rest of the estimation procedure remains the same. We provide an appropriate modification of the previously used estimator of $\hat Y_{ij}^*$ and discuss estimation of $\beta_0$ in more detail in Appendix~\ref{sec: missing data imputation}.

\subsection{Extension to directed networks and two-way models}
\label{ssec: two-way}
Finally, the proposed estimation procedure can also be generalized to cover directed networks and, more generally, two-way models. Specifically, consider a general interaction model
\begin{align*}
	Y_{ik} = W_{ik}' \beta_0 + g(\xi_i,\eta_k) + \varepsilon_{ik},
\end{align*}
where $i \in \mathcal I$ and $k \in \mathcal K$ index senders and receivers, and $\xi_i$ and $\eta_k$ denote the sender and receiver fixed effects. As in Section~\ref{ssec: missing}, we can also allow for missing interactions, with $Y_{ik}$ observed whenever $D_{ik} = 1$.

As before, our identification and estimation strategies are based on finding agents with similar values of the fixed effects. However, the considered interaction model consists of two types of agents, senders and receivers (e.g., firms and workers). As a result, we have the flexibility to decide whether we want to match senders or receivers depending on the context. For example, consider the sender-to-sender approach. To fix ideas, we also focus on the homoskedastic setting (the general heteroskedastic estimator can be constructed in a similar fashion). In this case, the sender-to-sender estimator of $\beta_0$ has exactly the same form as in \eqref{eq: general q}-\eqref{eq: general beta} previously provided in Section~\ref{ssec: missing}

As discussed in the previous section, for the sender-to-sender estimator to be consistent, we need a growing pool of senders that we can potentially match satisfying the same requirement $\abs{\mathcal O_{ij}} \rightarrow \infty$, so we can consistently determine if senders $i$ and $j$ are similar in terms of $\xi$ or not. Also notice that in this case receivers are allowed to participate only in a few interactions. For example, this suggests that in the matched employer-employee setting, the firm-to-firm approach can be appropriate even when the workers' mobility is limited, so long as we have a sufficient number of pairs of firms with a sufficient number of workers moving from one of them to another.

\section{Numerical Evidence}
\label{sec: numerical}
\subsection{Numerical experiment in a homoskedastic model}
\label{ssec: MC} 
In this section, we illustrate the finite sample properties of the proposed estimators. Specifically, we consider the following homoskedastic variation of \eqref{eq: Y_ij def}:
\begin{align}
	\label{eq: MC homo}
	Y_{ij} = (X_i - X_j)^2 \beta_0 - (\xi_i - \xi_j)^2 + \varepsilon_{ij}, \quad
  \begin{pmatrix} X_i \\ \xi_i  \end{pmatrix} \sim N \left( \begin{pmatrix} 0\\ 0  \end{pmatrix} , \begin{pmatrix} 1 & \rho \\ \rho & 1  \end{pmatrix} \right),
\end{align}
where $\{\varepsilon_{ij}\}_{i < j}$ are independent draws from $N(0,1)$. The true value of the parameter of interest is $\beta_0 = -1$, so the considered model features homophily based on both $X$ and $\xi$.

We study the performance of the following estimators. The first estimator $\hat \beta_{\text{FE}}$ is produced by the standard linear regression with additive fixed effects. The second estimator $\hat \beta$ is the kernel based estimator \eqref{eq: hat beta def} with $\hat d_{ij}^2$ computed as in \eqref{eq: hat d def} and using the Epanechnikov kernel.\footnote{The reported results are robust to the kernel choice.} We choose $h_n^2 = 0.9 \min \left\{\hat \sigma_{\hat d^2}, \text{IQR}_{\hat d^2}/1.349\right\} \binom{n}{2}^{-1/5}$
following the standard (kernel density estimation) rule of thumb. Here $\hat \sigma_{\hat d^2}$ and $\text{IQR}_{\hat d^2}$ stand for the standard deviation and the interquartile range of the estimated pseudo-distances $\{\hat d_{ij}^2\}_{i < j}$, and $\binom{n}{2}$ corresponds to the number of the estimated pseudo-distances.\footnote{Note that since the kernel weights in \eqref{eq: hat beta def} are $K(\frac{\hat d_{ij}^2}{h_n^2})$, the ``effective'' kernel density estimation bandwidth applied to $\{\hat d_{ij}^2\}_{i < j}$ is $h_n^2$, not $h_n$.} Finally, we also compute the 1 nearest neighbor pairwise-difference estimator $\hat \beta_{\text{NN1}}$, which, instead of using kernel weights as in \eqref{eq: hat beta def}, matches every unit $i$ with exactly one unit $j$ closest to it in terms of $\hat d_{ij}^2$.  

We simulate the model \eqref{eq: MC homo} for $n \in \{30, 50, 100\}$ and $\rho \in \{0, 0.3 , 0.5, 0.7\}$. The simulated finite sample properties of the considered estimators are reported in Table \ref{tab: MC 1} below. The number of replications is 10,000. The naive estimator $\hat \beta_{\text{FE}}$ is biased whenever the observed and unobserved characteristics of agents are correlated. The magnitude of this bias increases rapidly as $\rho$ grows. The proposed estimators $\hat \beta$ and $\hat \beta_{\text{NN1}}$ effectively remove the bias even in networks of a moderate size with $n = 30$. Notice that the magnitudes of the bias for $\hat \beta$ and $\hat \beta_{\text{NN}1}$ are approximately the same but the kernel based estimator $\hat \beta$ is consistently less dispersed. This might suggest that in the studied setting, the 1 nearest neighbor estimator $\hat \beta_{\text{NN}1}$ tends to undersmooth. Finally, notice that the proposed estimators $\hat \beta$ and $\hat \beta_{\text{NN}1}$ dominate the naive estimator $\hat \beta_{\text{FE}}$ not only in terms of the bias but also in terms of the standard deviation/IQR even when $\rho = 0$, i.e., when $\hat \beta_{\text{FE}}$ is consistent. Indeed, when the fixed effects contribution to the variability in $Y_{ij}$ is large, controlling for the unobservables (as both $\hat \beta$ and $\hat \beta_{\text{NN1}}$ do by differencing them out) can substantially improve precision even at the cost of significantly reducing the effective sample size.
 
\begin{table}[h!]
\begin{center}
\begin{footnotesize}
\begin{threeparttable}
\caption{Simulation results for model \eqref{eq: MC homo}}
\label{tab: MC 1}
\begin{tabular}{ c| c c c| c c c| c c c| c c c }
& \multicolumn{3}{c|}{Bias} & \multicolumn{3}{c|}{Med Bias} & \multicolumn{3}{c|}{Std. Dev.} & \multicolumn{3}{c}{IQR/1.349} \\
$\rho$ & $\hat \beta_{\text{FE}}$ & $\hat \beta$ & $\hat \beta_{\text{NN1}}$ & $\hat \beta_{\text{FE}}$ & $\hat \beta$ & $\hat \beta_{\text{NN1}}$  & $\hat \beta_{\text{FE}}$ & $\hat \beta$ & $\hat \beta_{\text{NN1}}$  & $\hat \beta_{\text{FE}}$ & $\hat \beta$ & $\hat \beta_{\text{NN1}}$ \\
\midrule
\multicolumn{13}{c}{$n=30$}\\
\midrule
0.0&0.000&-0.001&0.010&0.010&0.000&0.006&0.063&0.028&0.050&0.044&0.025&0.036\\
0.3&-0.090&-0.006&0.005&-0.061&-0.005&0.004&0.124&0.033&0.060&0.106&0.029&0.044\\
0.5&-0.251&-0.018&-0.009&-0.225&-0.016&-0.006&0.173&0.045&0.077&0.165&0.038&0.059\\
0.7&-0.491&-0.052&-0.058&-0.473&-0.048&-0.048&0.196&0.080&0.109&0.191&0.063&0.092\\
\midrule
\multicolumn{13}{c}{$n=50$}\\
\midrule
0.0&-0.000&-0.000&0.005&0.006&-0.000&0.003&0.035&0.015&0.029&0.026&0.014&0.021\\
0.3&-0.090&-0.004&0.004&-0.072&-0.004&0.002&0.090&0.018&0.036&0.085&0.016&0.027\\
0.5&-0.251&-0.013&-0.005&-0.235&-0.012&-0.004&0.130&0.024&0.046&0.128&0.022&0.036\\
0.7&-0.491&-0.036&-0.038&-0.481&-0.034&-0.032&0.148&0.042&0.068&0.147&0.037&0.057\\
\midrule
\multicolumn{13}{c}{$n=100$}\\
\midrule
0.0&-0.000&-0.000&0.003&0.003&-0.000&0.002&0.017&0.007&0.014&0.012&0.007&0.010\\
0.3&-0.091&-0.003&0.002&-0.082&-0.002&0.002&0.061&0.008&0.018&0.058&0.008&0.014\\
0.5&-0.251&-0.008&-0.002&-0.243&-0.007&-0.002&0.089&0.011&0.024&0.087&0.010&0.019\\
0.7&-0.491&-0.021&-0.020&-0.486&-0.021&-0.017&0.102&0.019&0.036&0.100&0.018&0.030\\
\bottomrule
\bottomrule
\end{tabular}
\begin{tablenotes}
\item \scriptsize This table reports the simulated bias, median bias, standard deviation, and interquartile range (divided by 1.349) for the additive fixed effects estimator $\hat \beta_{\text{FE}}$, the kernel estimator $\hat \beta$, and the 1 nearest neighbor estimator $\hat \beta_{\text{NN1}}$. The results are presented for different values $n$ (network size) and $\rho$ (correlation between $X_i$ and $\xi_i$). The number of replications is 10,000.
\end{tablenotes}
\end{threeparttable}
\end{footnotesize}
\end{center}
\end{table}

\subsection{Empirical illustration: homophily in online social networks}
\label{ssec: empirical}
In this section, we illustrate the usefulness of our method in the context of estimating homophily in online social networks using the Facebook 100 data set \citep{traud2012social}. This dataset contains Facebook friendship network as well as nodal covariates (gender, major, graduation year, etc.) collected at 100 colleges and universities in the US in 2005.

Specifically, we consider the following logistic network formation model
\begin{align}
  \label{eq: emp model}
  Y_{ij} = \ind\{\ind\{X_i = X_j\}\beta_0 + g(\xi_i,\xi_j) - U_{ij} \geqslant 0\},
\end{align}
where $Y_{ij}$ is a binary variable indicating whether students $i$ and $j$ are friend or not, $X_i$ is a gender dummy, and $\{U_{ij}\}_{i < j}$ are iid draws from the logistic distribution. In this model, $\beta_0$ captures gender homophily. Notice that \eqref{eq: emp model} can be represented in the regression form as
\begin{align*}
  Y_{ij} = \Lambda(\ind\{X_i = X_j\}\beta_0 + g(\xi_i,\xi_j)) + \varepsilon_{ij},
\end{align*}
where $\Lambda(\cdot)$ stands for the logistic CDF. Note that in the studied network formation model the conditional mean is \emph{nonlinear} and the errors $\varepsilon_{ij}$'s are \emph{heteroskedastic}.

We estimate $\beta_0$ using the following three estimators. The first is $\hat \beta_{\text{MLE}}$, the naive MLE estimator that ignores unobserved heterogeneity. The second is $\hat \beta_{\text{TL}}$, the tetra-logit estimator introduced in \citet{Graham2017} allowing for additive fixed effects, i.e., assuming that ${g(\xi_i,\xi_j) = \xi_i + \xi_j}$. Finally, we also construct the kernel estimator $\hat \beta$ following the procedure described in Section~\ref{sssec: single index}. Specifically, we use $n_i \approx 0.5 (n \ln n)^{1/2}$ for all $i$ for constructing $Y_{ij}^*$, and we choose $K$ and $h_n^2$ as described in Section~\ref{ssec: MC}.

We report results for Princeton University, the class of 2004 (the results are qualitatively similar across the universities and cohorts). This network consists of 541 students with the mean degree of 37.6, i.e., on average the students in this sample are connected with about 7\% of the whole class, so the studied network is characteristically sparse. 

The results are reported in Table~\ref{tab: empirical big} below. Both the MLE and tetra-logit homophily estimates are much higher compared to the one produced by $\hat \beta$. This is in line with one of the well known perils of estimating homophily effects: naive methods are likely to overestimate homophily associated with observables when agents also exhibit homophily based on their latent characteristics. For example, if students meet their friends in classrooms and gender is predictive of their choices of classes, naive methods might misattribute this effect to gender homophily.

However, since we do not have a confidence interval available for $\hat \beta$, it is not immediately clear if the difference in the produced estimates is systematic or simply attributed to the large sampling uncertainty. In order to address this concern and study the performance of our method in an empirically relevant setting, we perform the following numerical experiment designed to mimic the studied application. Specifically, we consider a version of \eqref{eq: emp model}
\begin{gather*}
    Y_{ij} = \ind\{\ind\{X_i = X_j\}\beta_0 + \alpha_0 - \kappa_0 \abs{\xi_i - \xi_j}  - U_{ij} \geqslant 0\},\\
    X_i \sim Bernoulli(0.5), \quad \xi_i = \pi_0 X_i + V_i, \quad V_i \sim U([-1,1]),
\end{gather*}
where $(X_i,\xi_i,V_i)$ are iid. We choose $(\alpha_0, \beta_0, \kappa_0, \pi_0) = (-1.5, 0.1, 2, 0.2)$ and $n=550$ to match the main features of the data including the average node degree, its standard deviation, as well as the estimates produced by the methods. Note that this model features homophily based on both $X$ and $\xi$, and ignoring the latter would result in overestimation of $\beta_0$.

The simulation results are also reported in Table~\ref{tab: empirical big} below. First, as expected, $\hat \beta_{\text{MLE}}$ and $\hat \beta_{\text{TL}}$ indeed severely overestimate $\beta_0$, and their biases are substantially larger than their standard deviations. At the same time our estimator $\hat \beta$ does not suffer from this bias and correctly estimates the true effect. Moreover, we also find that, in the studied setting, its standard deviation is comparable to the standard deviations of the other estimators.

In summation, this numerical experiment demonstrates that our method can perform well in an empirically relevant setting featuring nonlinearity of the regression function, heteroskedasticity of the errors, and sparsity representative of real world social networks. Its results also support our empirical finding documenting that the standard approaches overestimate the gender homophily effects and suggest that the difference in the estimates is systematic and not likely to be (entirely) attributed to the sampling variability.

\begin{table}[h!]
\begin{center}
\begin{threeparttable}
\caption{{Actual data and empirically calibrated simulation results}}
\label{tab: empirical big}
\begin{tabular}{ l| c c | c c c c }
& \multicolumn{2}{c|}{Actual Data} & \multicolumn{4}{c}{Simulation Results}\\
& Estimate & 95\% CI& Bias & Med. Bias & Std. Dev. & IQR/1.349 \\
\midrule
      $\hat \beta_{\text{MLE}}$ & 0.1792& [0.1387, 0.2196] &0.0814&    0.0802&    0.0297&    0.0307\\
      $\hat \beta_{\text{TL}}$ & 0.1862 & [0.1419, 0.2306] &0.0849&    0.0830&    0.0301&    0.0310\\
      $\hat \beta$ &  0.1106 &  &-0.0060&   -0.0060&    0.0321&    0.0339\\
\bottomrule
\bottomrule
\end{tabular}
\begin{tablenotes}
\item \footnotesize{This table reports the estimates (and whenever available confidence intervals) of the homophily parameter $\beta_0$ in model~\eqref{eq: emp model} for the naive MLE estimator $\hat \beta_{\text{MLE}}$, the tetra-logit estimator $\hat \beta_{\text{TL}}$, and the kernel estimator $\hat \beta$ described in Section~\ref{sssec: single index}. It also reports the simulated bias, median bias, standard deviation, and interquartile range (divided by 1.349) for the same estimators in the considered empirically calibrated numerical experiment. The number of replications is 10,000.}
\end{tablenotes}
\end{threeparttable}
\end{center}
\end{table}

\bibliographystyle{apalike}
\bibliography{library}

\begin{thebibliography}{}

\bibitem[Ahn and Powell, 1993]{Ahn1993}
Ahn, H. and Powell, J.~L. (1993).
\newblock {Semiparametric estimation of censored selection models with a
  nonparametric selection mechanism}.
\newblock {\em Journal of Econometrics}, 58(1-2):3--29.

\bibitem[Arcones, 1995]{arcones1995bernstein}
Arcones, M.~A. (1995).
\newblock A bernstein-type inequality for u-statistics and u-processes.
\newblock {\em Statistics \& Probability Letters}, 22(3):239--247.

\bibitem[Bennett, 1962]{bennett1962probability}
Bennett, G. (1962).
\newblock Probability inequalities for the sum of independent random variables.
\newblock {\em Journal of the American Statistical Association},
  57(297):33--45.

\bibitem[Boucheron et~al., 2013]{boucheron2013concentration}
Boucheron, S., Lugosi, G., and Massart, P. (2013).
\newblock {\em Concentration inequalities: A nonasymptotic theory of
  independence}.
\newblock Oxford university press.

\end{thebibliography}

\newpage

\appendix
\setcounter{page}{1}
\renewcommand{\thepage}{A.\arabic{page}}

\begin{large}
    \begin{center}
       Online Supplementary Appendix to ``Identification and Estimation\\ of Network Models with Nonparametric Unobserved Heterogeneity''
    \end{center}
\end{large}

\begin{large}
    \begin{center}
        Andrei Zeleneev
    \end{center}
\end{large}

\begin{large}
    \begin{center}
        \today
    \end{center}
\end{large}

\section{Proofs}

\subsection{Proofs of the results of Section \ref{ssec: beta rate}}

\subsubsection{Auxiliary lemmas}

\begin{AppLem}
  \label{lem: A aux kernel}
    Let $c_{ij} \coloneqq c(X_i,X_j,\xi_i)$, where $c(X_i,X_j,\xi_i)$ is defined in Assumption~\ref{ass: beta ID}\eqref{item: d* approximation}. Then, under hypotheses of Theorem~\ref{the: kernel rate}, there exists $\alpha > 0$ such that, for sufficiently large $n$, (i) $\abs{\xi_j - \xi_i} > \alpha h_n$ implies that  $K\left({d^2(X_i,X_j,\xi_i,\xi_j)}/{h_n^2}\right) = 0$ and $K\left({c_{ij} (\xi_j - \xi_i)^2}/{h_n^2}\right) = 0$ with probability one; (ii) $\sum_{i < j} K({\hat d^2_{ij}}/{h_n^2})^2 \ind\{\abs{\xi_j - \xi_i} > \alpha h_n\} = 0$ with probability approaching one; (iii) $\abs{r_K(X_i,X_j,\xi_i,\xi_j;h_n)} \leqslant C h_n \ind\{\abs{\xi_j - \xi_i} \leqslant \alpha h_n\}$ a.s. for some $C > 0$, where  $r_K(X_i,X_j,\xi_i,\xi_j;h_n) \coloneqq K({d^2(X_i,X_j,\xi_i,\xi_j)}/{h_n^2}) - K({c_{ij} (\xi_j - \xi_i)^2}/{h_n^2})$.
\end{AppLem}

\begin{proof}[Proof of Lemma \ref{lem: A aux kernel}]
  Note that Assumption \ref{ass: beta ID}\eqref{item: d* approximation} guarantees that there exists some $\delta_0 > 0$ and $c^* \in (0, \underline c)$ such that $d^2(X_i,X_j,\xi_i,\xi_j) \geqslant c^* (\xi_j - \xi_i)^2$ a.s. for $\abs{\xi_j - \xi_i} \leqslant \delta_0$. This implies that there exists $\alpha = 1/\sqrt{c^*}$ such that $d^2(X_i,X_j,\xi_i,\xi_j) > h_n^2$ whenever $\alpha h_n < \abs{\xi_j - \xi_i} \leqslant \delta_0$ (notice that $\alpha h_n < \delta_0$ for large enough $n$). At the same, by Assumption~\ref{ass: beta ID}\eqref{item: d* xi ID}, time, for large enough $n$, $d^2(X_i,X_j,\xi_i,\xi_j) > h_n^2$ a.s. for $\abs{\xi_i - \xi_j} \geqslant \delta_0$. Hence, we conclude that $\abs{\xi_j - \xi_i} > \alpha h_n$ implies $d^2(X_i,X_j,\xi_i,\xi_j) > h_n^2$ and, by Assumption~\ref{ass: K and h}\eqref{item: kernel}, $K\left({d^2(X_i,X_j,\xi_i,\xi_j)}/{h_n^2}\right) = 0$. Similarly, since $c^* < \underline c$, $\abs{\xi_j - \xi_i} > \alpha h_n$ also implies $c_{ij} (\xi_j - \xi_j)^2 > h_n^2$ and $K\left({c_{ij} (\xi_j - \xi_i)^2}/{h_n^2}\right) = 0$ a.s. This completes the proof of part (i).

  Clearly, we can also choose $\alpha$ such that $d^2(X_i,X_j,\xi_i,\xi_j) > c h_n^2$ for some $c > 1$ whenever $\abs{\xi_j - \xi_i} > \alpha h_n$. Consequently, for all pairs of agents $i$ and $j$ satisfying $\abs{\xi_j - \xi_i} > \alpha h_n$, we have ${d_{ij}^2}/{h_n^2} > c > 1$. Using \eqref{eq: R_n def} and Assumption~\ref{ass: K and h}\eqref{item: h_n}, we also conclude that, for all pairs of agents $i$ and $j$ satisfying that $\abs{\xi_j - \xi_i} > \alpha h_n$, with probability approaching one we also have ${\hat d^2_{ij}}/{h_n^2} > 1$. Combining this with Assumption~\ref{ass: K and h}\eqref{item: kernel} completes the proof of part (ii).

  Finally, notice that using the result of part (i), we obtain
  \begin{align*}
    r_K(X_i,X_j,\xi_i,\xi_j;h_n) = \left(K\left(\frac{d^2(X_i,X_j,\xi_i,\xi_j)}{h_n^2}\right) - K\left(\frac{c_{ij} (\xi_j - \xi_i)^2}{h_n^2}\right) \right) \ind\{|\xi_j - \xi_i| \leqslant \alpha h_n\}.
  \end{align*}
  Combining this with Assumptions \ref{ass: beta ID}\eqref{item: d* approximation} and \ref{ass: K and h}\eqref{item: kernel} delivers the required bound.
\end{proof}

\begin{AppLem}
  \label{lem: A aux O(1)}
  Suppose the hypotheses of Theorem~\ref{the: kernel rate} are satisfied. Then, for any $\alpha > 0$, we have: (i) $h_n^{-1} \ex{\ind\{\abs{\xi_j - \xi_i} \leqslant \alpha h_n\}|\xi_i} \leqslant C_\alpha$ a.s. and, hence, $h_n^{-1} \ex{\ind\{\abs{\xi_j - \xi_i} \leqslant \alpha h_n\}} \leqslant C_\alpha$ for some $C_\alpha > 0$; (ii) $\binom{n}{2}^{-1} h_n^{-1} \sum_{i < j} \ind\{\abs{\xi_j - \xi_i} \leqslant \alpha h_n\} = O_p(1)$.
\end{AppLem}

\begin{proof}[Proof of Lemma \ref{lem: A aux O(1)}]
  Let $q_{n,ij} \coloneqq h_n^{-1} \ind\{\abs{\xi_j - \xi_i} \leqslant \alpha h_n\}$. Using  Assumption~\ref{ass: xi dist}\eqref{item: xi pdf}  $\ex{q_{n,ij}|\xi_i} = h_n^{-1} \int \ind\{\abs{\xi_j - \xi_i} \leqslant \alpha h_n\} f_{\xi} (\xi_j) d \xi_j \leqslant 2 \alpha \overline f_\xi = C_\alpha$, which proves the first part. Similarly, $\ex{q_{n,ij}^2} \leqslant \frac{2 \alpha \overline f_\xi}{h_n} = O(h_n^{-1}) = o(n)$, where the last equality uses Assumption \ref{ass: K and h}\eqref{item: h_n}. Invoking Lemma A.3 of \citetsupp{Ahn1993}, $\binom{n}{2}^{-1} \sum_{i < j} q_{n,ij} = \ex{q_{n,ij}} + o_p(1) = O_p (1)$, where the last equality also uses $\ex{q_{n,ij}} \leqslant C_\alpha$.
\end{proof}

\subsubsection{Proof of Theorem~\ref{the: kernel rate}}
First, notice that $\hat \beta - \beta_0 = \hat A_n^{-1} (\hat B_n + \hat C_n)$, where
\begin{align*}
  \hat A_n &\coloneqq \binom{n}{2}^{-1} h_n^{-1} \sum_{i < j} K\left(\frac{\hat d_{ij}^2}{h_n^2}\right) \frac{1}{n-2} \sum_{k \neq i,j} (W_{ik} - W_{jk}) (W_{ik} - W_{jk})', \\
  \hat B_n &\coloneqq \binom{n}{2}^{-1} h_n^{-1} \sum_{i < j} K\left(\frac{\hat d_{ij}^2}{h_n^2}\right) \frac{1}{n-2} \sum_{k \neq i,j} (W_{ik} - W_{jk}) (g(\xi_i,\xi_k) - g(\xi_j,\xi_k)), \\
  \hat C_n &\coloneqq \binom{n}{2}^{-1} h_n^{-1} \sum_{i < j} K\left(\frac{\hat d_{ij}^2}{h_n^2}\right) \frac{1}{n-2} \sum_{k \neq i,j} (W_{ik} - W_{jk}) (\varepsilon_{ik} - \varepsilon_{jk}).
\end{align*}
We proof Theorem~\ref{the: kernel rate} by inspecting the asymptotic behavior of $\hat A_n$, $\hat B_n$, and $\hat C_n$: below we show (i) $\hat A_n = A + o_p(1)$ for some symmetric $A$ satisfying $\lambda_{min} (A) > C > 0$; (ii) $\hat B_n  = O_p(h_n^2 + \frac{R_n^{-1}}{h_n} + n^{-1})$; (iii) $\hat C_n = O_p(\frac{R_n^{-1}}{h_n^2} \left(\frac{\ln n}{n}\right)^{1/2} + n^{-1} )$, together delivering the desired result.

For simplicity of exposition and consistency with the assumptions provided in Section~\ref{ssec: beta rate}, we prove the result for $d_\xi = 1$. For $d_\xi > 1$ and with the assumptions appropriately modified as discussed in Remark~\ref{rem: multivariate xi}, the proof of the theorem remains essentially the same with the exceptions that (i) $\hat A_n$, $\hat B_n$, and $\hat C_n$ should be normalized by $h_n^{-d_\xi}$ instead of $h_n^{-1}$, and (ii) all the quantities in the statements of Lemma~\ref{lem: A aux O(1)} should be normalized by $h_n^{-d_\xi}$ instead of $h_n^{-1}$.

\subsubsection{Demonstrating consistency of $\hat A_n$}
  In this section, we argue that $\hat A_n = A + o_p(1)$, where $A$ is a symmetric invertible matrix.

  In particular, first, we argue that $\hat A_n - A_n = o_p(1)$, where
  \begin{align*}
    A_n \coloneqq \binom{n}{2}^{-1} h_n^{-1} \sum_{i < j} K\left(\frac{d_{ij}^2}{h_n^2}\right) \frac{1}{n-2} \sum_{k \neq i,j} (W_{ik} - W_{jk}) (W_{ik} - W_{jk})'.
  \end{align*}
  Then
  \begin{align*}
    \hat A_n - A_n = \binom{n}{2}^{-1} h_n^{-1} \sum_{i < j} \left(K\left(\frac{\hat d_{ij}^2}{h_n^2}\right) - K\left(\frac{d_{ij}^2}{h_n^2}\right)\right) \frac{1}{n-2} \sum_{k \neq i,j} (W_{ik} - W_{jk}) (W_{ik} - W_{jk})'.
  \end{align*}
  First, notice that $\norm{\frac{1}{n-2} \sum_{k \neq i,j} (W_{ik} - W_{jk}) (W_{ik} - W_{jk})'}$ is uniformly bounded thanks to Assumption \ref{ass: basic}\eqref{item: w}. Then, using Assumption \ref{ass: K and h}\eqref{item: kernel} and the result of Lemma \ref{lem: A aux kernel}(ii), we have.
  \begin{align*}
    \binom{n}{2}^{-1} h_n^{-1} \sum_{i < j} \left|K\left(\frac{\hat d_{ij}^2}{h_n^2}\right) - K\left(\frac{d_{ij}^2}{h_n^2}\right)\right| &\leqslant \overline K' \frac{\max_{i \neq j} \abs{\hat d_{ij}^2 - d_{ij}^2}}{h_n^2} \binom{n}{2}^{-1} h_n^{-1} \sum_{i < j} \ind\{\abs{\xi_j - \xi_i} \leqslant \alpha h_n\} \\
    &=  O_p \left( {R_n^{-1}}/{h_n^2} \right) = o_p(1),
  \end{align*}
  where the first equality uses \eqref{eq: R_n def} and Lemma \ref{lem: A aux kernel}(ii), and the second equality is due to Assumption \ref{ass: K and h}\eqref{item: h_n}. This completes the proof of $\hat A_n - A_n = o_p(1)$.

  Next, notice that $A_n = \frac{1}{n(n-1)(n-2)} \sum_{i \neq j \neq k} \zeta_n(Z_i,Z_j,Z_k)$, where
  \begin{align*}
    \zeta_n (Z_i,Z_j,Z_k) \coloneqq h_n^{-1} K\left(\frac{d^2(Z_i,Z_j)}{h_n^2}\right) (w(X_j,X_k) - w(X_i,X_k)) (w(X_j,X_k) - w(X_i,X_k))'.
  \end{align*}
  Note that $\zeta_n$ is symmetric in its first two arguments and, consequently, it can be symmetrized with
  \begin{align*}
    p_n (Z_i,Z_j,Z_k) \coloneqq \frac{1}{3} \left(\zeta_n(Z_i,Z_j,Z_k) + \zeta_n(Z_k,Z_i,Z_j) + \zeta_n(Z_j,Z_k,Z_i)\right).
  \end{align*}
  Then $A_n$ is a third order U-statistic with kernel $p_n$. First, we want to show that $\ex{\norm{p_n(Z_i,Z_j,Z_k)}^2} = o(n)$. Indeed, again since $W$ is bounded, using Assumption \ref{ass: K and h}\eqref{item: kernel} and the result of Lemma \ref{lem: A aux kernel}(i),
  \begin{align*}
    \ex{\norm{\zeta_n(Z_i,Z_j,Z_k)}^2} \leqslant C h_n^{-2} \ex{\ind\{\abs{\xi_j - \xi_i} \leqslant \alpha h_n\}} = o(h_n^{-1}) = o(n)
  \end{align*}
  where the first equality uses Lemma \ref{lem: A aux O(1)}, and the second uses Assumption \ref{ass: K and h}\eqref{item: h_n}. Similarly, $\ex{\norm{p_n(Z_i,Z_j,Z_k)}^2} = O(h_n^{-1}) = o(n)$, and by Lemma A.3 of \citesupp{Ahn1993}, we obtain $A_n = \ex{p_n(Z_i,Z_j,Z_k)} + o_p(1) = \ex{\zeta_n(Z_i,Z_j,Z_k)} + o_p(1)$.   Since we have previously established that $\hat A_n = A_n + o_p(1)$, we also get $\hat A_n = \ex{\zeta_n(Z_i,Z_j,Z_k)} + o_p(1)$.

  The rest of the proof deals with computing $\ex{\zeta_n(Z_i,Z_j,Z_k)}$. First, note
  \begin{align*}
    \ex{\zeta_n(Z_i,Z_j,Z_k)|X_i,X_j} = \ex{h_n^{-1} K\left(\frac{d^2(X_i,X_j,\xi_i,\xi_j)}{h_n^2}\right)|X_i,X_j} \mathcal C(X_i,X_j),
  \end{align*}
  where $\mathcal C(X_i,X_j)$ is as defined in \eqref{eq: C x1 x2 def}. Using the result of Lemma \ref{lem: A aux kernel}(iii) and Assumption \ref{ass: xi dist}\eqref{item: xi pdf}, we obtain
  \begin{align*}
    I(X_i,X_j,\xi_i;h_n) &\coloneqq \ex{h_n^{-1} K\left(\frac{d^2(X_i,X_j,\xi_i,\xi_j)}{h_n^2}\right)|X_i,X_j,\xi_i} = I_1(X_i,X_j,\xi_i;h_n) + O(h_n), \\
    I_1(X_i,X_j,\xi_i;h_n) &\coloneqq \int h_n^{-1} K \left(\frac{c_{ij} (\xi_j - \xi_i)^2}{h_n^2}\right) f_{\xi|X}(\xi_j;X_j)d\xi_j = \frac{1}{\sqrt{c_{ij}}} \int K(u^2) f_{\xi|X}\left(\xi_i + \frac{h_n u}{\sqrt{c_{ij}}};X_j\right) du,
  \end{align*}
  where the last follows from the change of the variable $\xi_j = \xi_i + h_n u / \sqrt{c_{ij}}$. Next,
  note that for all values of $\xi_i$ such that $f_{\xi|X}(\xi_i|X_j)$ is continuous at $\xi_i$, by the dominated convergence theorem (DCT), we have $I_1(X_i,X_j,\xi_i;h_n) \rightarrow \frac{\mu_K f_{\xi|X}(\xi_i|X_j) }{\sqrt{c(X_i,X_j,\xi_i)}}$ and, consequently, also $I(X_i,X_j,\xi_i;h_n) \rightarrow \frac{\mu_K f_{\xi|X}(\xi_i|X_j) }{\sqrt{c(X_i,X_j,\xi_i)}}$, where $c_{ij} = c(X_i,X_j,\xi_i)$, and $\mu_K$ is as defined in Assumption \ref{ass: K and h}\eqref{item: kernel}. By Assumption \ref{ass: xi dist}\eqref{item: xi cont}, this applies for almost all $\xi_i$. Moreover, by Lemmas \ref{lem: A aux kernel}(i) and \ref{lem: A aux O(1)}(i) and Assumption \ref{ass: K and h}\eqref{item: kernel}, $I(X_i,X_j,\xi_i;h_n)$ is uniformly bounded since we have
  \begin{align*}
    I(X_i,X_j,\xi_i;h_n) &\leqslant \overline K h_n^{-1} \ex{\ind\{\abs{\xi_j - \xi_i} \leqslant \alpha \} h_n | \xi_i} \leqslant C.
  \end{align*}
  Hence, by the DCT, for all $X_i$ and $X_j$, we have,
  \begin{align*}
    I(X_i,X_j;h_n) \coloneqq \ex{I(X_i,X_j,\xi_i;h_n)|X_i,X_j} \rightarrow  \underbrace{\int \frac{\mu_K}{\sqrt{c(X_i,X_j,\xi_i)}} f_{\xi|X}(\xi|X_i) f_{\xi|X}(\xi|X_j) d\xi}_{\coloneqq \lambda(X_i,X_j)}, 
  \end{align*}
  Finally, since $\ex{\zeta_n(Z_i,Z_j,Z_k)|X_i,X_j} = I(X_i,X_j;h_n) \mathcal C (X_i,X_j)$ is uniformly bounded, the DCT guarantees $\ex{\zeta_n(Z_i,Z_j,Z_k)} \rightarrow A \coloneqq \ex{\lambda(X_i,X_j) \mathcal C(X_i,X_j)}$.
  Hence, we conclude $\hat A_n = \ex{\zeta_n(Z_i,Z_j,Z_k)} + o_p(1) = A + o_p(1)$, where $\lambda_{min} (A) > C$ for some $C > 0$ by Assumption~\ref{ass: beta ID}\eqref{item: xi full rank}, which completes the proof.

\subsubsection{Bounding $\hat B_n$}
First, we introduce the following notations and prove an auxiliary lemma stated below.
\begin{align}
    q_n (Z_i,Z_j,Z_k) &\coloneqq h_n^{-1} K\left(\frac{d^2(Z_i,Z_j)}{h_n^2}\right) (W(X_i,X_k) - W(X_j,X_k)) (g(\xi_i,\xi_k) - g(\xi_j,\xi_k)), \notag \\ 
    p_n (Z_i,Z_j,Z_k) &\coloneqq \frac{1}{3} (q_n(Z_i,Z_j,Z_k) + q_n(Z_k,Z_i,Z_j) + q_n(Z_j,Z_k,Z_i)), \label{eq: A p_n def}
\end{align}
\begin{AppLem}
  \label{lem: A aux q}
  Suppose that the hypotheses of Theorem~\ref{the: kernel rate} are satisfied. Then, (i) $b_n \coloneqq \ex{q_n(Z_i,Z_j,Z_k)} = O(h_n^2)$,  and (ii) $\zeta_{1,n} \coloneqq \ex{\norm{\ex{p_n(Z_i,Z_j,Z_k)|Z_i} - b_n}^2} = O(h_n^3)$.
\end{AppLem}

\begin{proof}[Proof of Lemma \ref{lem: A aux q}]
Before starting the proof, we introduce the following notations $s_n^{(1)}(Z_i) \coloneqq \ex{q_n(Z_i,Z_j,Z_k)|Z_i} - b_n$, $s_n^{(2)}(Z_i) \coloneqq \ex{q_n(Z_k,Z_i,Z_j)|Z_i} - b_n$, and $s_n^{(3)}(Z_i) \coloneqq \ex{q_n(Z_j,Z_k,Z_i)|Z_i} - b_n$, so
  \begin{align}
    \ex{p_n(Z_i,Z_j,Z_k)|Z_i} - b_n &= \frac{1}{3} \left(s_n^{(1)}(Z_i) + s_n^{(2)}(Z_i) + s_n^{(3)}(Z_i)\right), \notag \\
    \zeta_{1,n}^2 &= \frac{1}{9} \ex{\norm{s_n^{(1)}(Z_i) + s_n^{(2)}(Z_i) + s_n^{(3)}(Z_i)}^2}. \label{eq: A zeta}
  \end{align}

  First, let us compute $E_{q} (X_i,\xi_i,X_j,Z_k;h_n) \coloneqq \ex{q_n(Z_i,Z_j,Z_k)|X_i, \xi_i, X_j, Z_k}$. Note that
  \begin{align*}
    E_q (X_i,\xi_i,X_j,Z_k;h_n) = &\underbrace{\ex{h_n^{-1} K \left(\frac{c_{ij}(\xi_j - \xi_i)^2}{h_n^2}\right) (W_{ik} - W_{jk}) (g(\xi_i,\xi_k) - g(\xi_j,\xi_k))|X_i,\xi_i,X_j,Z_k}}_{\coloneqq E_{q,2(X_i,\xi_i,X_j,Z_k;h_n)}} \\
    &+ \underbrace{\ex{h_n^{-1} r_K(Z_i,Z_j;h_n) (W_{ik} - W_{jk}) (g(\xi_i,\xi_k) - g(\xi_j,\xi_k)) |X_i,X_j,\xi_i }}_{\coloneqq r_{E_q}(X_i,\xi_i,X_j,Z_k;h_n)},
  \end{align*}
  where $r_K$ is defined in Lemma \ref{lem: A aux kernel}. Combining Lemma \ref{lem: A aux kernel}(iii), boundedness of $W$, Assumption~\ref{ass: basic}\eqref{item: lipschitz g}, and Lemma \ref{lem: A aux O(1)}(i), we conclude that there exists $C > 0$ such that $\abs{r_{E_q}(X_i,\xi_i,X_j,Z_k;h_n) } \leqslant C h_n^2$ a.s.

  Next, we compute $E_{q,2}(X_i,\xi_i,X_j,Z_k;h_n)$. We start with considering $\xi_i$ and $\xi_k$ such that $\abs{\xi_i - \xi_k} > \alpha h_n$. Using Lemma \ref{lem: A aux kernel}(i), Assumption \ref{ass: local smooth g}, and boundedness of $W$, we obtain
  \begin{align*}
    \bigg \Vert K \left(\frac{c_{ij}(\xi_j - \xi_i)^2}{h_n^2}\right) (W_{ik} - W_{jk}) (\underbrace{g(\xi_i,\xi_k) - g(\xi_j,\xi_k) - G(\xi_i,\xi_k) (\xi_i - \xi_j)}_{r_g(\xi_i,\xi_j,\xi_k)}) \bigg \Vert \leq C h_n^2 \ind\{\abs{\xi_j - \xi_i} \leqslant \alpha h_n\}
  \end{align*} 
  for some $C > 0$ for all $(Z_i,Z_j,Z_k)$ satisfying $\abs{\xi_i - \xi_k} > \alpha h_n$. Then, using Lemma \ref{lem: A aux O(1)}(i),
  \begin{align*}
    E_{q,2}(X_i,\xi_i,X_j,Z_k;h_n) = &\underbrace{\ex{h_n^{-1} K \left(\frac{c_{ij}(\xi_j - \xi_i)^2}{h_n^2}\right) (W_{jk} - W_{ik}) G(\xi_i,\xi_k) (\xi_j - \xi_i) | X_i,\xi_i,X_j,Z_k}}_{\coloneqq E_{q,l}(X_i,\xi_i,X_j,Z_k;h_n)} \\
    &+ r_{E_{q,2}}(X_i,\xi_i,X_j,Z_k;h_n),
  \end{align*}
  where $\abs{r_{E_{q,2}}(X_i,\xi_i,X_j,Z_k;h_n)} \leqslant C h_n^2$ a.s. for some $C > 0$ for $\abs{\xi_i - \xi_k} > \alpha h_n$. Next, consider
  \begin{align*}
    E_{q,l}(X_i,\xi_i,X_j,Z_k;h_n) &= (W_{jk} - W_{ik}) G(\xi_i,\xi_k) \int h_n^{-1} K \left(\frac{c_{ij}(\xi_j - \xi_i)^2}{h_n^2}\right) (\xi_j - \xi_i) f_{\xi|X} (\xi_j;X_j) d\xi_j \\
    &= \frac{h_n (W_{jk} - W_{ik}) G(\xi_i,\xi_k)}{{c_{ij}}} \int K(u^2) u f_{\xi|X}(\xi_i + h_n u/\sqrt{c_{ij}};X_j) du,
  \end{align*}
  where the last equality follows from the change of the variable $\xi_j = \xi_i + h_n u / \sqrt{c_{ij}}$. Note 
  \begin{align}
    &\abs{\int K(u^2) u f_{\xi|X}(\xi_i + h_n u/\sqrt{c_{ij}};X_j) du} \notag \\ \leqslant &\int \abs{K(u^2) u \left(f_{\xi|X} (\xi_i + h_n u /\sqrt{c_{ij}};X_j) - f_{\xi|X} (\xi_i;X_j)\right)} du 
    \leqslant \frac{\overline K C_\xi h_n}{\sqrt{c_{ij}}} \label{eq: A kernel integral bound}
  \end{align}
  where the first inequality used $\int u K(u^2) du = 0$ and the triangle inequality, and the second inequality holds for all $\xi_i$ such that $f_{\xi|X}(\cdot;X_j)$ is continuous on $B_{h_n u / \sqrt{c_{ij}}}(\xi_i)$ thanks to Assumptions \ref{ass: xi dist}\eqref{item: xi lipschitz} and \ref{ass: K and h}\eqref{item: kernel}. \eqref{eq: A kernel integral bound} together with boundedness of $W$, $G$ and $c_{ij}^{-1}$ (Assumptions \ref{ass: basic}\eqref{item: w}, \ref{ass: local smooth g}, and \ref{ass: beta ID}\eqref{item: d* approximation}, respectively) implies that there exists a uniform constant $C$ such that $\abs{E_{q,l}(X_i,\xi_i,X_j,Z_k;h_n)} \leqslant C h_n^2$
  almost surely for all $Z_i$, $Z_j$, and $Z_k$ such that $\abs{\xi_i - \xi_k} > \alpha h_n$ and $f_{\xi|X}(\cdot;X_j)$ is continuous on $B_{h_n u/\sqrt{c_{ij}}}(
  \xi_i)$. Combining this with the previously obtained bounds on $r_{E_q}$ and $r_{E_{q,2}}$, we conclude that there exists a uniform constant $C > 0$ such that (for sufficiently large $n$)
  \begin{align}
    \label{eq: A q quad bound}
    \abs{E_{q}(X_i,\xi_i,X_j,Z_k;h_n)} \leqslant C h_n^2
  \end{align}
  almost surely for all $Z_i$, $Z_j$, and $Z_k$ such that $\abs{\xi_i - \xi_k} > \alpha h_n$ and $f_{\xi|X}(\cdot;X_j)$ is continuous on $B_{h_n u/\sqrt{c_{ij}}}(\xi_i)$. Moreover, for all $(Z_i,Z_j,Z_k)$ without additional qualifiers, we also have
  \begin{align}
    \label{eq: A q uniform bound}
    \abs{E_{q}(X_i,\xi_i,X_j,Z_k;h_n)} \leqslant C h_n,
  \end{align}
  for some $C > 0$. To inspect \eqref{eq: A q uniform bound}, note $\norm{q_n(Z_i,Z_j,Z_k)} \leqslant C \ind\{\abs{\xi_j - \xi_i} \leqslant \alpha h_n\}$ a.s. for some $C > 0$ thanks to Lemma \ref{lem: A aux kernel}(i) and Assumptions \ref{ass: basic}\eqref{item: w} and \eqref{item: lipschitz g}, and apply Lemma \ref{lem: A aux O(1)}(i).

  Equipped with the bounds \eqref{eq: A q quad bound} and \eqref{eq: A q uniform bound}, now we want to bound $b_n = \ex{q_n(Z_i,Z_j,Z_k)}$. We start with considering $\ex{q_n(Z_i,Z_j,Z_k)|Z_i}$. Since $c_{ij} > \underline c > 0$ (Assumption \ref{ass: beta ID}\eqref{item: d* approximation}), Assumption \ref{ass: xi dist}\eqref{item: xi cont} guarantees that there exists $\gamma_1 > 0$ such that the probability mass of $\xi_i$ such that $f_{\xi|X}(\cdot;X_j)$ is continuous on $B_{h_n u/\sqrt{c_{ij}}}(\xi_i)$ is at least $1 - \gamma_1 h_n$ irrespectively of the values of $X_i$ and $X_j$. Also, by Assumption \ref{ass: xi dist}\eqref{item: xi pdf}, the probability mass of $\xi_k$ such that $\abs{\xi_k - \xi_i} > \alpha h_n$ is at least $1 - \gamma_{2} h_n$ irrespectively of the value of $\xi_i$ for some $\gamma_2 > 0$. For such values of $\xi_i$ and $\xi_k$, the bound \eqref{eq: A q quad bound} applies. Moreover, the bound \eqref{eq: A q uniform bound} applies with probability one. Hence, integrating $E_q(X_i,\xi_i,X_j,Z_k;h_n)$ over $(X_j,Z_k)$ ensures that there exists $C > 0$ such that $\norm{\ex{q_n(Z_i,Z_j,Z_k)|Z_i}} \leqslant C h_n^2$ for all $\xi_i$ such that $f_{\xi|X}(\cdot;X_j)$ is continuous on $B_{h_n u/\sqrt{c_{ij}}}(\xi_i)$, which happens with probability $1 - \gamma_1 h_n$ at least. Moreover, \eqref{eq: A q uniform bound} immediately implies that $\norm{\ex{q_n(Z_i,Z_j,Z_k)|Z_i}} \leqslant C h_n$ with probability one. Combining these bounds and integrating over $Z_i$ gives
  \begin{align*}
    \norm{\ex{q_n(Z_i,Z_j,Z_j)}} \leqslant \ex{\norm{{\ex{q_n(Z_i,Z_j,Z_j)|Z_i}}}} \leqslant (1 - \gamma_1 h_n) \times C h_n^2 + \gamma_1 h_n \times C h_n = O(h_n^2), 
  \end{align*}
  which completes the proof of the first statement of the lemma.

  To prove the second statement of the lemma, thanks to \eqref{eq: A zeta}, it is sufficient to verify that $\e[{\Vert{s_n^{(\ell)}(Z_i)}\Vert^2}] = O(h_n^3)$ for $\ell \in \{1,2,3\}$. We start with $\ell \in \{1,2\}$. By the same bounds on $\norm{\ex{{q_n(Z_i,Z_j,Z_j)}|Z_i}}$ as established above and $b_n = O(h_n^2)$, the following holds: (i) for some $C > 0$ and $\gamma_1 > 0$, we have $\Vert{s_n^{(1)}(Z_i)}\Vert^2 \leqslant C h_n^4$ with probability at least $1 - \gamma_1 h_n$ and $\Vert{s_n^{(1)}(Z_i)}\Vert^2 \leqslant C h_n^2$ with probability one. Hence, $\e[{\Vert{s_n^{(1)}(Z_i)}\Vert^2}] = \e[{\Vert{s_n^{(2)}(Z_i)}\Vert^2}] \leqslant C h_n^3$ for some $C > 0$. Finally, notice we can write $s_n^{(3)}(Z_k) = \ex{q_n(Z_i,Z_j,Z_k)|Z_k} - b_n$. Again, for a fixed $\xi_k$, the bound \eqref{eq: A q quad bound} applies for all $\xi_i$ satisfying (i) $\abs{\xi_k - \xi_i} > \alpha h_n$ and (ii) $f_{\xi|X}(\cdot;X_j)$ is continuous on $B_{h_n u/\sqrt{c_{ij}}}(\xi_i)$. By the same reasoning as above, the probability mass of such $\xi_i$ is at least $1 - \gamma_3 h_n$ for some $\gamma_3 > 0$ irrespectively of $(X_i,X_j,Z_k)$. At the same time, the bound \eqref{eq: A q uniform bound} applies with probability one. Hence, integrating $E_q(X_i,\xi_i,X_j,Z_k;h_n)$ over $(X_i,\xi_i,X_j)$ gives $\norm{\ex{q_n(Z_i,Z_j,Z_k)|Z_k}} \leqslant C h_n^2$ a.s. for some $C > 0$. This, paired with $b_n = O(h_n^2)$, implies that $\e [\Vert {{s_n^{(3)}(Z_k)}}\Vert^2] = \e [\Vert {{s_n^{(3)}(Z_i)}}\Vert^2] \leqslant C h_n^4$, which completes the proof.
\end{proof}

  {\noindent \bf Bounding $\hat B_n$} Equipped with Lemma~\ref{lem: A aux q}, we are now ready to bound $\hat B_n$.
  The first step of the proof is to bound $B_n$ defined as 
  \begin{align*}
    B_n &\coloneqq \binom{n}{2}^{-1} h_n^{-1} \sum_{i < j} K\left(\frac{d_{ij}^2}{h_n^2}\right) \frac{1}{n-2} \sum_{k \neq i,j} (W_{ik} - W_{jk}) (g(\xi_i,\xi_k) - g(\xi_j,\xi_k)) \\
    &= \binom{n}{3}^{-1} \sum_{i < j < k} p_n(Z_i,Z_j,Z_k)  =  b_n + \underbrace{\binom{n}{3}^{-1} \sum_{i < j < k} ( p_n(Z_i,Z_j,Z_k) - b_n)}_{\coloneqq U_n},
  \end{align*}
  so $B_n$ is a third order U-statistic with the symmetrized kernel $p_n (Z_i,Z_j,Z_k)$ given by \eqref{eq: A p_n def}, and $b_n \coloneqq \ex{q_n(Z_i,Z_j,Z_k)} = \ex{p_n(Z_i,Z_j,Z_k)}$ is defined and bounded in Lemma~\ref{lem: A aux q}. We proceed with bounding using a Bernstein type inequality for U-statistic developed in \citetsupp{arcones1995bernstein}.
  Specifically, Theorem 2 in \citetsupp{arcones1995bernstein} guarantees that $U_n = O_p\left( \max\left\{\frac{\zeta_{1,n}}{n^{1/2}},\frac{1}{n}\right\} \right)$, where $\zeta_{1,n}$ is defined in Lemma~\ref{lem: A aux q}, which demonstrates $\zeta_{1,n}^2 = O(h_n^3)$. Using Assumption~\ref{ass: K and h}\eqref{item: h_n}, we also obtain $\frac{\zeta_{1,n}}{n^{1/2}} =\frac{O(h_n^{3/2})}{n^{1/2}} = o (h_n^2)$. Finally, since Lemma~\ref{lem: A aux q}(i) guarantees $b_n = O(h_n^2)$, we conclude $B_n = b_n + U_n = O_p (h_n^2 + n^{-1})$.

  The second step is to bound $\hat B_n - B_n$. Combining the result of the result of Lemma \ref{lem: A aux kernel}(ii) with Assumption \ref{ass: K and h}\eqref{item: kernel}, we have, with probability approaching one,
  \begin{align*}
    \norm{\hat B_n - B_n} \leqslant &\overline K' \frac{\max_{i \neq j} \abs{\hat d_{ij}^2 - d_{ij}^2}}{h_n^2} \nonumber \\ &\times \binom{n}{2}^{-1} h_n^{-1} \sum_{i < j} \underbrace{\ind\{\abs{\xi_j - \xi_i} \leqslant \alpha h_n\} \frac{\sum_{k \neq i,j} \norm{(W_{ik} - W_{jk}) (g(\xi_i,\xi_k) - g(\xi_j,\xi_k))}}{n-2}}_{\leqslant C h_n \ind\{\abs{\xi_j - \xi_i} \leqslant \alpha h_n\}} .
  \end{align*}
  Notice that using Assumption \ref{ass: basic}\eqref{item: lipschitz g}, boundedness of $W$, and Lemma \ref{lem: A aux O(1)}(ii), we can bound the second line of the display equation above by $O_p(h_n)$. Together with \eqref{eq: R_n def}, this allows us to obtain $\hat B_n - B_n = O_p\left(\frac{R_n^{-1}}{h_n}\right)$. Together with the previously obtained bound on $B_n$, this delivers the required result.

\subsubsection{Bounding $\hat C_n$}\label{sssec: C_n}
  \begin{align}
    \hat C_n = & \underbrace{\binom{n}{2}^{-1} h_n^{-1} \sum_{i < j} K\left(\frac{d_{ij}^2}{h_n^2}\right) \frac{1}{n-2} \sum_{k \neq i,j} (W_{ik} - W_{jk}) (\varepsilon_{ik} - \varepsilon_{jk})}_{\coloneqq C_n} \notag \\
    &+ \underbrace{\binom{n}{2}^{-1} h_n^{-1} \sum_{i < j} \left(K\left(\frac{\hat d_{ij}^2}{h_n^2}\right) - K\left(\frac{d_{ij}^2}{h_n^2}\right)\right) \frac{1}{n-2} \sum_{k \neq i,j} (W_{ik} - W_{jk}) (\varepsilon_{ik} - \varepsilon_{jk})}_{\coloneqq \Delta \hat C_n}. \label{eq: A delta C_n def}
  \end{align}
  
  The first step is to argue that $C_n = O_p \left(n^{-1}\right)$.
  Let $ K_{n,ij} \coloneqq h_n^{-1} K\left(\frac{ d_{ij}^2}{h_n^2}\right)$. Note
  \begin{align*}
     \sum_{i < j}  K_{n,ij} \sum_{k \neq i,j} (W_{ik} - W_{jk})(\varepsilon_{ik} - \varepsilon_{jk}) 
     &= \sum_{i \neq j}  K_{n,ij} \sum_{k \neq i,j} (W_{ik} - W_{jk}) \varepsilon_{ik}\\
    &= \sum_{i < k} \left(\sum_{j \neq i,k} \left( K_{n,ij} (W_{ik} - W_{jk   }) +  K_{n,kj} (W_{ik} - W_{ji})\right) \right) \varepsilon_{ik}.
  \end{align*}
  Hence, we can represent $C_n$ as 
  \begin{align*}
    C_n   = \binom{n}{2}^{-1} \sum_{i < k} \hat \omega_{n,ik} \varepsilon_{ik}, \quad  \hat \omega_{n,ik} &\coloneqq \underbrace{\frac{\sum_{j \neq i,k} K_{n,ij} (W_{ik} - W_{jk})}{n-2}}_{\coloneqq \hat \kappa_{n, ik}}   + \underbrace{\frac{ \sum_{j \neq i,k} K_{n,kj} (W_{ik} - W_{ji})}{n-2}}_{\coloneqq \hat \eta_{n,ik}}.
  \end{align*}
  
  Next, we want to argue that $\max_{i \neq k} \norm{\hat \omega_{n,ik}} < C_\omega$ with probability approaching one for some $C_\omega > 0$. To this end, we first show that $\max_{i \neq k} \norm{\hat \omega_{n,ik} - \omega_{n,ik}} = o_p(1)$ and then we verify that $\max_{i \neq k} \norm{\omega_{n, ik}} < C $ for some $C > 0$, where
  \begin{align*}
    \omega_{n,ik} \coloneqq \underbrace{\ex{ K_{n,ij} (W_{ik} - W_{jk}) |Z_i,Z_k}}_{\coloneqq \kappa_{n, ik}} + \underbrace{\ex{K_{n,kj} (W_{ik} - W_{ji}) |Z_i,Z_k}}_{\coloneqq \eta_{n,ik}}.
  \end{align*}
  Uniform consistency of $\hat \omega_{n, ik}$ would follow from uniform consistency of $\hat \kappa_{n,ik}$ and $\eta_{n, ik}$ to $\kappa_{n, ik}$ and $\eta_{n, ik}$, respectively. Next, we will establish this result for $\hat \kappa_{n, ik}$ noting that the analogous result for $\hat \eta_{n, ik}$ follows by the same argument. Note that, conditional on $Z_i$ and $Z_k$, $\{K_{n,ij} (W_{ik} - W_{jk})\}_{j \neq i,k}$ is a collection of bounded (given the sample size) independent variables with $\norm{K_{n,ij} \left(W_{ik} - W_{jn}\right)} \leqslant C h_n^{-1}$ and $\ex{\norm{K_{n,ij} \left(W_{ik} - W_{jn}\right)}^2|Z_i,Z_k} \leqslant C h_n^{-1}$ for all $Z_i$ and $Z_k$ for some $C > 0$, where we used boundedness of $W$ and $K$ combined with Lemmas~\ref{lem: A aux kernel}(i) and \ref{lem: A aux O(1)}(i). Hence, applying Bernstein inequality \ref{the: A bounded Bernstein}, we conclude that there exist positive constants $a$, $b$, and $C$ such that for all $Z_i,Z_k$ and $\epsilon > 0$ we have
  \begin{align*}
    \pr\left(\norm{\hat \kappa_{n,ik} - \kappa_{n,ik}} \geqslant \epsilon |Z_i,Z_k \right) &\leqslant C \exp\left(-\frac{(n-2) h_n \epsilon^2}{a + b \epsilon}\right),\\
      \pr\left(\max_{i \neq k}\norm{\hat \kappa_{n,ik} - \kappa_{n,ik}} \geqslant \epsilon \right) &\leqslant \binom{n}{2} C \exp\left(-\frac{(n-2) h_n \epsilon^2}{a + b \epsilon}\right) \rightarrow 0,
  \end{align*}
  where the second inequality follows from the union bound, and the convergence follows from Assumption~\ref{ass: K and h}\eqref{item: h_n}. Hence, $\hat \kappa_{n,ik}$, and, consequently, $\hat \omega_{n, ik}$  are uniformly consistent for $\kappa_{n,ik}$, and $\omega_{n, ik}$, respectively. Finally, using the boundedness of $W$ and $K$ and Lemmas~\ref{lem: A aux kernel}(i) and \ref{lem: A aux O(1)}(i) again, we conclude that $\max_{i \neq k} \norm{\omega_{n,ik}} \leq C$ for some $C > 0$, which delivers the desired result for $\hat \omega_{n, ik}$.

  We are now ready to bound $C_n$ using that $\max_{i \neq k} \norm{\hat \omega_{n,ik}} < C_\omega$ with probability approaching one. Notice that, conditional on $\{Z_i\}_{i=1}^n$, $\{\hat \omega_{n,ik} \varepsilon_{ik}\}_{i < k}$ is a collection of independent vectors with zero mean, which satisfy the requirements of Theorem \ref{the: A unb Bernstein}. Therefore, Theorem \ref{the: A unb Bernstein} guarantees that there exist some positive constants $C$, $a$, $b$ such that for all $\{Z_i\}_{i=1}^n$ satisfying $\max_{i \neq k} \norm{\hat \omega_{n,ik}} < C_{\omega}$, for all $\epsilon > 0$,
  \begin{align*}
    \pr\left(\norm{C_n} > \epsilon |\{X_i,\xi_i\}_{i=1}^n\right) \leqslant C \exp\left(- \frac{\binom{n}{2} \epsilon^2}{a + b \epsilon}\right).
  \end{align*}
  Since the requirement $\max_{i \neq k} \norm{\hat \omega_{n,ik}} < C_\omega$ is satisfied with probability approaching one, we conclude that $C_n = O_p(n^{-1})$.

  The final part of the proof is to bound $\Delta \hat C_n$ defined in \eqref{eq: A delta C_n def}. Applying Corollary~\ref{cor: bernstein} conditional on $\{Z_i\}_{i=1}^n$ and using boundedness of $W$, we conclude there exist positive constants $C_2$, $a_2$ and $b_2$ such that for all $\{Z_i\}_{i=1}^n$ and for all $\epsilon > 0$
  \begin{align*}
    \pr\left(\norm{\frac{1}{n-2} \sum_{k \neq i,j} (W_{ik} - W_{jk}) (\varepsilon_{ik} - \varepsilon_{jk})} > \epsilon | \{Z_i\}_{i=1}^n \right) &\leqslant C_2 \exp\left(-\frac{(n-2)\epsilon^2}{a_2 + b_2 \epsilon}\right),\\
    \pr\left(\max_{i \neq j} \norm{\frac{1}{n-2} \sum_{k \neq i,j} (W_{ik} - W_{jk}) (\varepsilon_{ik} - \varepsilon_{jk})} > \epsilon \right) &\leqslant \binom{n}{2} C_2 \exp\left(-\frac{(n-2)\epsilon^2}{a_2 + b_2 \epsilon}\right),
  \end{align*}
  where the second inequality uses the union bound and the uniformity of those constants over $\{Z_i\}_{i=1}^n$. Hence, we have $\max_{i \neq j} \Vert{\frac{1}{n-2} \sum_{k \neq i,j} (W_{ik} - W_{jk}) (\varepsilon_{ik} - \varepsilon_{jk})}\Vert = O_{p} \left( \left(\frac{\ln n}{n}\right)^{1/2} \right)$ implying
  \begin{align*}
    \norm{\Delta \hat C_n} \leqslant \binom{n}{2}^{-1} h_n^{-1} \sum_{i < j} \abs{K\left(\frac{\hat d_{ij}^2}{h_n^2}\right) - K\left(\frac{d_{ij}^2}{h_n^2}\right)} \times O_{p} \left( \left(\frac{\ln n}{n}\right)^{1/2} \right).
  \end{align*}
  Combining this with \eqref{eq: R_n def}, we obtain $\Vert{\Delta \hat C_n}\Vert = O_p\left(\frac{R_n^{-1}}{h_n^2} \left(\frac{\ln n}{n}\right)^{1/2} \right)$, which, together with $C_n = O_p (n^{-1})$, delivers the desired result for $\hat C_n = C_n + \Delta \hat C_n$.

\subsection{Proof of the results of Section \ref{sssec: homoskedastic rate}}

First, we introduce the following notations and prove an auxiliary lemma stated below.
\begin{align}
  \hat q_{ij}^2(\beta) &\coloneqq \frac{1}{n-2} \sum_{k \neq i,j} \left(Y_{ik} - Y_{jk} - (W_{ik} - W_{jk})'\beta \right)^2, \nonumber \\
  d_{ij,n-2}^2(\beta) &\coloneqq \frac{1}{n-2} \sum_{k \neq i,j} (Y_{ik}^* - Y_{jk}^* - (W_{ik} - W_{jk})' \beta )^2, \nonumber \\
  d_{ij,n}^2(\beta) &\coloneqq \frac{1}{n} \sum_{k} (Y_{ik}^* - Y_{jk}^* - (W_{ik} - W_{jk})' \beta )^2, \label{eq: d ij n star}\\
  q_{ij}^2(\beta) &\coloneqq \ex{(Y_{ik} - Y_{jk} - (W_{ik} - W_{jk})'\beta)^2}, \quad d_{ij}^2(\beta) = {\ex{(Y_{ik}^* - Y_{jk}^* - (W_{ik} - W_{jk})'\beta)^2}} \nonumber.
\end{align}
For brevity, we will also use $g_{ij} \coloneqq g (\xi_i, \xi_i)$ for all $i$ and $j$.

\begin{AppLem}
  \label{lem: A d_n UC}
  Suppose that $\mathcal B$ is compact. Then, under Assumptions \ref{ass: DGP} and \ref{ass: basic},
  \begin{enumerate}[(i)]
    \item $\max_{i \neq j} \max_{\beta \in \mathcal B} \abs{d_{ij,n-2}^2(\beta) - d_{ij}^2(\beta)} = O_p\left(({\ln n}/{n})^{1/2}\right)$,
    \item $\max_{i \neq j} \max_{\beta \in \mathcal B} \abs{d_{ij,n}^2(\beta) - d_{ij}^2(\beta)} = O_p\left(({\ln n}/{n})^{1/2}\right)$.
  \end{enumerate}
\end{AppLem}

\begin{proof}[Proof of Lemma \ref{lem: A d_n UC}]
  Using $Y_{ik}^* - Y_{jk}^* = (W_{ik} - W_{jk})' \beta_0 + g_{ik} - g_{jk}$, we obtain
  \begin{align*}
    d_{ij,n-2}^2(\beta) = &(\beta_0 - \beta)' \frac{1}{n-2} \sum_{k \neq i,j} (W_{ik} - W_{jk}) (W_{ik} - W_{jk})' (\beta_0 - \beta) \\
    &+ \frac{2}{n-2} \sum_{k \neq i,j} (W_{ik} - W_{jk})(g_{ik} - g_{jk}) (\beta_0 - \beta) + \frac{1}{n-2} \sum_{k \neq i,j} (g_{ik} - g_{jk})^2,\\
    d_{ij}^2(\beta) = &(\beta_0 - \beta)' \ex{(W_{ik} - W_{jk})(W_{ik} - W_{jk})'|Z_i,Z_j} (\beta_0 - \beta) \\
    &+ 2 \ex{(W_{ik} - W_{jk})(g_{ik} - g_{jk})|Z_i,Z_j} (\beta_0 - \beta) + \ex{(g_{ik} - g_{jk})^2|Z_i,Z_j}.
  \end{align*}
  Since $\mathcal B$ is bounded, for some positive constants $c_1$ and $c_2$, we have
  \begin{align}
    \label{eq: A d* bound}
    &\max_{i \neq j} \max_{\beta \in \mathcal B} \abs{d_{ij,n-2}^2(\beta) - d_{ij}^2(\beta)} \leqslant c_1 S_1 + c_2 S_2 + S_3,\\
        &S_1 \coloneqq \max_{i \neq j} \norm{\frac{1}{n-2} \sum_{k \neq i,j} (W_{ik} - W_{jk}) (W_{ik} - W_{jk})' -  \mathcal C(X_i,X_j)}, \notag\\ 
    &S_2 \coloneqq \max_{i \neq j} \norm{\frac{1}{n-2} \sum_{k \neq i,j} (W_{ik} - W_{jk}) (g_{ik} - g_{jk})' -  \ex{(W_{ik} - W_{jk})(g_{ik} - g_{jk})|Z_i,Z_j}},\notag \\
    &S_3 \coloneqq \max_{i \neq j} \abs{\frac{1}{n-2} \sum_{k \neq i,j} (g_{ik} - g_{jk})^2 -  \ex{(g_{ik} - g_{jk})^2|Z_i,Z_j}}. \notag
  \end{align}
  where we used $\mathcal C(X_i,X_j) = \ex{(W_{ik} - W_{jk})(W_{ik} - W_{jk})'|Z_i,Z_j}$. Hence, to prove the first statement of the lemma, it is sufficient to show that $S_\ell = O_p\left(({\ln n}/{n})^{1/2}\right)$ for $\ell \in \{1,2,3\}$.

  We start with $\ell = 1$. Conditional on $Z_i$ and $Z_j$, $\{(W_{ik} - W_{jk})(W_{ik} - W_{jk})'\}_{k \neq i,j}$ is a collection of iid matrices, which are uniformly bounded over $Z_i$ and $Z_j$. Hence, by Bernstein inequality \ref{the: A bounded Bernstein}, there exist positive constants $a$, $b$, and $C$ such that for all $\epsilon > 0$ and (almost) all $Z_i$ and $Z_j$, we have
  \begin{align*}
    \pr\left( \norm{\frac{1}{n-2} \sum_{k \neq i,j} (W_{ik} - W_{jk}) (W_{ik} - W_{jk})' -  \mathcal C(X_i,X_j)} > \epsilon | Z_i, Z_j \right) \leqslant C \exp \left( - \frac{(n-2) \epsilon^2}{a + b \epsilon}\right).
  \end{align*}
  Since $a$, $b$, and $C$ are uniform over $Z_i$ and $Z_j$, applying the union bound, we have
  \begin{align*}
    \pr \left(S_1 > \epsilon \right) \leqslant \binom{n}{2} C \exp \left( - \frac{(n-2) \epsilon^2}{a + b \epsilon}\right),
  \end{align*}
  which implies $S_1 = O_p\left(({\ln n}/{n})^{1/2}\right)$. The same argument can be applied to inspect that $S_\ell = O_p\left(({\ln n}/{n})^{1/2}\right)$ for $\ell \in \{2,3\}$, which together with \eqref{eq: A d* bound} delivers the desired result for $d_{ij,n-2}^2$. Essentially the same argument applies to demonstrate the result for $d_{ij,n}^2$.
\end{proof}

  {\noindent \bf Proof of Lemma~\ref{lem: homosked d UC}}. Recall that under \eqref{eq: homosked eps}, $q_{ij}^2(\beta) = d_{ij}^2(\beta) + 2 \sigma^2$. First, we show that
  \begin{align}
    \max_{i \neq j} \max_{\beta \in \mathcal B} \abs{\hat q_{ij}^2(\beta) - q_{ij}^2(\beta)} = \max_{i \neq j} \max_{\beta \in \mathcal B} \abs{\hat q_{ij}^2(\beta) - d_{ij}^2(\beta) - 2 \sigma^2} = O_p\left(({\ln n}/{n})^{1/2}\right). \label{eq: A q UC}
  \end{align}
  Next, $\hat q_{ij}^2(\beta)$ can be decomposed as
  \begin{align*}
    \hat q_{ij}^2(\beta)
    = &\underbrace{\frac{1}{n-2} \sum_{k \neq i,j} \left(Y_{ik}^* - Y_{jk}^* - (W_{ik} - W_{jk})'\beta\right)^2}_{d_{ij,n-2}^2(\beta)} + \underbrace{\frac{1}{n-2} \sum_{k \neq i,j} (\varepsilon_{ik} - \varepsilon_{jk})^2}_{\coloneqq 2 \hat \sigma^2_{ij}} \\
    &+ \underbrace{\frac{2}{n-2} \sum_{k \neq i,j} \left((W_{ik} - W_{jk})' (\beta_0 - \beta) + g_{ik} - g_{jk}\right) (\epsilon_{ik} - \epsilon_{jk})}_{\coloneqq \hat \Delta_{ij}(\beta)}.
  \end{align*}
  Thanks to the result of Lemma~\ref{lem: A d_n UC}(i), to verify \eqref{eq: A q UC},  it is sufficient to show that (i)~$\max_{i \neq j} \max_{\beta} \vert \hat \Delta_{ij}(\beta) \vert = O_p\left(({\ln n}/{n})^{1/2}\right)$, and (ii) $\max_{i \neq j} \vert 2 \hat \sigma^2_{ij} - 2 \sigma^2 \vert = O_p\left(({\ln n}/{n})^{1/2}\right)$. To inspect (i), recall that $\max_{i \neq j} \Vert{\frac{1}{n-2} \sum_{k \neq i,j} (W_{ik} - W_{jk}) (\varepsilon_{ik} - \varepsilon_{jk})}\Vert = O_p\left(({\ln n}/{n})^{1/2}\right)$ has already been demonstrated in Section~\ref{sssec: C_n}. By a nearly identical argument, we also have $\max_{i \neq j} \abs{\frac{1}{n-2} \sum_{k \neq i,j} \left(g_{ik} - g_{jk}\right) (\epsilon_{ik} - \epsilon_{jk})} =  O_p\left(({\ln n}/{n})^{1/2}\right)$. Together with boundedness of $\mathcal B$, these to convergence results imply that (i) holds. Next, to inspect (ii), note that, conditional on $\{Z_i,Z_j\}$, $\{\varepsilon_{ik} - \varepsilon_{jk}\}_{k \neq i,j}$ is a collection of iid sub-Gaussian random variables.
  Then Corollary \ref{cor: bernstein} guarantees that there exist some positive constants $a$, $b$, and $C$ such that for all $\epsilon > 0$ and (almost) all $Z_i$ and $Z_j$
  \begin{align}
    \pr\left(\abs{\frac{1}{n-2} \sum_{k \neq i,j} \left((\varepsilon_{ik} - \varepsilon_{jk})^2 - 2\sigma^2\right) |Z_i, Z_j } > \epsilon \right) \leqslant C \exp \left(-\frac{(n-2) \epsilon^2}{a + b \epsilon}\right). \label{eq: A 2 sigma sq prob bound}
  \end{align}
  Importantly, Assumption \ref{ass: basic}\eqref{item: eps moments} guarantees that $a$, $b$, and $C$ are uniform over $Z_i$ and $Z_j$. Hence, (ii) following from turning \eqref{eq: A 2 sigma sq prob bound} into an analogous unconditional statement and combining it with the union bound. Inspecting (i) and (ii), completes the proof that \eqref{eq: A q UC} holds. Finally, \eqref{eq: A q UC}, in turn, guarantees that  $\max_{i \neq j} \abs{\hat q_{ij}^2 - q_{ij}^2} = O_p\left(({\ln n}/{n})^{1/2}\right)$, which completes the proof of the first part of the lemma.

  To prove the second part, note $2 \hat \sigma^2 = \min_{i \neq j} \hat q_{ij}^2 = \min_{i \neq j} q_{ij}^2 + O_p\left(({\ln n}/{n})^{1/2}\right)$, and $2 \sigma^2 \leqslant \min_{i \neq j} q_{ij}^2 \leqslant 2 \sigma^2 + \min_{i \neq j} \ex{\left(g(\xi_i,\xi_k) - g(\xi_j,\xi_k)\right)^2|\xi_i,\xi_j}$. Note that Assumption \ref{ass: basic}\eqref{item: lipschitz g} guarantees that $\min_{i \neq j} \ex{\left(g(\xi_i,\xi_k) - g(\xi_j,\xi_k)\right)^2|\xi_i,\xi_j} \leqslant \overline G^2 \min_{i \neq j} \norm{\xi_i - \xi_j}^2$, which delivers the result. \qed

\subsection{Proofs of the results of Section \ref{sssec: general hetero rates}}

\subsubsection{Proof of Theorem \ref{the: Y star convergence}}
First, we state prove two auxiliary lemmas.

\begin{AppLem}
  \label{lem: N(i)}
  Suppose that the hypotheses of Theorem \ref{the: Y star convergence} are satisfied. Then, for any $C > 0$, there exists $a_C > 0$ such that we have $\min_i \frac{\abs{\mathcal B_i (\delta_{n,C}) }}{n-1} \geqslant C (n^{-1} \ln n)^{1/2}$ with probability approaching one,
  where $\mathcal B_i (\delta) \coloneqq \{i' \neq i: X_{i'} = X_i, \xi_{i'} \in B_\delta (\xi_i) \}$, and $\delta_{n,C} \coloneqq a_C (n^{-1} \ln n)^{\frac{1}{2 d_\xi}}$. 
\end{AppLem}

\begin{proof}[Proof of Lemma \ref{lem: N(i)}]
  Let $p_Z(x,\xi;\delta) \coloneqq \pr(X_j = x) \pr(\xi_j \in B_\delta(\xi)|X_j = x)$, $\mathcal Q_i (\delta) \coloneqq \frac{\abs{\mathcal B_i(\delta)}}{n-1}$, and $p_r \coloneqq \pr(X = x_r)$. Then, by Bernstein inequality \ref{the: A bounded Bernstein}, we have for all $\epsilon \in (0,1)$
  \begin{align*}
    \pr(\abs{\mathcal Q_i (\delta) - p_Z(X_i,\xi_i;\delta)} \geqslant \epsilon |X_i, \xi_i ) \leqslant 2 \exp \left(- \frac{\frac{1}{2} (n-1) \epsilon^2}{1 + \frac{1}{3} \epsilon}\right) \leqslant 2 \exp\left(-\frac{1}{4} n \epsilon^2\right).
  \end{align*}
  Hence, we also have $\pr(\abs{\mathcal Q_i (\delta) - p_Z(X_i,\xi_i;\delta)} \geqslant \epsilon  ) \leqslant 2 \exp\left(-\frac{1}{4} n \epsilon^2\right)$
  and, using the union bound, $\pr(\max_{i} \abs{\mathcal Q_i (\delta) - p_Z(X_i,\xi_i;\delta)} \geqslant \epsilon  ) \leqslant 2 n \exp\left(-\frac{1}{4} n \epsilon^2\right)$. Hence, for any $\tilde C > 0$ (and large enough $n$), $\pr(\max_{i} \abs{\mathcal Q_i (\delta) - p_Z(X_i,\xi_i;\delta)} \geqslant \tilde C (n^{-1} \ln n)^{1/2}  ) \leqslant 2 n^{1 - \tilde C/4}$.
  Using Assumption \ref{ass: d* rate}\eqref{item: ball mass}, for $\delta \leqslant \overline \delta$,  we have $p_Z(X_i,\xi_i;\delta) 
      \geqslant \underline \kappa \delta^{d_\xi}$ a.s. where $\underline \kappa = \kappa \min_r p_r$.
  Then,  we have $\min p_Z (X_i,\xi_i; a (n^{-1} \ln n)^{\dxi}) \geqslant \underline{\kappa} a^{d_\xi} (n^{-1} \ln n)^{1/2}$ for $\delta = a (n^{-1} \ln n)^{\dxi}$ for a fixed $a > 0$ and for large enough $n$.
  Hence, with probability at least $1 - 2 n^{1 - \tilde C/4}$, we have
  \begin{align}
    \min_{i} \mathcal Q_i (a (n^{-1} \ln n)^{\dxi}) \geqslant  (\underline \kappa a^{d_\xi} - \tilde C) (n^{-1} \ln n)^{1/2}. \label{eq: A aux min Q}
  \end{align}
  Take some fixed $\tilde C > 4$. In this case, \eqref{eq: A aux min Q} holds with probability approaching one, and take any fixed $a_C > \left(\underline \kappa^{-1} (C + \tilde C)\right)^{1/d_\xi}$. For such $a_C$, $(\underline \kappa a^{d_\xi} - \tilde C) (n^{-1} \ln n)^{1/2} > C (n^{-1} \ln n)^{1/2}$. As a result, we conclude that with probability approaching one, we have $\min_{i} \mathcal Q_i (a (n^{-1} \ln n)^{\dxi}) > C (n^{-1} \ln n)^{1/2}$, which completes the proof.
\end{proof}

\begin{AppLem}
  \label{lem: Y* convergence}
  Suppose that the hypotheses of Theorem \ref{the: Y star convergence} are satisfied. Then we have:
  \begin{enumerate}[(i)]
    \item $\max_{i} \max_{i' \in \hatN_i(n_i)} n^{-1} \sum_{\ell} (Y_{i \ell}^* - Y_{i' \ell}^*)^2 = O_p\left((n^{-1} \ln n)^{\dxi}\right)$;

    \item $\max_i \max_{i' \in \hatN_i(n_i)} \max_k \abs{n^{-1} \sum_{\ell} Y_{k \ell}^* (Y_{i' \ell}^* - Y_{i \ell}^*)}= O_p\left((n^{-1} \ln n)^{\dxi}\right)$;

    \item $\max_{i} \max_{i' \in \hatN_i(n_i)} \int (g(\xi_i, \xi) - g(\xi_{i'}, \xi))^2 dP_\xi (\xi) = O_p\left((n^{-1} \ln n)^{\dxi}\right)$.
  \end{enumerate}
\end{AppLem}

\begin{proof}[Proof of Lemma \ref{lem: Y* convergence}]
   Note that both statements of the lemma trivially hold when $i = i'$. For this reason, below $\max_{i' \in \hatN_i(n_i)}$ should be read as $\max_{i' \in \hatN_i (n_i) : i' \neq i}$.

  \noindent
  \textit{Proof of Part (i).}
  First, take some $C > \overline C$, where $\overline C$ is defined in the text of Theorem~\ref{the: Y star convergence}. Lemma \ref{lem: N(i)} guarantees, that there exist some $a_C$ and $\delta_{n,C} = a_C (n^{-1} \ln n)^{\dxi}$ such that with probability approaching one, we have $\min_{i} \abs{\mathcal B_i (\delta_{n,C})} \geqslant C (n \ln n)^{1/2}$.
  Since we have $\max_{i} n_i \leqslant \overline C (n \ln n)^{1/2} < C (n \ln n)^{1/2}$,
  we conclude that $\min_{i} \abs{\mathcal B_i (\delta_{n,C})} > n_i$ holds with probability approaching one. For any agents $i$ and $i' \in \mathcal \hatN_i (n_i)$, let $k$ and $k'$ be two agents (other than $i$ and $i'$) such that $k \in \mathcal B_i (\delta_{n,C})$ and $k' \in \mathcal B_{i'} (\delta_{n,C})$ (by Lemma \ref{lem: N(i)}, this happens simultaneously for all $i$ and $i'$ with probability approaching one for large enough $n$). By boundedness of $Y^*$ and Assumption \ref{ass: basic}\eqref{item: lipschitz g}, for $\mathcal G \coloneqq \norm{Y^*}_{\infty} \overline G$, we have
  \begin{align}
    \vert{n^{-1} \sum_{\ell} Y_{i \ell}^{*2} - n^{-1} \sum_{\ell} Y^*_{i \ell} Y^*_{k \ell}}\vert \leqslant \mathcal G \delta_{n,C}, \quad \vert{n^{-1} \sum_{\ell} Y^*_{i' \ell} Y^*_{i \ell} - n^{-1} \sum_{\ell} Y^*_{i' \ell} Y^*_{k \ell}}\vert \leqslant \mathcal G \delta_{n,C}, \notag \\
    \vert{n^{-1} \sum_{\ell} Y_{i' \ell}^{*2} - n^{-1} \sum_{\ell} Y^*_{i' \ell} Y^*_{k' \ell}}\vert \leqslant \mathcal G \delta_{n,C}, \quad
    \vert{n^{-1} \sum_{\ell} Y^*_{i \ell} Y^*_{i' \ell} - n^{-1} \sum_{\ell} Y^*_{i \ell} Y^*_{k' \ell}}\vert \leqslant \mathcal G \delta_{n,C} \label{eq: A G delta bound}
  \end{align}
  with probability approaching one for all $i$, $i' \in \hatN_i (n_i)$. Then
  \begin{align*}
    n^{-1} \sum_{\ell} (Y^*_{i \ell} - Y^*_{i' \ell})^2 &\leqslant \vert{n^{-1} \sum_{\ell} Y_{i \ell}^{*2} - n^{-1} \sum_{\ell} Y^*_{i' \ell} Y^*_{i \ell} }\vert + \vert{n^{-1} \sum_{\ell} Y_{i' \ell}^{*2} - n^{-1} \sum_{\ell} Y^*_{i \ell} Y^*_{i' \ell} }\vert \\
    &\leqslant \vert{n^{-1} \sum_\ell (Y^*_{i \ell} - Y^*_{i' \ell}) Y^*_{k \ell} }\vert + \vert{n^{-1} \sum_\ell (Y^*_{i' \ell} - Y^*_{i \ell}) Y^*_{k' \ell} }\vert + 4 \mathcal G \delta_{n,C}.
  \end{align*}
  Note that since $Y^*$ is bounded, there exists $\Delta > 0$ such that for all $i$, $i'$, and $k$,
  \begin{align}
    \vert{n^{-1} \sum_\ell (Y^*_{i \ell} - Y^*_{i' \ell}) Y^*_{k \ell} }\vert \leqslant \vert{(n-3)^{-1} \sum_{\ell \neq i,i',k} (Y^*_{i \ell}- Y^*_{i' \ell}) Y^*_{k \ell} }\vert + \frac{\Delta}{n}.\label{eq: A bounded Y Delta n}
  \end{align}
  Moreover, as usual, combining Corollary \ref{cor: bernstein} with the union bound, we have
  \begin{align}
    r_n \coloneqq \max_{i \neq i' \neq k} \abs{\frac{\sum_{\ell \neq i,i',k} (Y_{i \ell}- Y_{i' \ell}) Y_{k \ell} - \sum_{\ell \neq i,i',k} (Y^*_{i \ell}- Y^*_{i' \ell}) Y^*_{k \ell}}{n-3}} = O_p \left(\left(\frac{\ln n}{n}\right)^{1/2}\right). \label{eq: A r_n}
  \end{align}
  So, we have
  \begin{align}
    &\max_{i} \max_{i' \in \hatN_i(n_i)} n^{-1} \sum_{\ell} (Y_{i \ell}^* - Y_{i' \ell}^*)^2 \leqslant \max_{i} \max_{i' \in \hatN_i(n_i)} \max_{k \neq i,i'} \vert{(n-3)^{-1} \sum_{\ell \neq i,i',k} (Y_{i \ell}- Y_{i' \ell}) Y_{k \ell} }\vert \notag \\ + &\max_{i} \max_{i' \in \hatN_i(n_i)} \max_{k' \neq i,i'} \vert{(n-3)^{-1} \sum_{\ell \neq i,i',k} (Y_{i' \ell}- Y_{i \ell}) Y_{k' \ell} }\vert  + 4 \mathcal G \delta_{n,C} + \frac{2 \Delta}{n} + 2 r_n \notag \\
    = &2 \max_{i} \max_{i' \in \hatN_i(n_i)} \hat d_{\infty}^2(i,i') + O_p((n^{-1} \ln n)^{\dxi}), \label{eq: A aux neigh bound}
  \end{align}
  where the last line uses the definitions of $\hat d_{\infty}^2(i,i')$ and $\delta_{n,C} = a_C (n^{-1} \ln n)^{\dxi}$.

  Next, let $\hat d_{\infty, *}^2(i, i') \coloneqq \max_{k \neq i,i'} \vert{(n-3)^{-1} \sum_{\ell \neq i,i',k} (Y_{i \ell}^* - Y_{i' \ell}^*) Y_{k \ell}^* }\vert$. For any $i$ and $i'$, \eqref{eq: A r_n} guarantees that $\max_{i} \max_{i' \neq i} \vert{\hat d_{\infty}^2(i,i') - \hat d_{\infty, *}^2 (i,i')}\vert = O_p\left((n^{-1} \ln n)^{1/2}\right)$.  Analogously to \eqref{eq: A G delta bound}, for all $i$ and $i' \in \mathcal B_i (\delta_{n,C})$, we have $\hat d_{\infty, *}^2 (i,i') \leqslant \mathcal G \delta_{n,C}$ and hence also
  \begin{align*}
    \hat d_{\infty}^2 (i,i') \leqslant \mathcal G \delta_{n,C} + O_{p,n}\left((n^{-1} \ln n)^{1/2}\right).
  \end{align*}
  Then, by definition of $\mathcal N_i$ and the fact that for all $i$ we have $n_i < \abs{\mathcal B_{i} (\delta_{n,C})}$ with probability approaching one, we conclude
  \begin{align}
    \max_{i} \max_{i' \in \hatN_i(n_i)} \hat d_{\infty}^2(i,i') &\leqslant \mathcal G \delta_{n,C} + O_{p,n}\left((n^{-1} \ln n)^{1/2}\right) = O_{p,n}\left((n^{-1} \ln n)^{\dxi}\right), \label{eq: A tilde d rate}
  \end{align}
  which together with \eqref{eq: A aux neigh bound} completes the proof of Part (i).

  \noindent
  \textit{Proof of Part (ii).} We prove the result by bounding the quantity of interest for $k \neq i, i'$ and then for $k = i$ (or $k = i'$). 
  First, we consider $k \neq i, i'$. Then, using \eqref{eq: A bounded Y Delta n} and \eqref{eq: A r_n},
  \begin{small}
  \begin{align}
    &\max_i \max_{i' \in \hatN_i(n_i)} \max_{k \neq i,i'} \vert{n^{-1} \sum_{\ell} Y_{k \ell}^* (Y_{i' \ell}^* - Y_{i \ell}^*)}\vert = \max_i \max_{i' \in \hatN_i(n_i)} \max_{k \neq i,i'} \vert{(n-3)^{-1} \sum_{\ell \neq i, i', k} Y_{k \ell} (Y_{i' \ell} - Y_{i \ell})}\vert \nonumber \notag\\ + &O_p\left(\left({\ln n}/{n}\right)^{1/2}\right) = \max_{i} \max_{i' \in \hatN_i(n_i)} \hat d_{\infty}^2(i,i') + O_p\left(\left({\ln n}/{n}\right)^{1/2}\right) = O_{p,n}\left((n^{-1} \ln n)^{\dxi}\right), \label{eq: A k neq i case}
  \end{align}
  \end{small}
  where the last equality uses \eqref{eq: A tilde d rate}.

  To complete the proof, consider $k = i$ (the case $k = i'$ is analogous). Then, as argued at the beginning of the proof of Part (i), for any pair agents $i$ and $i'$, there exist some $\tilde k \in \mathcal B_i (\delta_{n,C})$ (other than $i$ and $i'$) with probability approaching one. Then, by boundedness of $Y^*$ and Assumption \ref{ass: basic}\eqref{item: lipschitz g},
  \begin{align*}
    \max_i \max_{i' \in \hatN_i (n_i)} \vert{n^{-1} \sum_\ell Y_{i \ell}^* (Y_{i' \ell}^* - Y_{i \ell}^*) }\vert &=\max_i \max_{i' \in \hatN_i (n_i)} \vert{n^{-1} \sum_\ell Y_{\tilde k \ell}^* (Y_{i' \ell}^* - Y_{i \ell}^*) }\vert + O(\delta_{n,C}) \\&= O_{p,n}\left((n^{-1} \ln n)^{\dxi}\right),
  \end{align*}
  where the second equality follows from \eqref{eq: A k neq i case}. This completes the proof of Part (ii).

  \noindent
  \textit{Proof of Part (iii).} Recall that for $i' \in \hatN_i (n_i)$, we have $Y_{i \ell }^* - Y_{i' \ell}^* = g(\xi_i,\xi_\ell) - g(\xi_{i'},\xi_\ell)$. Hence, the result of Part (i) reads as
  \begin{align}
    \label{eq: A L2 sum}
    \max_{i} \max_{i' \in \hatN_i(n_i)} n^{-1} \sum_{\ell} (g(\xi_i, \xi_\ell) - g(\xi_{i'},\xi_\ell))^2 = O_p\left((n^{-1} \ln n)^{\dxi}\right).
  \end{align}
  Applying Bernstein inequality \ref{the: A bounded Bernstein} conditional on $\xi_i$ and $\xi_{i'}$ and the union bound, we obtain
  \begin{align}
    \label{eq: A L2 UC}
    \max_{i,i'} \Big\vert{n^{-1} \sum_{\ell} (g(\xi_i, \xi_\ell) - g(\xi_{i'},\xi_\ell))^2 - \int (g(\xi_i, \xi) - g(\xi_{i'}, \xi))^2 dP_\xi (\xi)}\Big\vert = O_p\left(\left({\ln n}/{n}\right)^{1/2}\right).
  \end{align}
  Finally, \eqref{eq: A L2 sum} and \eqref{eq: A L2 UC} together deliver the result.
\end{proof}

\begin{proof}[Proof of Theorem \ref{the: Y star convergence}]
  \noindent
  \textit{Proof of Part (i).} First, notice that
  \begin{align}
    &n^{-1} \sum_{j} \left(\hat Y_{ij}^* - Y_{ij}^*\right)^2 = n^{-1} \sum_j \left(\frac{\sum_{i' \in \hatN_i (n_i)} (Y_{i' j} - Y_{ij}^*) }{n_i}\right)^2 \leq 2 (\mathcal S_i + \mathcal J_i),  \label{eq: A 2 S J bound}  \\
    &\mathcal S_{i} \coloneqq \frac{1}{n} \sum_j \left(\frac{\sum_{i' \in \hatN_i (n_i)} (Y_{i' j} - Y_{i' j}^*) }{n_i} \right)^2, \quad \mathcal J_i \coloneqq \frac{1}{n} \sum_j \left(\frac{\sum_{i' \in \hatN_i (n_i)} (Y_{i' j}^* - Y_{i j}^*) }{n_i} \right)^2. \notag
  \end{align}
  We start with bounding $\mathcal S_i$. Note that, for all $i'$ and $j$, $Y_{i' j} - Y_{i' j}^* = \varepsilon_{i' j}$, with the convention $\varepsilon_{i' j} = - Y^*_{i' i'}$ when $i' = j$. Then $\mathcal S_{i}$ can be decomposed as follows
  \begin{align*}
    \mathcal S_{i} &= \frac{1}{n n_i^2} \sum_{j} \left(\sum_{i' \in \hatN_i (n_i)} \varepsilon_{i' j}^2 + \sum_{i' \in \hatN_i (n_i)} \sum_{i'' \in \hatN_i(n_i), i'' \neq i'} \varepsilon_{i' j} \varepsilon_{i'' j}  \right) \\
    &= \underbrace{\frac{1}{n_i^2} \sum_{i' \in \hatN_i (n_i)} \frac{1}{n} \sum_{j} \varepsilon_{i' j}^2}_{\coloneqq \mathcal S_{1i}} + \underbrace{\frac{1}{n_i^2} \sum_{i' \in \hatN_i (n_i)} \sum_{i'' \in \hatN_i(n_i), i'' \neq i'} \frac{1}{n} \sum_{j} \varepsilon_{i' j} \varepsilon_{i'' j}}_{\coloneqq \mathcal S_{2i}}.
  \end{align*}
  First, $\frac{1}{n} \sum_{j} \varepsilon_{i' j}^2 = \frac{1}{n} \sum_{j \neq i'} \varepsilon_{i' j}^2 - \frac{Y_{i' i'}^*}{n}$.
  Applying Corollary \ref{cor: bernstein} (conditional on $Z_{i'}$) and the union bound gives $\max_{i'} \vert{\frac{1}{n-1} \sum_{j \neq i'} \varepsilon_{i' j}^2 - \sigma_{i'}^2 }\vert = O_p\left(\left({\ln n}/{n}\right)^{1/2}\right)$, where $\sigma_i^2 \coloneqq \ex{\varepsilon_{ij}^2|Z_i} \leqslant C$ a.s. by Assumption \ref{ass: basic}\eqref{item: eps moments}. Finally, since $Y^*$ is bounded, we conclude
  \begin{align}
    \label{eq: A S1 bound}
    \max_{i'} \frac{1}{n} \sum_j \varepsilon_{i' j}^2 \leqslant C + O_p\left(\left({\ln n}/{n}\right)^{1/2}\right), \quad \max_{i} \mathcal S_{1i} \leqslant \frac{C}{n_i} + o_p (n_i^{-1}).
  \end{align}
  To bound $\mathcal S_{2i}$, note $\frac{1}{n} \sum_{j} \varepsilon_{i' j} \varepsilon_{i'' j} = \frac{1}{n} \sum_{j \neq i',i''} \varepsilon_{i' j} \varepsilon_{i'' j} - \frac{1}{n} (Y_{i' i'}^* \varepsilon_{i'' i'} + \varepsilon_{i' i''} Y_{i'' i''}^*)$.
  First, as usual, applying Bernstein inequality \ref{the: A unb Bernstein} (conditional on $Z_{i'}$ and $Z_{i''}$) and the union bound gives $\max_{i' \neq i''} \vert{\frac{1}{n-2} \sum_{j \neq i',i''} \varepsilon_{i' j} \varepsilon_{i'' j}}\vert = O_p\left(\left({\ln n}/{n}\right)^{1/2}\right)$. Similarly, since $\varepsilon_{i' i''}$'s are (uniformly) sub-Gaussian, we have $\max_{i' \neq i''} \abs{\frac{1}{n} \varepsilon_{i' i''}} = O_p\left(n^{-1} (\ln n)^{1/2} \right)$. Since $Y^*$ is bounded, we conclude
  \begin{align}
    \label{eq: A S2 bound}
    \max_{i' \neq i''} \vert{\frac{1}{n} \sum_j \varepsilon_{i' j} \varepsilon_{i'' j}}\vert = O_p\left(\left({\ln n}/{n}\right)^{1/2}\right), \quad \max_{i} \abs{\mathcal S_{2i}} = O_p\left(\left({\ln n}/{n}\right)^{1/2}\right).
  \end{align}
  Combining \eqref{eq: A S1 bound} and \eqref{eq: A S2 bound} gives $\max_{i} \mathcal S_{i} = O_p\left(\left({\ln n}/{n}\right)^{1/2}\right)$.

  Finally, note $\mathcal J_i \leqslant \frac{1}{n} \sum_j \frac{1}{n_i} \sum_{i' \in \hatN_i (n_i)} (Y_{i'j}^* - Y_{ij}^*)^2 = \frac{1}{n_i} \sum_{i' \in \hatN_i (n_i)} \frac{1}{n} \sum_{j} (Y_{i'j}^* - Y_{ij}^*)^2$.
  Applying Lemma \ref{lem: Y* convergence}(i), we obtain $\max_{i} \mathcal J_i = O_p\left((n^{-1} \ln n)^{\dxi}\right)$. This, together with the bound on $\mathcal S_i$ and \eqref{eq: A 2 S J bound}, delivers the desired result.

  \noindent
  \textit{Proof of Part (ii).}
  First, notice that
  \begin{small}
  \begin{align*}
      &n^{-1} \sum_\ell Y_{k \ell}^* (\hat Y_{i \ell}^* - Y_{i \ell}^*)
     = n_i^{-1} \sum_{i' \in \hatN_i(n_i)} \Big(n^{-1} \sum_\ell Y_{k \ell}^* (Y_{i' \ell}^* - Y_{i \ell}^*) + n^{-1} \sum_\ell Y_{k \ell}^* \varepsilon_{i' \ell}\Big),\\
    &\max_{k} \max_{i} \vert{n^{-1} \sum_\ell Y_{k \ell}^* (\hat Y_{i \ell}^* - Y_{i \ell}^*)}\vert \leqslant \underbrace{\max_i \max_{i' \in \hatN_i(n_i)} \max_k \vert{n^{-1} \sum_{\ell} Y_{k \ell}^* (Y_{i' \ell}^* - Y_{i \ell}^*)}\vert}_{= O_p\left(\left(\frac{\ln n}{n}\right)^{\dxi}\right)} + \max_{k} \max_{i'} \vert{ n^{-1}\sum_\ell Y_{k \ell}^* \varepsilon_{i' \ell}}\vert,
  \end{align*}
  \end{small}
  where the bound on the first term on the right-hand side follows from Lemma \ref{lem: Y* convergence}(ii). To complete the proof, we need to to bound the second term, which can be represented as
  \begin{align}
    n^{-1} \sum_\ell Y_{k \ell}^* \varepsilon_{i' \ell} = n^{-1} \sum_{\ell \neq k, i'} Y_{k \ell}^* \varepsilon_{i' \ell} + n^{-1} (Y_{k k}^* \varepsilon_{i' k} + Y_{k i'}^* \varepsilon_{i' i'}). \label{eq: A aux decomposition Part ii}
  \end{align}
  Again, applying Corollary \ref{cor: bernstein} (conditional on $Z_k$ and $Z_{i'}$) and the union bound, we obtain $\max_{k} \max_{i'} \vert{n^{-1} \sum_{\ell \neq k,i'} Y_{k \ell}^* \varepsilon_{i' \ell}}\vert = O_p\left(\left({\ln n}/{n}\right)^{1/2}\right)$. Similarly, since either $\varepsilon_{i' k}$ is (uniformly) sub-Gaussian (for $i' \neq k$) or is equal to $-Y_{k k}^*$ (which is bounded), we have $\max_{k} \max_{i'} \abs{n^{-1} Y_{kk}^* \varepsilon_{i' k}} = O_p\left(n^{-1} (\ln n)^{1/2}\right)$. Finally, since $\varepsilon_{i' i'} = - Y_{i' i'}^*$ and $Y^*$ is bounded, we have $\max_{k} \max_{i'} \abs{n^{-1} Y_{k i'}^* \varepsilon_{i' i'}} \leqslant C/n$. These bounds, together with \eqref{eq: A aux decomposition Part ii}, imply that $\max_{k} \max_{i'} \vert{ n^{-1}\sum_\ell Y_{k \ell}^* \varepsilon_{i' \ell}}\vert = O_p\left(\left({\ln n}/{n}\right)^{1/2}\right)$, which completes the proof.
  \end{proof}

\subsubsection{Proof of Theorem \ref{the: disc X d star rate}}

First, we state and prove the following auxiliary lemma.

\begin{AppLem}
  \label{lem: d* R_p}
  Suppose $\mathcal B = \mathbb R^p$. Then, under Assumptions \ref{ass: DGP}, \ref{ass: basic}, \ref{ass: d* rate}\eqref{item: disc X}, we have $\max_{i \neq j} \abs{d_{ij,n}^2 - d_{ij}^2} = O_p\left(\left({\ln n}/{n}\right)^{1/2}\right)$, where $d_{ij,n}^2 \coloneqq \min_{\beta \in \mathbb R_p} d_{ij,n}^2 (\beta)$, and $d_{ij,n}^2 (\beta)$ is given by \eqref{eq: d ij n star}.
\end{AppLem}

\begin{proof}[Proof of Lemma \ref{lem: d* R_p}]
  First, we argue that there exists some compact $\overline {\mathcal B}$ such that, for all pairs of agents, $d_{ij,n}^2 \coloneqq \min_{\beta \in \mathbb R^p} d_{ij,n}^2 (\beta) = \min_{\beta \in \overline{\mathcal B}} d_{ij,n}^2 (\beta)$ with probability approaching one.
  For all $x, \tilde x \in \supp{X}$, let $\Delta w(X_k;x,\tilde x) \coloneqq w(x,X_k) - w(\tilde x,X_k)$ and
  \begin{align*}
    \hatC (x, \tilde x) \coloneqq \frac{1}{n} \sum_{k=1}^n \Delta w(X_k;x,\tilde x) \Delta w(X_k;x,\tilde x)' = \sum_{r=1}^R \hat p_r H(x_r;x, \tilde x)',
  \end{align*}
  where $\hat p_r = n^{-1} \sum_{k=1}^n \ind\{X_k = x_r\}$ and $H (x_r;x,\tilde x) \coloneqq \Delta w(x_r; x, \tilde x) \Delta w(x_r;x, \tilde x)'$.

  Notice that since $\supp{X} = \{x_1, \dots, x_R\}$, there exists $C_\lambda > 0$ such that the minimal non-zero eigenvalue of $H(x_r;x,\tilde x)$ is greater than $C_\lambda$ uniformly over all $x_r,x,\tilde x \in \{x_1, \dots, x_R\}$. Formally, we have $\lambda_{min,+} (H(x_r;x,\tilde x)) > C_\lambda$ for all $x_r, x, \tilde x \in \{x_1, \dots, x_R\}$ such that $H(x_r;x,\tilde x)$ is non-zero, where $\lambda_{min,+}(H)$ denotes the minimal positive eigenvalue of a positive semidefinite (non-zero) matrix $H$. Next, notice that there exists $C_p > 0$ such that with probability approaching one, we have $\min_{r} \hat p_r > C_p$. Hence, we conclude that for some $C_{\mathcal C} > 0$ we have, with probability approaching one, $\lambda_{min,+} (\hatC (x,\tilde x)) > C_{\mathcal C}$
  for all $x, \tilde x \in \{x_1, \dots, x_R\}$ such that $\hatC (x,\tilde x)$ is non-zero. Next, for all $x, \tilde x \in \{x_1, \dots, x_R\}$, let $\hat O (x, \tilde x)$ be an orthogonal matrix, which diagonalizes $\hatC (x,\tilde x)$, i.e., $\hat O (x,\tilde x) \hatC (x, \tilde x) \hat O(x, \tilde x)'= \hat \Lambda(x,\tilde x)$, where $\hat \Lambda(x,\tilde x)$ is diagonal. Notice that non-zero elements $\hat \Lambda(x,\tilde x)$ of are bounded away from zero by $C_{\mathcal C}$ with probability approaching (uniformly over $x, \tilde x \in \{x_1, \dots, x_R\}$).

  Take any pair of agents $i$ and $j$. Since $\hat O(X_i,X_j)' \hat O(X_i,X_j)$ is an identity matrix, 
  \begin{align*}
    d_{ij,n}^2 
    &= \min_{\beta \in \mathbb R_p} \frac{1}{n} \sum_{k} (Y_{ik}^* - Y_{jk}^* - \Delta w (X_k;X_i,X_j)' \hat O(X_i,X_j)' \hat O(X_i,X_j) \ \beta )^2 \\
    &= \min_{b \in \mathbb R_p} \frac{1}{n} \sum_{k} (Y_{ik}^* - Y_{jk}^* - \Delta w (X_k;X_i,X_j)' \hat O(X_i,X_j)' b)^2,
  \end{align*}
  where the last equality follows from the change of the variable $b = \hat O(X_i,X_j) \beta$. The minimum is achieved at
  \begin{align*}
    \hat b_{ij} = \left(\hat \Lambda (X_i,X_j) \right)^{+} \left(\hat O(X_i,X_j) n^{-1} \sum_k \Delta w (X_k;X_i,X_j) (Y_{ik}^* - Y_{jk}^*) \right),
  \end{align*}
  where $+$ denotes the Moore–Penrose inverse. Notice that since the non-zero elements of diagonal matrix $\hat \Lambda(X_i,X_j)$ are bounded away from zero (uniformly in $X_i$ and $X_j$) with probability approaching one, we conclude that $\max_{i \neq j}\lambda_{max}\left(\left(\hat \Lambda(X_i,X_j)\right)^+\right) < \overline C_\Lambda$ with probability approaching one. Since $w$ and $Y^*$ are bounded,  we also have $\max_{i \neq j} \norm{\hat O(X_i,X_j) n^{-1} \sum_k w (X_k;X_i,X_j) (Y_{ik}^* - Y_{jk}^*)} < C$, and thus we conclude that for some $C_\beta > 0$, we have with probability approaching one $\max_{i \neq j} \Vert{\hat b_{ij}}\Vert \leqslant C_\beta$.
  
  Since $b = \hat O(X_i,X_j) \beta$, the minimum of $d_{ij,n}^2(\beta)$ is achieved at $\hat \beta_{ij} = \hat O(X_i,X_j)' \hat b_{ij}$, i.e., $d_{ij,n}^2 = d_{ij,n}^2 (\hat \beta_{ij})$, and, with probability approaching one, we also have $\max_{i \neq j} \Vert{\hat \beta_{ij}}\Vert \leqslant C_\beta$.
  Then, if $\overline{\mathcal B}$ includes all $\beta$ such that $\norm{\beta} \leqslant C_\beta$, we have $d_{ij,n}^2 = d_{ij,n}^2 (\hat \beta_{ij}) = \min_{\beta \in \overline{\mathcal B}} d_{ij,n}^2 (\beta)$. Such a compact set $\overline{\mathcal B}$ clearly exists, the proof of the first part is complete.

  Second, note that by the same argument, we have $d_{ij}^2 = \min_{\beta \in \mathbb R_p} d_{ij}^2(\beta) = \min_{\beta \in \overline{\mathcal B}} d_{ij}^2 (\beta)$, where, without loss of generality, $\overline{\mathcal B}$ can be taken the same as before.
  
  Finally, since $\overline{\mathcal B}$ is compact, we can apply the result of Lemma \ref{lem: A d_n UC}, which guarantees that $\max_{i \neq j} \sup_{\beta \in \overline{\mathcal B}} \abs{d_{ij,n}^2(\beta) - d_{ij}^2(\beta)} = O_p\left((\ln n/n)^{1/2}\right)$.  This necessarily implies that
  \begin{align*}
     O_p\left(( \ln n / n)^{1/2}\right) = \max_{i \neq j} \vert{\min_{\beta \in \overline{\mathcal B}} d_{ij,n}^2(\beta) - \min_{\beta \in \overline{\mathcal B}} d_{ij}^2(\beta)}\vert = \max_{i \neq j} \abs{d_{ij,n}^2 - d_{ij}^2},
  \end{align*}
  which completes the proof.
\end{proof}
 
\begin{proof}[Proof of Theorem \ref{the: disc X d star rate}]
  First, denote
  \begin{align*}
    \mathbf W_{i-j} = \begin{pmatrix}
      (W_{i1} - W_{j1})' \\
      (W_{i2} - W_{j2})' \\
      \cdots \\
      (W_{in} - W_{jn})'
    \end{pmatrix},
    \quad
    \mathbf Y^*_{i-j} = \begin{pmatrix}
      Y^*_{i1} - Y^*_{j1} \\
      Y^*_{i2} - Y^*_{j2} \\
      \cdots \\
      Y^*_{in} - Y^*_{jn}
    \end{pmatrix},
    \quad
    \hat{\mathbf Y}_{i-j}^* = \begin{pmatrix}
      \hat Y^*_{i1} - \hat Y^*_{j1} \\
      \hat Y^*_{i2} - \hat Y^*_{j2} \\
      \cdots \\
      \hat Y^*_{in} - \hat Y^*_{jn}
    \end{pmatrix}.
  \end{align*}
  Then
  \begin{align*}
    \hat d_{ij}^2 = n^{-1} \left({\mathbf Y}_{i-j}^{* \prime} \mathbf P_{i-j} {\mathbf Y}_{i-j}^* + \Delta \hat{\mathbf Y}_{i-j}^{* \prime} \mathbf P_{i-j} {\mathbf Y}_{i-j}^* + {\mathbf Y}_{i-j}^{* \prime} \mathbf P_{i-j} \Delta \hat{\mathbf Y}_{i-j}^* + \Delta \hat {\mathbf Y}_{i-j}^{* \prime} \mathbf P_{i-j} \Delta \hat{\mathbf Y}_{i-j}^*\right),
  \end{align*}
  where $\Delta \hat{\mathbf Y}_{i-j}^* = \hat{\mathbf Y}_{i-j}^* - {\mathbf Y}_{i-j}^*$ and $\mathbf P_{i-j} = \mathbf I_n - \mathbf W_{i-j} (\mathbf W_{i-j}' \mathbf W_{i-j})^+ \mathbf W_{i-j}'$, where $+$ stands for the Moore–Penrose inverse.
  Note that $n^{-1} {\mathbf Y}_{i-j}^{* \prime} \mathbf P_{i-j} {\mathbf Y}_{i-j}^* = d_{ij,n}^2$, where $d_{ij,n}^2$ is as defined in Lemma \ref{lem: d* R_p}, which ensures $\max_{i \neq j} \abs{d_{ij,n}^2 - d_{ij}^2} = O_p((\ln n/n)^{1/2})$. Hence, to complete the proof, it would be sufficient to show that $\max_{i \neq j} \vert{n^{-1} \Delta \hat{\mathbf Y}_{i-j}^{* \prime} \mathbf P_{i-j} {\mathbf Y}_{i-j}^*}\vert = O_p\left(\left({\ln n}/{n}\right)^\dxi\right)$, and $\max_{i \neq j} \vert{n^{-1} \Delta \hat{\mathbf Y}_{i-j}^{* \prime} \mathbf P_{i-j} \Delta \hat{\mathbf Y}_{i-j}^*}\vert = O_p\left(\left({\ln n}/{n}\right)^\dxi\right)$. To inspect these, note
  \begin{align*}
    \max_{i \neq j} \vert{n^{-1} \Delta \hat{\mathbf Y}_{i-j}^{* \prime} \mathbf P_{i-j} {\mathbf Y}_{i-j}^*}\vert &\leqslant \max_{i \neq j} \vert{n^{-1} \Delta \hat{\mathbf Y}_{i-j}^{* \prime} {\mathbf Y}_{i-j}^*}\vert = \max_{i \neq j} \vert{n^{-1} \sum_\ell (Y_{i \ell}^* - Y_{j \ell}^*) (\Delta \hat Y_{i \ell}^* - \Delta \hat Y_{j \ell}^*) }\vert \\
    & \leqslant 4 \max_k \max_i \vert{n^{-1} \sum_\ell Y_{k \ell}^* (\hat Y_{i \ell}^* - Y_{i \ell}^*) }\vert = O_p\left(\left({\ln n}/{n}\right)^{\frac{1}{2 d_\xi}}\right),
  \end{align*}
  where the last equality follows from Theorem \ref{the: Y star convergence}(ii).
  Similarly,
  \begin{align*}
    \max_{i \neq j} \vert{n^{-1} \Delta \hat{\mathbf Y}_{i-j}^{* \prime} \mathbf P_{i-j} \Delta \hat {\mathbf Y}_{i-j}^*}\vert &\leqslant \max_{i \neq j} \vert{n^{-1} \Delta \hat{\mathbf Y}_{i-j}^{* \prime} \Delta \hat {\mathbf Y}_{i-j}^*}\vert = \max_{i \neq j} \vert{n^{-1} \sum_\ell (\Delta \hat Y_{i \ell}^* - \Delta \hat Y_{j \ell}^*)^2 }\vert \\
    & \leqslant 4 \max_i n^{-1} \sum_\ell (\hat Y_{i \ell}^* - Y_{i \ell}^*)^2 = O_p\left(\left({\ln n}/{n}\right)^{\frac{1}{2 d_\xi}}\right),
  \end{align*}
  where the last equality is due Theorem \ref{the: Y star convergence}(i), which completes the proof.
\end{proof}

\subsubsection{Proof of Lemma \ref{lem: d star rate}}
\begin{proof}[Proof of Lemma \ref{lem: d star rate}]
  Let $\hat d_{ij}^2 (\beta) \coloneqq \frac{1}{n} \sum_{k} (\hat Y_{ik}^* - \hat Y_{jk}^* - (W_{ik} - W_{jk}) \beta )^2$,
  so $\hat d_{ij}^2 = \min_{\beta \in \mathcal B} \hat d_{ij}^2 (\beta)$.
  \begin{align*}
    \max_{i \neq j} \vert{\hat d_{ij}^2 - d_{ij}^2}\vert = \max_{i \neq j} \vert{\min_{\beta \in \mathcal B} \hat d_{ij}^2 (\beta) - \min_{\beta \in \mathcal B} d_{ij}^2 (\beta) }\vert \leqslant 2 \max_{i \neq j} \max_{\beta \in \mathcal B} \vert{\hat d_{ij}^2 (\beta) - d_{ij}^2(\beta)}\vert.
  \end{align*}
  Hence, it suffices to show $\max_{i \neq j} \max_{\beta \in \mathcal B} \Vert{\hat d_{ij}^2(\beta) - d_{ij}^2(\beta)}\Vert = O_p\left(\left({\ln n}/{n}\right)^{1/2} + \mathcal R_n^{-1/2}\right)$.
  Let $\Delta \hat Y_{ik}^* \coloneqq \hat Y_{ik}^* - Y_{ik}^*$, and note that $\hat d_{ij}^2 (\beta) - d_{ij}(\beta)$ can be decomposed as
  \begin{align}
    &\hat d_{ij}^2 (\beta) - d_{ij}(\beta) = d_{ij,n}^2(\beta) - d_{ij}(\beta)\notag\\
    &+ \frac{2}{n} \sum_{k} \left(Y_{ik}^* - Y_{jk}^* + (W_{ik} - W_{jk})' \beta\right) (\Delta \hat Y_{ik}^* - \Delta \hat Y_{jk}^*)+ \frac{1}{n} \sum_{k} (\Delta \hat Y_{ik}^* - \Delta \hat Y_{jk}^*)^2.\label{eq: A aux d decomposition}
  \end{align}
  By the Cauchy-Schwartz inequality,
  \begin{small}
  \begin{align*}
    &\max_{i \neq j} \frac{1}{n} \sum_{k} (\Delta \hat Y_{ik}^* - \Delta \hat Y_{jk}^*)^2 \leqslant 4 \max_{i} \frac{1}{n} \sum_i (\hat Y_{ik}^* - Y_{ik}^*)^2 = O_p (\mathcal R_n^{-1}),\\
    &\max_{i \neq j} \max_{\beta \in \mathcal B} \abs{\frac{1}{n} \sum_{k} (Y_{ik}^* - Y_{jk}^* + (W_{ik} - W_{jk})'\beta) (\Delta \hat Y_{ik}^* - \Delta \hat Y_{jk}^*) } \leqslant C \sqrt{\max_{i \neq j}\frac{1}{n} \sum_{k} (\Delta \hat Y_{ik}^* - \Delta \hat Y_{jk}^*)^2}  = O_p (\mathcal R_n^{-1/2}),
  \end{align*}
  \end{small}
  where the second inequality guaranteed by the fact that $Y^*$, $W$, and $\mathcal B$ are bounded. These bounds, paired with Lemma \ref{lem: A d_n UC} and \eqref{eq: A aux decomposition Part ii}, deliver the result.
\end{proof}

\subsection{Proof of Theorem \ref{the: Y star UC}}

\begin{proof}[Proof of Theorem \ref{the: Y star UC}] First, we decompose $\tilde Y_{ij}^*$ as
  \begin{align*}
    \tilde Y_{ij}^* &= \frac{1}{m_{ij}} \sum_{(i',j') \in \hatM_{ij}} Y_{i' j'} = Y_{ij}^* + \frac{1}{m_{ij}} \sum_{(i',j') \in \hatM_{ij}}  \left( (g(\xi_{i'},\xi_{j'}) - g(\xi_i,\xi_j)) + \varepsilon_{i' j'} \right).
  \end{align*}
  To prove the result, it suffices to show that $\max_{i,j} \vert{\frac{1}{m_{ij}} \sum_{(i',j') \in \hatM_{ij}} (g(\xi_{i'},\xi_{j'}) - g(\xi_i,\xi_j))}\vert$ and $\max_{i,j} \vert{\frac{1}{m_{ij}} \sum_{(i',j') \in \hatM_{ij}} \varepsilon_{i' j'}}\vert$ are both $o_p(1)$.

  First, notice that Lemma \ref{lem: Y* convergence}(iii) combined with the hypotheses of the theorem implies that $\max_{i} \max_{i' \in \mathcal \hatN_i (n_i)} \norm{\xi_i - \xi_{i'}} = o_p(1)$. Hence, Using Assumption \ref{ass: basic}\eqref{item: lipschitz g},
  \begin{align*}
    \max_{i,j} \Big\vert{\frac{1}{m_{ij}} \sum_{(i',j') \in \hatM_{ij}} (g(\xi_{i'},\xi_{j'}) - g(\xi_i,\xi_j))}\Big\vert &\leqslant \max_{i,j} \max_{i' \in \hatN_i(n_i)} \max_{j' \in \hatN_j(n_j)} \abs{g(\xi_{i'},\xi_{j'}) - g(\xi_i,\xi_j)}\\
    &\leqslant 2 \overline G \max_{i} \max_{i' \in \mathcal \hatN_i (n_i)} \norm{\xi_i - \xi_{i'}} = o_p(1).
  \end{align*}
  
  Next, let $\overline{\mathcal M}$ denote the set of all possible realizations of $\hatM_{ij}$ formed according to \eqref{eq: hat Mij def}. Unlike $\hatM_{ij}$, $\overline{\mathcal M}$ is not random. Note that by construction 
  \begin{align}
    \label{eq: A aux bound using M}
    \max_{i,j} \Bigg\vert{\frac{1}{m_{ij}} \sum_{(i',j') \in \hatM_{ij}} \varepsilon_{i' j'}}\Bigg\vert \leqslant \max_{\mathcal M \in \overline{\mathcal M}} \Bigg\vert{\frac{1}{\abs{\mathcal M}} \sum_{(i',j') \in \mathcal M} \varepsilon_{i' j'}}\Bigg\vert.
  \end{align}
  Note that, for $\mathcal M \in \overline{\mathcal M}$, $\abs{\mathcal M} \geqslant \underline m \coloneqq \underline n (\underline n - 1)/2$, where $\underline n = \min_{i} n_i$. Applying Bernstein inequality \ref{the: A unb Bernstein} conditional on $\{X_i,\xi_i\}_{i=1}^n$ and using uniform sub-Gaussianity of the errors, we conclude that there exist some positive constants $C,a$, and $b$ such that for all $\epsilon > 0$ and for all $\mathcal M \in \overline{\mathcal M}$
  \begin{align}
    \pro{\abs{\frac{1}{\abs{\mathcal M}} \sum_{(i',j') \in \mathcal M} \varepsilon_{i' j'}} > \epsilon} &\leqslant C \exp\left(-\frac{\abs{\mathcal M} \epsilon^2}{a + b \epsilon}\right) \leqslant C \exp\left(-\frac{\underline m \epsilon^2}{a + b \epsilon}\right).\label{eq: A prob bound on eps M}
  \end{align}
  To apply the union bound, we need to bound the cardinality of $\overline{\mathcal M}$. Note that $\abs{\overline{\mathcal M}} \leqslant \abs{\overline{\mathcal N}}^2$, where $\abs{\overline{\mathcal N}}$ is the total number of possible realizations of $\mathcal{\hat N}_i (n_i)$ also allowing for possible variability in $n_i$. For a given $n_i$, the number of possible realizations of $\mathcal{\hat N}_i (n_i)$ is $\binom{n}{n_i} \leqslant \binom{n}{\overline n}$ with $\overline n = \max_i n_i$, where we also used monotonicity of $\binom{n}{k}$ for $k \leqslant \lceil n/2 \rceil$. Hence, factoring in variability in $n_i$ taking possible values from $\underline n$ to $\overline n$, we conclude that $\abs{\overline{\mathcal N}} \leqslant \overline n^2 \binom{n}{\overline n}$. Thus, combining the union bound with \eqref{eq: A prob bound on eps M}, we conclude
  \begin{align}
    \pro{\max_{\mathcal M \in \overline{\mathcal M}}\abs{\frac{1}{\abs{\mathcal M}} \sum_{(i',j') \in \mathcal M} \varepsilon_{i' j'}} > \epsilon} \leqslant \overline n^4 \binom{n}{\overline n}^2 C \exp\left(-\frac{\underline m \epsilon^2}{a + b \epsilon}\right).
    \label{eq: A unif Y tilde aux bound}
  \end{align}
  Since $\binom{n}{\overline n} < (\frac{n \times e}{\overline n})^{\overline n} $, we have $\ln \left(\binom{n}{\overline n}\right) < \overline n \ln (n) < \overline C n^{1/2} (\ln n)^{3/2}$ for sufficiently large $\overline n$. Since $\underline m > C_m n \ln n$ for some $C_m > 0$, we conclude that the probability in \eqref{eq: A unif Y tilde aux bound} goes to zero for any fixed $\epsilon > 0$, which, together with \eqref{eq: A aux bound using M}, completes the proof.
\end{proof}

\subsection{Proofs of the results of Section \ref{ssec: NID}}
\begin{proof}[Proof of Lemma \ref{lem: nonparam ID}]
  The proof is by contradiction. Suppose that for some $\mu_{ij}^*(x,\tilde x) \neq \msh (x,\tilde x)$, we have $\msd_{ij}^2 (x, \tilde x) = \msd_{ij}^2 (\mu_{ij}^* (x, \tilde x), x, \tilde x) = 0$, which implies
  \begin{align*}
    0 &= \ex{\left(Y_{ik}^* - Y_{jk}^* + \mu_{ij}^*(x,\tilde x) \right)^2|X_i,\xi_i,X_j,\xi_j,X_k = x} \\
      &=  \ex{\left(-\msh (x,\tilde x) + g(\xi_i,\xi_k) - g(\xi_j,\xi_k) + \mu_{ij}^*(x,\tilde x) \right)^2|X_i,\xi_i,X_j,\xi_j,X_k = x}.
  \end{align*}
  Since $\mathcal E_{x,\tilde x} \subseteq \supp{\xi_k|X_k=x}$, the above implies
  \begin{align}
    \label{eq: A delta g MSW}
    g(\xi_i,\xi_k) - g(\xi_j,\xi_k) = \msh (x,\tilde x) - \mu_{ij}^*(x,\tilde x) \neq 0 
  \end{align}
  for all $\xi_k \in \supp{\xi_k|X_k=x}$. Next, consider
  \begin{small}
  \begin{align}
    &\ex{\left(Y_{ik}^* - Y_{jk}^* - \mu_{ij}^*(x,\tilde x) \right)^2|X_i,\xi_i,X_j,\xi_j,X_k = \tilde x} \nonumber \\ 
    = &\ex{\left(\msh (x,\tilde x) + g(\xi_i,\xi_k) - g(\xi_j,\xi_k) - \mu_{ij}^*(x,\tilde x) \right)^2|X_i,\xi_i,X_j,\xi_j,X_k = \tilde x} \nonumber \\
    \geqslant &\pr(\xi_k \in \mathcal E_{x,\tilde x}|X_k = \tilde x) \ex{\left(\msh (x,\tilde x) + g(\xi_i,\xi_k) - g(\xi_j,\xi_k) - \mu_{ij}^*(x,\tilde x) \right)^2|X_i,\xi_i,X_j,\xi_j,X_k = \tilde x,\xi_k \in \mathcal E_{x,\tilde x}} \nonumber \\
    =& 4 \pr(\xi_k \in \mathcal E_{x,\tilde x}|X_k = \tilde x) (\msh (x,\tilde x) - \mu_{ij}^*(x,\tilde x))^2 > 0, \label{eq: A contr ineq}
  \end{align}
  \end{small}
  where the last equality follows from \eqref{eq: A delta g MSW}, and the last inequality follows from \eqref{eq: A delta g MSW} and Assumption \ref{ass: nonparam w}\eqref{item: w xi supp}. Note that \eqref{eq: A contr ineq} implies that $\msd_{ij}^2 (x, \tilde x) = \msd_{ij}^2 (\mu_{ij}^*(x,\tilde x),x,\tilde x) > 0$, which contradicts the initial hypothesis.
\end{proof}

\subsection{Bernstein inequalities}
In this section we specify the Bernstein inequalities, which we refer to in the proofs.

\subsubsection{Bernstein inequality for bounded random variables}

\begin{AppThe}[Bernstein inequality; see, for example, \citetsupp{bennett1962probability}]
	\label{the: A bounded Bernstein}
	Let $Z_1, \dots, Z_n$ be mean zero independent random variables. Assume there exists a positive constant $M$ such that $\abs{Z_i} \leqslant M$ with probability one for each $i$. Also let $\sigma^2 \coloneqq \frac{1}{n} \sumin \ex{Z_i^2}$. Then, for all $\epsilon > 0$ , we have $\pr\left(\abs{\frac{1}{n} \sumin Z_i} \geqslant \epsilon\right) \leqslant 2 \exp \left(- \frac{n \epsilon^2}{2 \left(\sigma^2 + \frac{1}{3} M \epsilon\right)}\right)$.
\end{AppThe}

\subsubsection{Bernstein inequality for unbounded random variables}

\begin{AppLem}[Moments of a sub-Gaussian random variable]
	\label{lem: A subG moments}
	Let $Z$ be a mean zero random variable satisfying $\ex{e^{\lambda Z}} \leqslant e^{v \lambda^2}$ for all $\lambda \in \mathbb R$, for some $v > 0$. Then for every integer $q \geqslant 1$, we have $\ex{Z^{2q}} \leqslant q! (4v)^{q}$.
\end{AppLem}

\begin{proof}[Proof of Lemma \ref{lem: A subG moments}]
	See Theorem 2.1 in \citetsupp{boucheron2013concentration}.
\end{proof}

\begin{AppThe}[Bernstein inequality for unbounded random variables]
	\label{the: A unb Bernstein}
	Let $Z_1, \dots, Z_n$ be independent random variables. Assume that there exist some positive constants $\nu$ and $c$ such that $\frac{1}{n} \sumin \ex{Z_i^2} \leqslant \nu$ such that, for all integers $q \geqslant 3$, $\frac{1}{n} \sum \ex{\abs{Z_i}^q} \leqslant \frac{q! c^{q-2}}{2} \nu.$
	Then, for all $\epsilon > 0$, we have $\pr\left(\abs{\frac{1}{n} \sumin (Z_i - \ex{Z_i})} \geqslant \epsilon\right) \leqslant 2 \exp \left(- \frac{n \epsilon^2}{2 (\nu + c \epsilon)}\right)$.
\end{AppThe}

\begin{proof}[Proof of Theorem \ref{the: A unb Bernstein}]
	See Corollary 2.11 in \citetsupp{boucheron2013concentration}.
\end{proof}

Specifically, we make use of the following corollary.

\begin{AppCor}
	\label{cor: bernstein}
	Let $Z_1, \dots, Z_n$ be mean zero independent random variables. Assume that there exists some $v > 0$ such that $\ex{e^{\lambda Z_i}} \leqslant e^{v \lambda^2}$ for all $\lambda \in \mathbb R$ and for all $i \in \{1, \dots, n\}$. Then, there exist some positive constants $C$, $a$, and $b$ such that for all constants $\alpha_1$, \dots, $\alpha_n$ satisfying $\max_{i} \abs{\alpha_i} < \overline \alpha$ and for all $\epsilon > 0$,
	\begin{align*}
		\pr \left(\abs{n^{-1} \sumin \alpha_i Z_i } \geqslant \epsilon \right) \leqslant C \exp\left(-\frac{n \epsilon^2}{a + b \epsilon}\right)
	\end{align*}
	and
	\begin{align*}
		\pr \left(\abs{n^{-1} \sumin \alpha_i \left(Z_i^2 - \ex{Z_i^2}\right) } \geqslant \epsilon \right) \leqslant C \exp\left(-\frac{n \epsilon^2}{a + b \epsilon}\right).
	\end{align*}
\end{AppCor}

\begin{proof}[Proof of Corollary \ref{cor: bernstein}]
	Follows from Lemma \ref{lem: A subG moments} and Theorem \ref{the: A unb Bernstein}.
\end{proof}

\begin{Rem}
	Note that the constants $C$, $a$, and $b$ depend on $v$ and $\overline \alpha$ only.
\end{Rem}

\section{Illustration of Assumption~\ref{ass: local smooth g}}
\label{sec: non-diff illustration}
Suppose $\xi$ is scalar and $g (\xi_i,\xi_k) = \kappa \abs{\xi_i - \xi_k}$. Note that, as a function of $\xi_i$, $g(\xi_i,\xi_k)$ is non-differentiable at $\xi_i = \xi_k$. When $\xi_i, \xi_j \leqslant \xi_k$, we have
\begin{align*}
g(\xi_i,\xi_k) - g(\xi_j,\xi_k) = \kappa (\xi_k - \xi_i) - \kappa (\xi_k - \xi_j) = - \kappa (\xi_i - \xi_j).
\end{align*}
If $\xi_i,\xi_j \geqslant \xi_k$,
\begin{align*}
g(\xi_i,\xi_k) - g(\xi_j,\xi_k) = \kappa (\xi_i - \xi_k) - \kappa (\xi_j - \xi_k) = \kappa (\xi_i - \xi_j).
\end{align*}
So, we can take
\begin{align*}
G (\xi_i,\xi_k) = \begin{cases}
  - \kappa, & \xi_i < \xi_k\\
  \kappa, & \xi_i \geqslant \xi_k
\end{cases}.
\end{align*}
Then, the remainder $r_g  (\xi_i,\xi_j,\xi_k) = 0$ when $\xi_i,\xi_j \leqslant \xi_k$ or $\xi_i,\xi_j \geqslant \xi_k$. However, if, for example, $\xi_i \leqslant \xi_k \leqslant \xi_j$,
\begin{align*}
g(\xi_i,\xi_k) - g(\xi_j,\xi_k) = \kappa( 2 \xi_k - \xi_i - \xi_j).
\end{align*}
Since $G(\xi_i,\xi_k) = - \kappa$,
\begin{align*}
g(\xi_i,\xi_k) - g(\xi_j,\xi_k) = - \kappa (\xi_i - \xi_j) + 2 \kappa (\xi_k - \xi_j) = G(\xi_i,\xi_k) (\xi_i - \xi_j) + r_g(\xi_i,\xi_j,\xi_k),
\end{align*}
so $r_g(\xi_i,\xi_j,\xi_k) = 2 \kappa (\xi_k - \xi_j)$ when $\xi_i \leqslant \xi_k \leqslant \xi_j$. Clearly, in this case, the linearization remainder is no longer $O(\abs{\xi_i - \xi_j}^2)$. Importantly, Assumption \ref{ass: local smooth g} allows for this possibility: the remainder $r_g$ is bounded by $C \delta_n^2$ only when $\abs{\xi_i - \xi_k} > \delta_n$ and $\abs{\xi_j - \xi_i} \leqslant \delta_n$. Under these restrictions, $\xi_k$ cannot lie between $\xi_i$ and $\xi_j$ implying that in this case $r_g  (\xi_i,\xi_j,\xi_k) = 0$, so the bound on $r_g(\xi_i,\xi_j,\xi_k)$ imposed by Assumption \ref{ass: local smooth g} is trivially satisfied.

\section{Imputation of $\hat Y_{ij}^*$ with Missing Data}
\label{sec: missing data imputation}
In this section, we discuss how one can adjust the previously constructed estimator $\hat Y_{ij}^*$ when some interaction outcomes are missing. We will stick with the notation previously introduced in Section~\ref{ssec: missing}.

We start with constructing $\hat d_{\infty}^2(i,j)$ previously defined in \eqref{eq: hat d inf def}. In the studied setting, $\hat d_{\infty}^2(i,j)$ can be computed as
\begin{align}
  \label{eq: d inf missing def}
  \hat d_{\infty}^2(i,j) \coloneqq \max_{k \neq i,j} \Big\vert{\abs{\mathcal O_{ijk}}^{-1} \sum_{\ell \in \mathcal O_{ijk}} (Y_{i \ell} - Y_{j \ell}) Y_{k \ell}}\Big\vert,
\end{align}
where $\mathcal O_{ijk} = \{\ell: D_{i \ell} = D_{j \ell} = D_{k \ell} = 1 \}$. In practice, if $\abs{\mathcal{O}_{ijk}}$ is too small for a given $k$, this $k$ can be dropped from the calculation of $\hat d_{\infty}^2(i,j)$ above; it is sufficient to compute the maximum in \eqref{eq: d inf missing def} over a subset of all the other agents so long as it is not too small (otherwise, we can set $\hat d_{\infty}^2(i,j) = \infty$).

Next, we proceed with calculation of $\hat Y_{ij}^*$. To this end, we need to introduce a $j$-specific neighborhood of $i$ denoted by $\hatN_{ij}(n_{ij})$. As before, for simplicity, we will stick with discrete $X$. First, let us define a pool of candidate agents (donors) which  can potentially be used to impute $\hat Y_{ij}^*$ denoted by $\mathcal P_{ij} \coloneqq \{i':X_{i'} = X_i, D_{i'j} = 1\}$, where the requirement $D_{i'j} = 1$ ensures that $Y_{i'j}$ is observed. Then $\hatN_{ij}(n_{ij})$ is constructed as a collection of $n_{ij}$ agents from $\mathcal P_{ij}$ closest to agent $i$ in terms of $\hat d_{\infty}^2(i,i')$, i.e.,
\begin{align*}
  \hatN_{ij}(n_{ij}) \coloneqq \{i' \in \mathcal P_{ij}: \text{Rank}(\hat d_{\infty}^2(i,i')|\mathcal P_{ij}) \leqslant n_i \},
\end{align*}
and $\hat Y_{ij}^*$ can be estimated as
\begin{align*}
  \hat Y_{ij}^* = \frac{\sum_{i' \in \hatN_{ij}(n_{ij})} Y_{i' j}}{n_{ij}}.
\end{align*}

When the observed matrix $Y$ is so sparse and $\mathcal P_{ij}$ is so limited in the first place to find enough decent matches for imputation of $Y_{ij}^*$, one can try to sequentially impute the elements of $Y^*$ with a sufficient number of candidate donors in $\mathcal P$ first, and then use the newly imputed $\hat Y_{i'j}^*$'s as observed outcomes $Y_{i'j}$ to expand the previously thin pool of potential donors $\mathcal P_{ij}$. One can apply this procedure sequentially if needed to impute the whole $Y^*$ or a part of it.

Once $\hat Y^*$ is constructed, we can estimate $\hat d_{ij}^2$ as in \eqref{eq: hat d star def}, or if some elements of $\hat Y^*$ are still missing, we can modify \eqref{eq: hat d star def} in the same fashion as in \eqref{eq: general q}, where the overlaps should be defined based on whether the imputed $\hat Y^*$'s are missing or not. Once $\hat d_{ij}^2$'s are obtained, we can calculate $\hat \beta$ as in \eqref{eq: general beta}, or even use the imputed outcomes instead of the plain outcomes in \eqref{eq: general beta}, including in the definition of the overlap $\mathcal O_{ij}$ again.

\bibliographystylesupp{apalike}
\bibliographysupp{library}

\end{document}